\documentclass[11pt]{article}
\usepackage{amsmath, amsthm, amssymb, url, graphicx, subcaption}
\usepackage[round]{natbib}
\usepackage[margin=1.25in]{geometry}
\usepackage{paralist}
\usepackage{xr} 

\usepackage{booktabs}

\newcommand{\tran}{\mathsf{T}}
\newcommand{\wh}{\widehat}
\newcommand{\wt}{\widetilde}
\newcommand{\diag}{\mathrm{diag}}
\renewcommand{\ge}{\geqslant}
\renewcommand{\le}{\leqslant}

\newcommand{\dnorm}{\mathcal{N}}
\newcommand{\sumdot}{\text{\tiny$\bullet$}}
\newcommand{\e}{\mathbb{E}}

\newcommand{\zij}{Z_{ij}}
\newcommand{\zijp}{Z_{ij'}}
\newcommand{\zis}{Z_{is}}
\newcommand{\zisp}{Z_{is'}}

\newcommand{\zipj}{Z_{i'j}}
\newcommand{\zipjp}{Z_{i'j'}}
\newcommand{\zips}{Z_{i's}}
\newcommand{\zipsp}{Z_{i's'}}

\newcommand{\zrj}{Z_{rj}}
\newcommand{\zrjp}{Z_{rj'}}
\newcommand{\zrs}{Z_{rs}}
\newcommand{\zrsp}{Z_{rs'}}

\newcommand{\zrpj}{Z_{r'j}}
\newcommand{\zrpjp}{Z_{r'j'}}
\newcommand{\zrps}{Z_{r's}}
\newcommand{\zrpsp}{Z_{r's'}}

\newcommand{\zzt}{(ZZ^\tran)}
\newcommand{\ztz}{(Z^\tran Z)}

\newcommand{\daiip}{\mathbb{D}_{A,ii'}}
\newcommand{\dbjjp}{\mathbb{D}_{B,jj'}}
\newcommand{\deijipjp}{\mathbb{D}_{E,ij,i'j'}}
\newcommand{\deijijp}{\mathbb{D}_{E,ij,ij'}}

\newcommand{\daii}{\mathbb{D}_{A,ii}}

\newcommand{\darrp}{\mathbb{D}_{A,rr'}}
\newcommand{\darr}{\mathbb{D}_{A,rr}}
\newcommand{\dbssp}{\mathbb{D}_{B,ss'}}
\newcommand{\dersrpsp}{\mathbb{D}_{E,rs,r's'}}
\newcommand{\dersrsp}{\mathbb{D}_{E,rs,rs'}}
\newcommand{\dersrps}{\mathbb{D}_{E,rs,r's}}

\newcommand{\baiiprrp}{\mathbb{B}_{A,ii',rr'}}
\newcommand{\baiirr}{\mathbb{B}_{A,ii,rr}}

\newcommand{\bbjjpssp}{\mathbb{B}_{B,jj',ss'}}

\newcommand{\beijipjprsrpsp}{\mathbb{B}_{E,iji'j',rsr's'}}
\newcommand{\beijijprsrsp}{\mathbb{B}_{E,ijij',rsrs'}}

\newcommand{\beijijprsrpsp}{\mathbb{B}_{E,ijij',rsr's'}}

\newcommand{\qaiiprrp}{\mathbb{Q}_{A,ii',rr'}}
\newcommand{\qaiirr}{\mathbb{Q}_{A,ii,rr}}

\newcommand{\qbjjpssp}{\mathbb{Q}_{B,jj',ss'}}

\newcommand{\qeijipjprsrpsp}{\mathbb{Q}_{E,iji'j',rsr's'}}
\newcommand{\qeijijprsrsp}{\mathbb{Q}_{E,ijij',rsrs'}}
\newcommand{\qeijijprsrps}{\mathbb{Q}_{E,ijij',rsr's}}
\newcommand{\qeijijprsrpsp}{\mathbb{Q}_{E,ijij',rsr's'}}

\newcommand{\nid}{N_{i\sumdot}}
\newcommand{\ndj}{N_{\sumdot j}}
\newcommand{\nipd}{N_{i'\sumdot}}
\newcommand{\ndjp}{N_{\sumdot j'}}

\newcommand{\nrd}{N_{r\sumdot}}
\newcommand{\nds}{N_{\sumdot s}}

\newcommand{\ai}{a_i}
\newcommand{\aip}{a_{i'}}
\newcommand{\ar}{a_r}
\newcommand{\arp}{a_{r'}}

\newcommand{\bj}{b_j}
\newcommand{\bjp}{b_{j'}}
\newcommand{\bs}{b_s}
\newcommand{\bsp}{b_{s'}}

\newcommand{\ssa}{\sigma^2_A}
\newcommand{\ssb}{\sigma^2_B}
\newcommand{\sse}{\sigma^2_E}
\newcommand{\ssy}{\sigma^2_Y}

\newcommand{\fpa}{\sigma^4_A}
\newcommand{\fpb}{\sigma^4_B}
\newcommand{\fpe}{\sigma^4_E}
\newcommand{\fpy}{\sigma^4_Y}

\newcommand{\ssah}{\hat\sigma^2_A}
\newcommand{\ssbh}{\hat\sigma^2_B}
\newcommand{\sseh}{\hat\sigma^2_E}

\newcommand{\fpah}{\hat\sigma^4_A}
\newcommand{\fpbh}{\hat\sigma^4_B}
\newcommand{\fpeh}{\hat\sigma^4_E}

\newcommand{\ka}{\kappa_A}
\newcommand{\kb}{\kappa_B}
\newcommand{\ke}{\kappa_E}
\newcommand{\ky}{\kappa_Y}

\newcommand{\oir}{1_{i=r}}
\newcommand{\oipr}{1_{i'=r}}
\newcommand{\oiprp}{1_{i'=r'}}
\newcommand{\oirp}{1_{i=r'}}

\newcommand{\ojs}{1_{j=s}}
\newcommand{\ojps}{1_{j'=s}}
\newcommand{\ojpsp}{1_{j'=s'}}
\newcommand{\ojsp}{1_{j=s'}}

\newcommand{\ossp}{1_{s=s'}}
\newcommand{\orrp}{1_{r=r'}}

\newcommand{\oiip}{1_{i=i'}}
\newcommand{\ojjp}{1_{j=j'}}

\newcommand{\oii}{1_{i=i}} 

\newcommand{\phe}{\phantom{=}\,\,\,} 
\newcommand{\phz}{\phantom{0}}

\newcommand{\yijh}{{\hat Y}_{ij}}

\newcommand{\yij}{Y_{ij}}
\newcommand{\yijp}{Y_{ij'}}
\newcommand{\yipj}{Y_{i'j}}
\newcommand{\yipjp}{Y_{i'j'}}

\newcommand{\yrs}{Y_{rs}}
\newcommand{\yrsp}{Y_{rs'}}
\newcommand{\yrps}{Y_{r's}}
\newcommand{\yrpsp}{Y_{r's'}}

\newcommand{\lij}{\lambda_{ij}}

\newcommand{\lis}{\lambda_{is}}

\newcommand{\lrj}{\lambda_{rj}}

\newcommand{\lrs}{\lambda_{rs}}
\newcommand{\lrsp}{\lambda_{rs'}}

\newcommand{\lrps}{\lambda_{r's}}
\newcommand{\lrpsp}{\lambda_{r's'}}

\newcommand{\byid}{\bar Y_{i\sumdot}}
\newcommand{\bydj}{\bar Y_{\sumdot j}}
\newcommand{\syid}{S_{i\sumdot}}
\newcommand{\sydj}{S_{\sumdot j}}

\newcommand{\bydd}{{\bar Y}_{\sumdot\sumdot}}

\newcommand{\eij}{e_{ij}}

\newcommand{\eijp}{e_{ij'}}
\newcommand{\eipjp}{e_{i'j'}}

\newcommand{\ers}{e_{rs}}

\newcommand{\erpsp}{e_{r's'}}

\newcommand{\simiid}{\stackrel{\mathrm{iid}}{\sim}}

\newcommand{\natu}{\mathbb{N}}

\newcommand{\ui}{u_i}

\newcommand{\vj}{v_j}

\newcommand{\wij}{w_{ij}}

\newcommand{\var}{\mathrm{Var}}
\newcommand{\cov}{\mathrm{Cov}}

\newcommand{\ts}{\textstyle}

\newcommand{\muaf}{\mu_{A,4}}
\newcommand{\mubf}{\mu_{B,4}}
\newcommand{\muef}{\mu_{E,4}}

\newcommand{\yid}{Y_{i\sumdot}}
\newcommand{\ydj}{Y_{\sumdot j}}
\newcommand{\sid}{S_{i\sumdot}}
\newcommand{\sdj}{S_{\sumdot j}}

\newcommand{\ydd}{Y_{\sumdot\sumdot}}
\newcommand{\sdd}{S_{\sumdot\sumdot}}

\newcommand{\zn}{\mathrm{ZN}}

\newcommand{\tid}{T_{i\sumdot}}
\newcommand{\tdj}{T_{\sumdot j}}

\newcommand{\kah}{\hat{\kappa}_A}
\newcommand{\kbh}{\hat{\kappa}_B}
\newcommand{\keh}{\hat{\kappa}_E}

\newcommand{\mfa}{\mu_{A,4}}
\newcommand{\mfah}{\hat{\mu}_{A,4}}
\newcommand{\mfb}{\mu_{B,4}}
\newcommand{\mfbh}{\hat{\mu}_{B,4}}
\newcommand{\mfe}{\mu_{E,4}}
\newcommand{\mfeh}{\hat{\mu}_{E,4}}

\newcommand{\um}{\underline{m}}
\newcommand{\om}{\overline{m}}

\newtheorem{thm}{Theorem}[section]
\newtheorem{lemma}{Lemma}[section]
\newtheorem{model}{Model}

\theoremstyle{definition}
\newtheorem{define}{Definition}[section]

\title{Efficient moment calculations for variance components 
in large unbalanced crossed random effects models}
\author{Katelyn Gao \\Stanford University 
\and 
Art B. Owen\\Stanford University}
\date{January 2016}

\begin{document}

\maketitle

\begin{abstract}
Large crossed data sets, described by generalized linear mixed models, have become increasingly common and provide challenges for statistical analysis. At very large sizes it becomes desirable to have the computational costs of estimation, inference and prediction (both space and time) grow at most linearly with sample size.

Both traditional maximum likelihood estimation and numerous Markov chain Monte Carlo Bayesian algorithms take superlinear time in order to obtain good parameter estimates. We propose moment based algorithms that, with at most linear cost, estimate variance components, measure the uncertainties of those estimates, and generate shrinkage based predictions for missing observations. When run on simulated normally distributed data, our algorithm performs competitively with maximum likelihood methods.
\end{abstract}

\section{Introduction}

Modern electronic activity generates enormous data sets with an unbalanced crossed random effects structure. The factors are customer IDs, URLs, product IDs, cookies, IP addresses, news stories, tweets, and query strings, among others. These variables could be treated as fixed effects, plain categorical variables that just happen to have a large number of levels. But in many cases, the specific category levels are evanescent. Customers turn over 
at some rate, cookies get deleted at an even faster rate, products or news stories grow in popularity but then fade. In such cases it is more realistic to treat such variables as random effects. We want our inferences to apply 
to the population from which the future and observed levels of those variables are sampled. Furthermore, for realism we should treat data in the same level of a factor as correlated.

The statistically efficient way to treat data sets with crossed random effects is through generalized linear mixed models (GLMMs), maximizing the likelihood with respect to both the parameters and the random effects. However, the cost of these computations is dominated by a Cholesky decomposition that takes time cubic in the number of distinct levels and space quadratic in that number; see \cite{B14} or \cite{R93}. Such costs are infeasible for big data. 

It has been suggested to us that stochastic gradient descent (SGD) could provide an alternative way to maximize the likelihood. However, SGD approaches have only been developed for data that can be split into independent subsets, which is not possible for data sets with crossed random effects.

With GLMMs infeasible, it is natural to consider the Gibbs sampler and other Markov Chain Monte Carlo (MCMC) methods. But, as shown in Section~\ref{sec:mcmc}, those methods in the crossed random effects context has computational cost that is superlinear in the sample size. This is very different from the great success that MCMC has on hierarchical models for data with a nested structure. See for instance \cite{GVHB12}, \cite{S14} and \cite{yu2011center}.

With both likelihood and Bayesian methods running into difficulties, we turn to the method of moments. It seems ironic to use a 19th century method in this era of increased computer power. But data growth has been outpacing processing power for single-threaded computation, so it is appropriate to revisit methods from an earlier time when the data was large compared to the available computing power. A compelling advantage of the method of moments is that it is easily parallelizable. It also makes very weak assumptions, has no tuning parameters, and does not require cumbersome diagnostics.

We are motivated by generalized linear mixed models with linear predictors but we focus the present paper on a very special case. We consider a setting with identity link, just two factors that are both random, and intercept only regression. In this paper, we assume that the data follows the model
\begin{model} 
Two-factor crossed random effects:
\begin{align}\label{eq:refmodel}
\begin{split}
&\yij=\mu+a_i+b_j+e_{ij},\quad i,j\in\natu\quad\text{where}\\
&a_i\simiid (0,\ssa),\quad b_j\simiid(0,\ssb),\quad e_{ij}\simiid(0,\sse)\quad\text{and} \\
&\e(a_i^4) < \infty,\quad \e(b_j^4) < \infty,\quad \e(e_{ij}^4) < \infty
\end{split}
\end{align}
\end{model}
In the available data we only see $N$ of the $\yij$, where $1\le N<\infty$, in $R$ distinct rows ($i$'s) and $C$ distinct columns ($j$'s). We assume that observations are missing completely at random. See Section~\ref{sec:missingness} for comments on informative missingness. Note that we do not make any distributional assumptions.

We choose this model because it is the simplest case that exhibits the intrinsic difficulty of the large unbalanced crossed random effects setting, even though it may not describe real-world data well. Our goal is not to resolve the issue of analyzing massive crossed data sets via GLMMs in one go. Instead, we consider a simple GLMM for crossed data and study parameter estimation in that model, which is still a challenging problem.

Let $\theta=(\ssa,\ssb,\sse)^\tran$ be the vector of variance components. Our first task is to get an unbiased estimate $\hat\theta$ of $\theta$ at computational cost $O(N)$ and using additional storage that is $O(R+C)$, which is often sublinear in $N$.

Our second and more challenging task is to find the variance of $\hat\theta$, $\var(\hat\theta \mid \theta,\kappa)$. This variance depends on both $\theta$ and the vector of kurtoses of the random effects $\kappa = (\ka,\kb,\ke)^\tran$. We develop formulas $V(\theta,\kappa)$ approximating $\var(\hat\theta\mid\theta,\kappa)$ that can be computed in $O(N)$ time and $O(R+C)$ storage, given values for $\theta$ and $\kappa$. After developing an estimate $\hat\kappa$ that can be computed in $O(N)$ time and $O(R+C)$ space, we let $\wh\var(\hat\theta)=V(\hat\theta,\hat\kappa)$ be our plug-in estimate of the variance of $\hat\theta$.

Notice that in order to achieve the complexity bounds, we choose to over-estimate $\var(\hat\theta)$. Specifically, we require the functions $V$ to satisfy $\diag(V(\theta,\kappa))\ge\diag(\var(\hat\theta\mid\theta,\kappa))$. There is a trade-off in selecting $V$ though; the less conservative it is, the more time needed to compute it. 

For large data sets we might suppose that $\var(\hat\theta)$ is necessarily very small and getting exact values is not important. While this may be true, it is wise to check. The effective sample size (as defined in \cite{L08}) in model~\eqref{eq:refmodel} might be as small as $R$ or $C$ if the row or column effects dominate. Moreover, if the sampling frequencies of rows or columns are very unequal, then the effective sample size can be much smaller than $R$ or $C$. For example, the Netflix data set \citep{benn:lann:2007} has $N\doteq 10^8$. But there are only about $18{,}000$ movies and so for statistics dominated by the movie effect the effective sample size might be closer to $18{,}000$. That the movies do not appear equally often would further reduce the effective sample size. Indeed, \cite{pbs} shows that for some linear statistics the variance could be as much as $50{,}000$ times larger than a formula based on IID sampling would yield. That factor is perhaps extreme but it would translate a nominal sample size of $10^8$ into an effective sample size closer to $2{,}000$.

An outline of this paper is as follows. Section~\ref{sec:mcmc} describes the difficulties with Gibbs sampling and other MCMC algorithms for crossed random effects, as suggested by theoretical results and shown through simulations. Section~\ref{sec:notation} introduces further notation and assumptions. Section~\ref{sec:thetahat}
presents our linear-cost algorithm to estimate $\theta$ and conservatively approximate the variance of that estimate. Section~\ref{sec:shrinkage} studies how knowledge of $\ssa$, $\ssb$, and $\sse$ can be used to construct shrinkage predictions of unknown $\yij$. Section~\ref{sec:experiments} illustrates the methods in Section~\ref{sec:thetahat} on both simulated Gaussian data and real world data. Section~\ref{sec:discussion} concludes the paper and discusses informative missingness. The appendix, Section~\ref{sec:appendix}, has a proof of convergence rates for MCMC methods and tables of their simulation results. A supplement, Sections~\ref{sec:supplfirst}--\ref{sec:prooflemsmoothing}, develops the variance formulas for our moment estimates and provides proofs of our theorems about prediction. We conclude this section with a few more pointers to the literature.

Our procedure to find variance component estimates are similar to those of \cite{H53} as described in \citet[Chapter 5]{SCM09}. Some differences are that we use $U$-statistics, and that we find variance component estimates and variances of those estimates in time and space $O(N)$. For one of Henderson's algorithms, even the point estimates require superlinear computation in inverting $R\times R$ or $C\times C$ matrices. Moreover, the majority of \cite{SCM09} considers Gaussian data which makes the kurtoses zero. Gaussian variables are not a reasonable assumption in our target applications and so we develop kurtosis estimates.

For crossed random effects models with missing data \cite{CR99} propose an alternating imputation-posterior (AIP) algorithm, which they show has good performance on fairly large data sets. It may be termed a `pseudo-MCMC' method since it alternates between sampling the missing data from its distribution given the parameter estimates and sampling the parameters from a distribution centered on the maximum likelihood estimates. Because of this last step, we do not consider AIP to be scalable to Internet size problems.

In our model~\eqref{eq:refmodel}, for simplicity the variance components are homoscedastic. Alternatively, we could allow them to be heteroscedastic; see \cite{pbs} or \cite{OE12}, who study bootstrap variance estimates for means and smooth functions of means. The latter paper also considers a more complex model in the sense that there are more than two factors as well as interactions among factors.

\section{MCMC for large crossed data}\label{sec:mcmc}

In this section we consider some common MCMC methods to estimate the parameters $\ssa$, $\ssb$, and $\sse$ of model \eqref{eq:refmodel}. For this section only, we assume that $a_i$, $b_j$ and $\eij$ are normally distributed.

Balanced data is a fully sampled $R\times C$ matrix with $\yij$ for rows $i=1,\dots,R$ and columns $j=1,\dots,C$. We present some analyses for the balanced case with interspersed remarks on how the general unbalanced case behaves. The balanced case allows sharp formulas that we find useful and that case is the one we simulate. In particular, we can obtain convergence rates for some MCMC algorithms. 

To estimate $\ssa$, $\ssb$, and $\sse$ we sample from the posterior distribution given the data: \\ $\pi=p(\mu,a,b,\ssa,\ssb,\sse \mid Y)$ where $a$ is the vector of $a_i$ and $b$ is the vector of $b_j$. Let 
\[S^{(t)}=\bigl(\begin{matrix}
\mu^{(t)}&  a^{(t)^\tran}&  b^{(t)^\tran}&  \sigma_A^{2(t)}& \sigma_B^{2(t)}&  \sigma_E^{2(t)} 
\end{matrix}\bigr)^\tran,\quad \text{for $t\ge1$}\]
denote the resulting chain. While MCMC is effective for hierarchical random effects models, it scales badly for crossed random effects models as we see here. In limits where $R,C \to \infty$, the dimension of our chain $S^{(t)}$ approaches infinity. Convergence rates of many MCMC methods slow down as the dimension of the chain increases, making them ineffective for high dimensional parameter spaces.

The MCMC methods we consider go over the entire data set at each iteration. There are alternative samplers that save computation time by only looking at subsets of data at each iteration. However, so far those approaches are developed for IID data and not the crossed random effects setting.

\subsection{Gibbs sampling}\label{sec:gibbssampling}

In each iteration of Gibbs sampling \citep{GG84}, we draw from the conditional posteriors of $\mu$, $a$, $b$, $\ssa$, $\ssb$, and $\sse$ in turn. For elucidation, let us consider the problem of Gibbs sampling from the `smaller' distribution $\phi=p(a,b \mid \mu,\ssa,\ssb,\sse, Y)$. At iteration $t+1$, we sample $a^{(t+1)}\sim p(a \mid b^{(t)},\mu,\ssa,\ssb,\sse,Y)$ and $b^{(t+1)}\sim p(b \mid a^{(t+1)},\mu,\ssa,\ssb,\sse,Y)$, which are normal distributions with diagonal covariance matrices. Let $X^{(t)}$ be the resulting chain.

\cite{RS97} give the following definition.
\begin{define}\label{def:mcconvergence}
Let $\theta^{(t)}$, for integer $t\ge0$  be a Markov chain with stationary distribution $h$. 
Its convergence rate is the minimum number $\rho$ such that 
\[\lim_{t \to \infty} \e_h\bigl((\e_h(f(\theta^{(t)}) \mid \theta^{(0)})-\e_h(f(\theta)))^2\bigr)r^{-t}=0\]
holds for all measurable functions $f$ such that $\e_h(f(\theta)^2)<\infty$ and all $r>\rho$.
\end{define}

\begin{thm}\label{thm:gibbsrate} 
Let $\rho$ be the convergence rate of $X^{(t)}$ to $\phi$, as in Definition~\ref{def:mcconvergence}. Then,
$$\rho=\dfrac{\ssb}{\ssb+\sse/R}\times\dfrac{\ssa}{\ssa+\sse/C}.$$
\end{thm}
\begin{proof}
See Section~\ref{sec:proofthmgibbsrate}.
\end{proof}

We see that $\rho \to 1$ as $R,C \to \infty$, outside of trivial cases with $\ssa$ or $\ssb$ equal to zero. If $R$ and $C$ grow proportionately then $\rho=1-\alpha/\sqrt{N}+O(1/N)$ for some $\alpha>0$. We can therefore expect the Gibbs sampler to require at least some constant multiple of $\sqrt{N}$ iterations to approximate the target distribution sufficiently. When the data are not perfectly balanced numerical computation of $\rho$ shows that Gibbs still mixes increasingly slowly as $N \to \infty$. But in that case, the sampler requires $O(N)$ computation per iteration. In sum, Gibbs takes $O(N^{3/2})$ work to sample from $\phi$, which is not scalable.

Because sampling from $\phi$ can be viewed as a subproblem of sampling from $\pi$, we believe that the Gibbs sampler that draws from $\pi$, which also requires $O(N)$ time per iteration, will exhibit the same slow convergence and hence require superlinear computation time.

\subsection{Other MCMC algorithms}

The Gibbs sampler is widely used for problems like this, where the full conditional distributions are tractable. But there are other MCMC algorithms that one could use. Here we consider random walk Metropolis (RWM), Langevin diffusion, and Metropolis adjusted Langevin (MALA). They also have difficulties scaling to large data sets.

At iteration $t+1$ of RWM, a Gaussian random walk proposal $S^{(t+1)}\sim \dnorm(S^{(t)},\sigma^2 I)$ for $\sigma^2>0$ is made and the step is taken with the Metropolis-Hastings acceptance probability. If the target distribution is a product distribution of dimension $d$, the chain $\tilde{S}^{(t)}\equiv S^{(dt)}$ (i.e. the chain formed by every $d$th state of the chain $S^{(t)}$) converges to a diffusion whose solution is the target distribution. We may interpret this as a convergence time for the algorithm that grows as $O(d)$ \citep{RR01}. 

For our problem, evaluating the acceptance probability requires time at least $O(N)$, so the overall algorithm then takes $O(N(R+C))$ time. This is at best $O(N^{3/2})$, as we found for Gibbs sampling, and could be worse for sparse data where $N\ll RC$. Our target distribution is not of product form, and we have no reason to expect that RWM mixes orders of magnitude faster here than for a distribution of product form. Indeed, it seems more likely that mixing would be faster for product distributions than for distributions with more complicated dependence patterns such as ours.

At iteration $t+1$, Langevin diffusion steps $S^{(t+1)}\sim \dnorm(S^{(t)}+(h/2)\nabla \log{\pi}(S^{(t)}),hI)$ for $h>0$. As $h \to 0$, the stationary distribution for this process converges to $\pi$, as shown for general target distributions in \citep{L04}. Because $h \neq 0$ in practice, the Langevin algorithm is biased. To correct this, the MALA algorithm uses the Metropolis-Hastings algorithm with the Langevin proposal $S^{(t+1)}$. When the target distribution is a product distribution of dimension $d$, the chain $\tilde{S}^{(t)}\equiv S^{(d^{1/3}t)}$ converges to a diffusion with solution $\pi$; the convergence time grows as $O(d^{1/3})$ \citep{RR01}. With similar reasoning as for RWM, the computation time is $O(N(R+C)^{1/3})$, which is at best $O(N^{1+1/6})$.

\subsection{Simulation results}

We carried out simulations of the four algorithms described above, as well as five others: the block Gibbs sampler (`Block'), the reparameterized Gibbs sampler (`Reparam.'), the independence sampler (`Indp.'), RWM with subsampling (`RWM Sub.'), and the pCN algorithm of \cite{HSV14}. Descriptions of these five algorithms are given below with discussions of their simulation results. Every algorithm was implemented in MATLAB and run on a cluster using 4GB memory.

For each algorithm and a range of values of $R$ and $C$, we generated balanced data from model \eqref{eq:refmodel} with $\mu=1$, $\ssa=2$, $\ssb=0.5$, and $\sse=1$. We ran $20{,}000$ iterations of the algorithm, retaining the last $10{,}000$ for analysis. We record the CPU time required, the median values of $\mu$, $\ssa$, $\ssb$, and $\sse$, and the number of lags needed for their sample auto-correlation functions (ACF) to go below $0.5$.

The entire process is repeated in $10$ independent runs. Table~\ref{tab:mcmcsummary} presents median values of the recorded statistics over the $10$ runs for the case $R=C=1000$. Tables~\ref{tab:cpu} through~\ref{tab:ssesim} of the appendix collect corresponding results at a range of $(R,C)$ sizes. 

\begin{table}[t]
\small 
\centering 
\resizebox{\textwidth}{!}{\begin{tabular}{lccccccccc}
\toprule
Method    & Gibbs & Block & Reparam. & Lang. & MALA & Indp. & RWM  & RWM Sub. & pCN  \\ 
\toprule
CPU sec.      & 3432  & 15046 & 4099     &  2302 & 4760  &  2513 & 2141 &  2635    & 1966 \\ 
\midrule 
med $\mu$  & 0.97  & 1.02  & 1.04     &  \phz0.99 & \phz0.96 &  2.39 & 1.55 &   1.07   & 1.53 \\
med $\ssa$ & 1.96  & 1.99  &  2.02    &  \phz1.90 & \phz1.95 &  1.78 & 2.01 &  1.96    & 1.99 \\
med $\ssb$  &  0.51 & 0.50  & 0.50     &  \phz0.40 & \phz0.50 &  2.94 & 0.51 &   0.50   & 0.49 \\
med $\sse$   & 1.00  & 1.00  & 1.00     &  65.22 & 2.66  &  0.15 &    0 &  0.93    &    0 \\
\midrule
ACF$(\mu)$ & 801   & 790   & 694      &     1 &    2501 & 5000+ & 1133 &   1656   & 1008 \\ 
ACF$(\ssa)$ & 1     & 1     &     1    &   122 &  2656 & 5000+ & 1133 &   \phz989    & \phz912  \\ 
ACF $(\ssb)$  & 1     &   1   &    1     &   477 &  2514 & 5000+ & 1133 &    \phz855   & \phz556  \\ 
ACV$(\sse)$  & 1     &   1   &    1     &    385 &  3062   & 5000+ & 1518 &  1724    &  \phz621 \\
\bottomrule
\end{tabular}}
\caption{\label{tab:mcmcsummary}
Summary of simulation results for cases with $R=C=1000$.
The first row gives CPU time in seconds.  The next four rows give median
estimates of the $4$ parameters. The next four rows give the number of lags
required to get an autocorrelation below $0.5$.
}
\end{table}

Block Gibbs, which updates $a$ and $b$ together to try to improve mixing, has computation time superlinear in the number of observations. Also to improve mixing, reparameterized Gibbs scales the random effects to have equal variance. This gives an algorithm equivalent to the conditional augmentation of \cite{VM01}. For all three Gibbs-type algorithms, the parameter estimates are good but $\mu$ mixes slower as $R$ and $C$ increase, while the variance components do not exhibit this behavior.

The computation times of Langevin diffusion (`Lang.') and MALA are approximately linear in the number of observations. However, $\sse$ tends to explode for large data sets in Langevin diffusion, while the chain does not mix well in MALA.

The independent sampler is a Metropolis-Hastings algorithm where the proposal distribution is fixed. We propose $\mu \sim \dnorm(1,1)$, $a=\dnorm(0,I_R)$, $b=\dnorm(0,I_C)$, and $\ssa,\ssb,\sse \sim \mathrm{InvGamma}(1,1)$. The computation time grows linearly with the data size. The parameters do not mix well, and their estimates are not good. It is possible that better results would be obtained from a different proposal distribution, but it is not clear how best to choose one in practice.

RWM and RWM with subsampling, the latter of which updates a subset of parameters at each iteration, both have computation time linear in the number of observations. Neither algorithm mixed well, and for RWM $\sse$ tended to go to zero in large data sets.

The pCN algorithm is Metropolis-Hastings where the proposals are Gaussian random walk steps shrunken towards zero: $S^{(t+1)} \sim \dnorm(\sqrt{1-\sigma^2}S^{(t)},\sigma^2 I)$, for $\sigma^2 \leq 1$. \cite{HSV14} show that under certain conditions on the target distribution, the convergence rate of this algorithm does not slow with the dimension of the distribution. We include it here, even though our $\pi$ does not satisfy those conditions. The computation time grows linearly with the data size. However, the estimates for $\mu$ and $\sse$ are not good, and those for $\sse$ even get worse as the data size increases. None of the parameters seem to mix well.

In summary, for large data sets each algorithm mixes increasingly slowly or returns flawed estimates of $\mu$ and the variance components. We have also simulated some unbalanced data sets and slow mixing is once again the norm, with worse performance as $R$ and $C$ grow.

\section{Further notation and assumptions}\label{sec:notation}

In this section, we go over pertinent notation and assumptions about the pattern of observations. Our data are realizations from model \eqref{eq:refmodel}.

We refer to the first index of $\yij$ as the `row' and the second as the `column'. We use integers $i,i',r,r'$ to index rows and $j,j',s,s'$ for columns. The actual indices may be URLs, customer IDs, or query strings and are not necessarily the integers we use here. 

The variable $\zij$ takes the value $1$ if $\yij$ is observed and $0$ otherwise. We assume that there can be at most one observation in position $(i,j)$.

The sample size is $N=\sum_{ij}\zij<\infty$. The number of observations in row $i$ is $\nid=\sum_j\zij$ and the number in column $j$ is $\ndj=\sum_i\zij$. The number of distinct rows is $R=\sum_i1_{\nid>0}$ and there are $C=\sum_j1_{\ndj>0}$ distinct columns. In the following, all of our sums over rows are only over rows $i$ with $\nid>0$, and similarly for sums over columns. We state this because there are a small number of expressions where omitting rows without data changes their values. This convention corresponds to what happens when one makes a pass through the whole data set.

Let $Z$ be the matrix containing $\zij$.  Of interest are $\zzt_{ii'} = \sum_j\zij\zipj$, the number of columns for which we have data in both rows $i$ and $i'$, and $\ztz_{jj'}$. Note that $\zzt_{ii'}\le\nid$ and furthermore
$$
\sum_{ir}\zzt_{ir} = \sum_{jir}\zij\zrj=\sum_j\ndj^2,\quad\text{and}\quad 
\sum_{js}\ztz_{js}=\sum_i\nid^2. 
$$

Two other useful idioms are
\begin{align}\label{eq:totdegrees}
\tid = \sum_j\zij\ndj\quad\text{and}\quad \tdj=\sum_i\zij\nid.
\end{align}
$\tid$ is the total number of observations in all of the columns $j$ that are represented in row $i$.

Our notation allows for an arbitrary pattern of observations. Some special cases are as follows. A balanced crossed design can be described via $\zij = 1_{i\le R}1_{j\le C}$. If $\max_i\nid=1$ but $\max_j\ndj>1$ then the data have a nested structure with rows nested in columns. If $\max_i\nid=\max_j\ndj=1$, then the observed $\yij$ are IID.

Some patterns are difficult to handle. For example, if all the observations are in the same row or column, some of the variance components are not identifiable. We are motivated by problems that are not such worst cases.

The quantities 
\begin{align}\label{eq:assum:rowcoleps}
\epsilon_R = \max_i\nid/N,\quad\text{and}\quad\epsilon_C=\max_j\ndj/N 
\end{align}
measure the extent to which a single row or column dominates the data set. We expect that these are both small and in limiting arguments, where $N \to\infty$, we may assume that
\begin{align}\label{eq:assum:nohogs}
\max(\epsilon_R,\epsilon_C) \to 0. 
\end{align}
It is also often reasonable to suppose that $\max_i\tid/N$ and $\max_j\tdj/N$ are both small.

In many data sets, the average row and column sizes are large, but much smaller than $N$. One way to measure the average row size is $N/R$. Another way to measure it is to randomly choose an observation and inspect its row size, obtaining an expected value of $(1/N)\sum_i\nid^2$. Similar formulas hold for the average column size. Therefore, we assume that as $N \to\infty$
\begin{align}\label{eq:assum:bigavg}
\max(R/N,C/N) \to 0  
\end{align} 
and
\begin{align}\label{eq:assum:bigavgtwo}
\begin{split}
&\min\Bigl(\frac1N\sum_i\nid^2,\frac1N\sum_j\ndj^2\Bigr) \to \infty,\quad\text{and}\\
&\max\Bigl(\frac1{N^2}\sum_i\nid^2,\frac1{N^2}\sum_j\ndj^2\Bigr) \to 0.
\end{split}
\end{align} 

Notice that
\begin{align}\label{eq:freebies}
\frac1{N^2}\sum_i\nid^2\le \frac{1}{N^2}\sum_i\nid(\epsilon_R N)\le \epsilon_R,\quad\text{and}\quad
\frac1{N^2}\sum_j\ndj^2\le\epsilon_C
\end{align}
and so the second part of~\eqref{eq:assum:bigavgtwo} merely follows from \eqref{eq:assum:rowcoleps} and \eqref{eq:assum:nohogs}.

While the average row count may be large, many of the rows corresponding to newly seen entities can have $\nid=1$. In our analysis, it is not necessary to assume that all of the rows and columns contain at least some minimum number of observations. Thus, we avoid losing information by the practice of iteratively removing all rows and columns with few observations.

As a demonstration of the validity of our assumptions, the Netflix data has $N=100{,}480{,}507$ ratings on $R=17{,}770$ movies by $C=480{,}189$ customers. Therefore $R/N\doteq0.00018$ and $C/N\doteq 0.0047$. It is sparse with $N/(RC)\doteq 0.012$. It is not dominated by a single row or column because $\epsilon_R\doteq 0.0023$ and $\epsilon_R = 0.00018$ even though one customer has rated an astonishing $17{,}653$ movies. Similarly 
\begin{align*}
\frac{N}{\sum_i\nid^2} & \doteq 1.78\times10^{-5}, && \frac{\sum_j\ndj^2}{N^2} \doteq0.00056,\\
\frac{N}{\sum_j\ndj^2} & \doteq 0.0015,\quad\text{and}&& \frac{\sum_j\ndj^2}{N^2} \doteq 6.43\times 10^{-6}
\end{align*}
so that the average row or column has size $\gg1$ and $\ll N$. 

There are various possible data storage models. We consider the log-file model with a collection of $(i,j,\yij)$ triples, which for the purposes of this paper we assume are stored at the same location. A pass over the data proceeds via an iteration over all $(i,j,\yij)$ triples in the data set. Such a pass may generate intermediate values that we assume can be retained for further computations.

\section{Moment estimates of variance components}\label{sec:thetahat}

Here we develop a method of moments estimate $\hat\theta$ for $\theta = (\ssa,\ssb,\sse)^\tran$ that requires one pass over the data. We also find an expression for $\var(\hat\theta\mid\theta,\kappa)$ and describe how to obtain an approximation of it after a second pass over the data.

Naturally, we would also want to estimate $\mu$, and there are a number of ways to do so. The simplest is to let $\hat{\mu}=\bydd$, the sample mean. From \cite{OE12}, 
\begin{align}\label{eq:varbydd}
\var(\bydd)&=\ssa\dfrac{\sum_r \nrd^2}{N^2}+\ssb\dfrac{\sum_s \nds^2}{N^2}+\dfrac{\sse}{N}
\le\epsilon_R\ssa+\epsilon_C\ssb+\dfrac{\sse}{N}.
\end{align}
The upper bound in~\eqref{eq:varbydd} is tight for balanced data, but otherwise it can be very conservative. We anticipate that $1\gg\epsilon_R,\epsilon_C\gg1/N$ holds for our motivating applications as it did in the examples of \cite{OE12}. The properties of this estimator has been well-studied in the literature, so in this paper we focus on estimating the variance components. 

\subsection{$U$-statistics for variance components}\label{sec:ustatforvarcomp}

We use $U$-statistics in our method of moments estimators. The usual unbiased sample variance estimate can be formulated as a $U$-statistic, which is more convenient to analyze. We use the following U-statistics:

\begin{align}\label{eq:ourustats}
\begin{split}
U_a &= \dfrac{1}{2} \sum_{ijj'}\nid^{-1}\zij \zijp(Y_{ij}-Y_{ij'})^2, \\
U_b &= \dfrac{1}{2} \sum_{jii'}\ndj^{-1}\zij \zipj (Y_{ij}-Y_{i'j})^2,\quad\text{and}\\
U_e &= \dfrac{1}{2} \sum_{iji'j'}\zij \zipjp (Y_{ij}-Y_{i'j'})^2.
\end{split}
\end{align}
To understand $U_a$ note that for each row $i$, the quantities $\yij-\mu-a_i$ are IID with variance $\ssb+\sse$. Thus, $U_a$ is a weighted sum of within-row unbiased estimates of $\ssb+\sse$. The explanation for $U_b$ is similar, while $U_e$ is a proportional to the sample variance estimate of all $N$ observations.

\begin{lemma}\label{lem:eu:main}
Let $\yij$ follow the two-factor crossed random effects model~\eqref{eq:refmodel} with 
the observation pattern $\zij$ as described in Section~\ref{sec:notation}. 
Then the $U$-statistics defined at~\eqref{eq:ourustats} satisfy
\begin{align*}
\e(U_a) &= (\ssb+\sse)(N-R) \\
\e(U_b) &= (\ssa+\sse)(N-C) ,\quad\text{and}\\
\e(U_e) &= \ssa(N^2-\sum_i \nid^2)+\ssb(N^2-\sum_j \ndj^2)+\sse(N^2-N).
\end{align*}
\end{lemma}
\begin{proof}
See Section~\ref{sec:proof:lem:eu} of the supplement.
\end{proof}

To obtain unbiased estimates $\ssah$, $\ssbh$, and $\sseh$ given values of the $U$-statistics, we solve the $3 \times 3$ system of equations
\begin{align}\label{eq:momvar}
M\begin{pmatrix}\ssah \\ \ssbh \\ \sseh \end{pmatrix} &= \begin{pmatrix}U_a \\ U_b \\ U_e \end{pmatrix},\quad
\text{for}\quad 
M=\begin{pmatrix}0 & N-R & N-R \\ N-C & 0 & N-C \\ N^2-\sum_i \nid^2 & N^2-\sum_j \ndj^2 & N^2-N\end{pmatrix}
\end{align}

For our method to return unique and meaningful estimates, the determinant of $M$ 
\begin{align*}
\det{M}&=(N-R)(N-C) \Bigl(N^2-\sum_i \nid^2-\sum_j \ndj^2+N\Bigr) \\
&\ge (N-R)(N-C)(N^2(1-\epsilon_R-\epsilon_C)+N)
\end{align*} 
must be nonzero. This is true when no row or column has more than half of the data, and at least one row and at least one column has more than one observation.

To compute the $U$-statistics, notice that $U_a=\sum_i \syid$, where $\syid=\sum_j\zij(\yij-\byid)^2$ and $\byid=(1/\nid)\sum_j\zij\yij$. In one pass over the data and time $O(N)$, we compute $\nid$, $\byid$, and $\syid$ for all $R$ observed levels of $i$ using the incremental algorithm described in the next paragraph. We can also compute $N$, $R$ and $C$ in such a pass if they are not known beforehand.

\cite{chan1983algorithms} show how to compute both $\yid=\nid\byid$ and $\syid$ in a numerically stable one pass algorithm. At the initial appearance of an observation in row $i$, with corresponding column $j=j(1)$, set $\nid=1$, $\yid=\yij$ and $\sid=0$. After that, at the $k$th appearance of an observation in row $i$, with corresponding column $j(k)$,  
\begin{align}\label{eq:changolubleveque}
\nid &\gets \nid+1,
&&\yid \gets \yid + Y_{ij(k)},\quad\text{and}
&&\sid \gets \sid + \frac{(k\times Y_{ij(k)}-\yid)^2}{k(k-1)}.
\end{align}
\cite{chan1983algorithms} give a detailed analysis of roundoff error for update~\eqref{eq:changolubleveque} as well as generalizations that update higher moments from groups of data values.

In that same pass over the data, $U_e$ and the analogous quantities needed to compute $U_b$ ($\sydj$, $\bydj$, $\ndj$) are also computed using the incremental algorithm. Finally, in additional time $O(R+C)$, we calculate $\sum_i \syid$, $\sum_j \sydj$, $\sum_i \nid^2$, and $\sum_j \ndj^2$. Now, we have $U_a$, $U_b$, $U_e$, and all the entries of $M$.

Given $U_a$, $U_b$, $U_e$, and $M$ we can calculate $\ssah$, $\ssbh$, and $\sseh$ in constant time. Therefore, finding our method of moments estimators takes $O(N)$ time overall.

\subsection{Variances of the estimators}

In this section we present how to estimate the covariance matrix of $\hat\theta = (\ssah,\ssbh,\sseh)^\tran$. 

\subsubsection{True variance of $\hat\theta$}\label{sec:varU}

This section discusses the finite sample covariance matrix of $\hat\theta$. Theorem~\ref{thm:vu} below gives the exact variances and covariances of our $U$-statistics.
\begin{thm}\label{thm:vu}
Let $\yij$ follow the random effects model~\eqref{eq:refmodel} with 
the observation pattern $\zij$ as described in Section~\ref{sec:notation}. 
Then the $U$-statistics defined at~\eqref{eq:ourustats} have variances
\begin{align}\label{eq:varuathm}
\begin{split}
\var(U_a)&=\fpb(\kb+2)\sum_{ir}(ZZ^\tran)_{ir}(1-\nid^{-1})(1-\nrd^{-1})\\
&\phe+2\fpb\sum_{ir}\nid^{-1}\nrd^{-1}(ZZ^\tran)_{ir}((ZZ^\tran)_{ir}-1)) +4\ssb\sse(N-R)\\
&\phe+\fpe(\ke+2)\sum_i \nid(1-\nid^{-1})^2+2\fpe\sum_i (1-\nid^{-1}),
\end{split}
\end{align} 
and
\begin{align}\label{eq:varubthm}
\begin{split}
\var(U_b)&=\fpa(\ka+2)\sum_{js}(Z^\tran Z)_{js}(1-\ndj^{-1})(1-\nds^{-1})\\
&\phe+2\fpa\sum_{js}\ndj^{-1}\nds^{-1}(Z^\tran Z)_{js}((Z^\tran Z)_{js}-1)) +4\ssa\sse(N-C) \\
&\phe+\fpe(\ke+2)\sum_j \ndj(1-\ndj^{-1})^2+2\fpe\sum_j (1-\ndj^{-1}),
\end{split}
\end{align}
and $\var(U_e)$ equals
\begin{align}\label{eq:varuethm}
\begin{split}
&2\fpa((\sum_i \nid^2)^2-\sum_i \nid^4)+\fpa(\ka+2)(N^2\sum_i \nid^2-2N\sum_i \nid^3+\sum_i \nid^4) \\
&\phe+2\fpb((\sum_j \ndj^2)^2-\sum_j \ndj^4)+\fpb(\kb+2)(N^2\sum_j \ndj^2-2N\sum_j \ndj^3+\sum_j \ndj^4) \\
&\phe+2\fpe N(N-1)+\fpe (\ke+2)N(N-1)^2\\
&\phe+4\ssa\ssb (N^3-2N\sum_{ij}\zij\nid\ndj+\sum_{ij}\nid^2\ndj^2) \\
&\phe+4\ssa\sse (N^3-N\sum_i \nid^2)+ 4\ssb\sse (N^3-N\sum_j \ndj^2).
\end{split}
\end{align}
Their covariances are
\begin{align}
\cov(U_a,U_b)&=\fpe(\ke+2)\sum_{ij}\zij(1-\nid^{-1})(1-\ndj^{-1}),\label{eq:covuaubthm}\\
\cov(U_a,U_e) &= 
2\fpb\Bigl(\sum_i\nid^{-1}\tid^2-\sum_{ij}\zij\nid^{-1}\ndj^2\Bigr)\notag\\
&\phe+\fpb(\kb+2)\sum_{ij}\zij(N-\ndj)\ndj(1-\nid^{-1}) \label{eq:covuauethm}\\
&\phe+2\fpe (N-R)+\fpe(\ke+2)(N-R)(N-1) \notag\\
&\phe+4\ssb\sse N(N-R),\quad\text{and}\notag\\
\cov(U_b,U_e) 
&= 2\fpa\Bigl(\sum_j\ndj^{-1}\tdj^2-\sum_{ij}\zij\ndj^{-1}\nid^2\Bigr) \notag\\
&\phe+\fpa(\ka+2)\sum_{ij}\zij(N-\nid)\nid(1-\ndj^{-1}) \label{eq:covubuethm}\\
&\phe+2\fpe (N-C)+\fpe(\ke+2)(N-C)(N-1)\notag\\
&\phe+4\ssa\sse N(N-C). \notag
\end{align}
\end{thm}
\begin{proof}
Equation~\eqref{eq:varuathm} is proved in Section~\ref{sec:varua} of the supplement
and then equation~\eqref{eq:varubthm} follows by exchanging indices.
Equation~\eqref{eq:varuethm} is proved in Section~\ref{sec:varue} of the supplement.
Equation~\eqref{eq:covuaubthm} is proved in Section~\ref{sec:covuaub} of the supplement.
Equation~\eqref{eq:covuauethm} is proved in Section~\ref{sec:covuaue} of the supplement
and then equation~\eqref{eq:covubuethm} follows by exchanging indices.
\end{proof}

Now we consider $\var(\hat\theta)$. From \eqref{eq:momvar}
\begin{align}\label{eq:varthetahat}
\var(\hat\theta) & = M^{-1}\var\begin{pmatrix}U_a \\ U_b \\ U_e \end{pmatrix}(M^{-1})^\tran.
\end{align}

We show in Section~\ref{sec:vhatuapprox} that while $\var(U_e)$ and the covariances of the $U$-statistics may be exactly computed in time $O(N)$, $\var(U_a)$ and $\var(U_b)$ cannot. Therefore, we approximate $\var(U_a)$ and $\var(U_b)$ such that when we apply formula~\eqref{eq:varthetahat} we get conservative estimates of $\var(\ssah)$, $\var(\ssbh)$, and $\var(\sseh)$ (the values of primary interest). 

For intuition on what sort of approximation is needed, we give a linear expansion of $\var(\hat\theta)$ in terms of the variances and covariances of the $U$-statistics. Letting $\epsilon =\max(\epsilon_R,\epsilon_C,R/N,C/N)$ we have that as $\epsilon\to 0$
$$ 
M = \begin{pmatrix} N\\ & N\\ &&N^2\end{pmatrix}
\begin{pmatrix} 
0 & 1 & 1 \\ 
1 & 0 & 1 \\ 
1 & 1 & 1 \\
\end{pmatrix}(1+O(\epsilon))$$ 
and so
$$
M^{-1} 
= 
\begin{pmatrix}
-1 & 0 & 1 \\
0 & -1 & 1 \\
1 & 1 & -1 \\
\end{pmatrix}
\begin{pmatrix} N^{-1}\\ & N^{-1}\\ &&N^{-2}\end{pmatrix}
(1+O(\epsilon)) .
$$
It follows that
\begin{align}\label{eq:fromutoss}
\begin{split}
\ssah & = (U_e/N^2-U_a/N)(1+O(\epsilon)),\\
\ssbh & = (U_e/N^2-U_b/N)(1+O(\epsilon)),\quad\text{and}\\
\sseh & = (U_a/N + U_b/N - U_e/N^2)(1+O(\epsilon)). 
\end{split}
\end{align}
Disregarding the $O(\epsilon)$ terms,
\begin{align}\label{eq:approxvars}
\begin{split}
\var(\ssah) & \doteq \var(U_e)/N^4+\var(U_a)/N^2-2\cov(U_a,U_e)/N^3,\\
\var(\ssbh) & \doteq \var(U_e)/N^4+\var(U_b)/N^2-2\cov(U_b,U_e)/N^3,\quad\text{and}\\
\var(\sseh) & \doteq \var(U_a)/N^2 + \var(U_b)/N^2+\var(U_e)/N^4\\
&\phe -2\cov(U_a,U_e)/N^3 -2\cov(U_b,U_e)/N^3+2\cov(U_a,U_b)/N^2.  
\end{split}
\end{align}  

In light of equation~\eqref{eq:approxvars}, to find computationally attractive but conservative approximations of $\var(\hat\theta)$ in finite samples, we use over-estimates of $\var(U_a)$ and $\var(U_b)$. We discuss how to do so in Section~\ref{sec:vhatuapprox}.

In practice, when obtaining $\wh\var(\hat\theta)$, unless we are in the asymptotic situation described in Section 
\ref{sec:asympvthetahat}, we plug in $\ssah$, $\ssbh$, $\sseh$, and estimates of the kurtoses into the covariance matrix of the $U$-statistics where $\var(U_a)$ and $\var(U_b)$ have been replaced by their over-estimates. Then we apply equation~\eqref{eq:varthetahat}. We discuss estimating the kurtoses in Section~\ref{sec:kurthat}.

\subsubsection{Computable approximations of $\var(U)$}\label{sec:vhatuapprox}

First, we show how to obtain over-estimates of $\var(U_a)$ in time $O(N)$; the case of $\var(U_b)$ is similar. In addition to $N-R$, $\var(U_a)$ contains the following quantities
\begin{align*}
&\sum_{ir}\zzt_{ir}(1-\nid^{-1}) (1-\nrd^{-1})
&&\sum_{ir}\nid^{-1}\nrd^{-1}\zzt_{ir}(\zzt_{ir}-1)\\
&\sum_i\nid(1-\nid^{-1})^2,\quad\text{and}
&&\sum_i(1-\nid^{-1}).
\end{align*}
The third and fourth quantities above can be computed in $O(R)$ work after the first pass over the data.

The first quantity is a sum over $i$ and $r$, and cannot be simplified any further. Computing it takes more than $O(N)$ work. Since its coefficient $\fpb(\kb+2)$ is nonnegative, we must use an upper bound to obtain an over-estimate of $\var(U_a)$. We have the bound
\begin{align*}
\sum_{ir}\zzt_{ir}(1-\nid^{-1}) (1-\nrd^{-1}) 
&\le \sum_{ij}\sum_r\zij\zrj (1-\nid^{-1})\\
& = \sum_j\ndj^2-\sum_{ij}\zij\ndj \nid^{-1},
\end{align*}
which can be computed in $O(N)$ work in a second pass over the data. Other weaker bounds may be obtained without the second pass. An example is
\begin{align*}
\sum_{ir}\zzt_{ir}(1-\nid^{-1}) (1-\nrd^{-1}) 
&\le \sum_{ij}\sum_r\zij\zrj = \sum_{j}\ndj^2
\end{align*}
which can be computed in $O(C)$ work.

For the same reason the second quantity cannot be computed in time $O(N)$ and we upper bound it via $\zzt_{ir}\le\nrd$, getting
\begin{align*}
\sum_{ir}\nid^{-1}\nrd^{-1}\zzt_{ir}(\zzt_{ir}-1)
&\le\sum_{ir}\nid^{-1}\nrd^{-1}\zzt_{ir}(\nrd-1)\\
&=\sum_{ij}\zij\nid^{-1}\ndj - \sum_{ir}\nid^{-1}\nrd^{-1}\zzt_{ir}\\
&\le\sum_{ij}\zij\nid^{-1}\ndj
\end{align*}
which can be computed in $O(N)$ work on a second pass.

All but one expression in $\var(U_e)$ (see~\eqref{eq:varuethm}) can be computed in $O(R+C)$ time after the first pass over the data. The one expression is
\begin{align}\label{eq:unstablechisquare}
N^3 - 2\sum_{ij}\zij\nid\ndj+\Bigl(\sum_i\nid^2\Bigr)\Bigl(\sum_j\ndj^2\Bigr).
\end{align}
The second term in~\eqref{eq:unstablechisquare} requires a second pass over the data in time $O(N)$, because it is the sum over $i$ and $j$ of a polynomial of $\zij$, $\nid$, and $\ndj$. The quantity in~\eqref{eq:unstablechisquare} alternatively can be expressed as
\begin{align}\label{eq:stablechisquare}
\sum_i\sum_j(\nid\ndj-N\zij)^2,
\end{align}
which shows that it is a kind of unnormalized test for row versus column independence in the observation process. Equation~\eqref{eq:stablechisquare} is numerically more stable than~\eqref{eq:unstablechisquare} but requires $O(RC)$ computation which is ordinarily too expensive. 

With the same reasoning as for the second term of~\eqref{eq:unstablechisquare}, we see that $\cov(U_a,U_b)$ can be computed in a second pass over the data in time $O(N)$. This reasoning also shows that we can compute nearly every term in $\cov(U_a,U_e)$ in a second pass over the data; the exception is 
\begin{align}\label{eq:specialcovae}
\sum_i \nid^{-1}\tid^2
\end{align} 
We compute $\tid$ for each $i$ in a second pass over the data. But, we must use additional time $O(R)$ to get \eqref{eq:specialcovae}. Nevertheless, the total computation time is still $O(N)$. Symmetrically $\cov(U_b,U_e)$ can be computed in time $O(N)$ as well.

\subsubsection{Asymptotic approximation of $\var(\hat\theta)$}\label{sec:asympvthetahat}

Under asymptotic conditions, we may obtain simple, analytic approximate expressions for the covariance matrix of our method of moments estimators.

\begin{thm}\label{thm:approx}
As described in Section~\ref{sec:notation}, suppose that 
\begin{align*}
&\nid\le\delta N, && \ndj\le\delta N, && R\le \delta N, && C\le\delta N,
&& N\le\delta \sum_i\nid^2,\quad\text{and} && N\le\delta\sum_j\ndj^2, 
\end{align*}
hold for the same small $\delta>0$ and that 
$$ 0< \ka+2, \kb+2,\ke+2,\fpa,\fpb,\fpe <\infty.$$
Suppose additionally that 
\begin{align}\label{eq:morebounds:main}
& \sum_{ij}\zij\nid^{-1}\ndj \le \delta\sum_i\nid^2,\quad\text{and}
&& \sum_{ij}\zij\nid\ndj^{-1} \le \delta\sum_j\ndj^2 
\end{align}
hold. 
Then 
\begin{align*}
\var(U_a) & = \fpb(\kb+2)\sum_j\ndj^2(1+O(\delta))\\
\var(U_b) & = \fpa(\ka+2)\sum_i\nid^2(1+O(\delta)),\quad\text{and}\\
\var(U_e) & = \Bigl(\fpa(\ka+2)N^2\sum_i\nid^2+\fpb(\kb+2)N^2\sum_j\ndj^2\Bigr)(1+O(\delta)). 
\end{align*}
Similarly 
\begin{align*}
\cov(U_a,U_b)&= \fpe(\ke+2)N(1+O(\delta)),\\
\cov(U_a,U_e)& = \fpb(\kb+2)N\sum_j\ndj^2(1+O(\delta)),\quad\text{and}\\
\cov(U_b,U_e)& = \fpa(\ka+2)N\sum_i\nid^2(1+O(\delta)). 
\end{align*}
Finally $\ssah$, $\ssbh$ and $\sseh$ are asymptotically uncorrelated as $\delta\to0$ with 
\begin{align*}
\var(\ssah) &=\fpa(\ka+2) \frac1{N^2}\sum_j\nid^2(1+O(\delta))\\
\var(\ssbh) &= \fpb(\kb+2) \frac1{N^2}\sum_j\ndj^2(1+O(\delta)),\quad\text{and}\\
\var(\sseh) &= \fpe(\ke+2) \frac1{N}(1+O(\delta)). 
\end{align*}
\end{thm}
\begin{proof}
See Section~\ref{sec:proof:thm:approx} of the supplement. 
\end{proof}

We think that the typical $\ndj$ is large, so $\sum_i\nid^2 = \sum_{ij}\zij\nid$ ought to be much larger than $\sum_{ij}\zij\nid\ndj^{-1}$. A similar argument applies for $\nid$. Thus, the additional bounds in~\eqref{eq:morebounds:main} seem very reasonable. However, it is possible that the pairs where $\zij=1$ with large $\nid$ may have small $\ndj$ and vice versa. \cite{dyer2011visualizing} report such a head-to-tail affinity in several data sets but it would have to be quite extreme for~\eqref{eq:morebounds:main} to require a large $\delta$.

The variance of $\sseh$ is the same variance we would have gotten had $\ssa=\ssb=0$ held. Similar remarks apply for $\ssah$ and $\ssbh$.

\subsubsection{Estimating kurtoses}\label{sec:kurthat}

Under a Gaussian assumption, $\ka=\kb=\ke=0$. If however the data have heavier tails than this, a Gaussian assumption will lead to underestimates of $\var(\hat\theta)$. Therefore, we will estimate the kurtoses by $U$-statistics.

Let $\mfa = \e( a_i^4 ) = (\ka+3)\fpa$, $\mfb = \e( b_i^4 ) = (\kb+3)\fpb$, and $\mfe = \e( e_{ij}^4 ) = (\ke+3)\fpe$. The fourth moment $U$-statistics we use are
\begin{align}\label{eq:wstats}
\begin{split}
W_a &= \dfrac{1}{2} \sum_{ijj'}\nid^{-1}\zij \zijp(Y_{ij}-Y_{ij'})^4 \\
W_b &= \dfrac{1}{2} \sum_{iji'}\ndj^{-1}\zij \zipj (Y_{ij}-Y_{i'j})^4,\quad\text{and}\\
W_e &= \dfrac{1}{2} \sum_{iji'j'}\zij \zipjp (Y_{ij}-Y_{i'j'})^4. 
\end{split}
\end{align}

\begin{thm}\label{thm:ew}
Let $\yij$ follow the random effects model~\eqref{eq:refmodel} with the observation pattern $\zij$ as described in Section~\ref{sec:notation}. Then the statistics defined at~\eqref{eq:wstats} have means
\begin{align*}
\e(W_a) &= (\mfb+3\fpb+12\ssb\sse+\mfe+3\fpe)(N-R)\\
\e(W_b) &= (\mfa+3\fpa+12\ssa\sse+\mfe+3\fpe)(N-C),\quad\text{and}\\
\e(W_e) &= (\mfa+3\fpa+12\ssa\sse)(N^2-\sum_i \nid^2) \\
&\phe+(\mfb+3\fpb+12\ssb\sse)(N^2-\sum_j \ndj^2)\\
&\phe+(\mfe+3\fpe)(N^2-N)+12\ssa\ssb(N^2-\sum_i \nid^2-\sum_j \ndj^2+N).
\end{align*}
\end{thm}
\begin{proof}
See Section \ref{sec:estimatingkurtoses} of the supplement.
\end{proof}

Using Theorem~\ref{thm:ew}, we compute estimates
$\mfah$, $\mfbh$, and $\mfeh$, by solving the $3 \times 3$ system of equations 
\begin{align}\label{eq:momkurt}
M\begin{pmatrix} \mfah \\ \mfbh \\ \mfeh \end{pmatrix} &= 
\begin{pmatrix}
W_a-m_a \\ 
W_b -m_b\\ 
W_e -m_e \end{pmatrix},
\end{align}
where $M$ is the same matrix that we used for the $U$-statistics
in equation~\eqref{eq:momvar}, with
\begin{align*}
m_a &= (3\fpbh+12\ssbh\sseh+3\fpeh)(N-R),\\
m_b &= (3\fpah+12\ssah\sseh+3\fpeh)(N-C), \quad\text{and} \\
m_e &= (3\fpah+12\ssah\sseh)(N^2-\sum_i \nid^2)+(3\fpbh+12\ssbh\sseh)(N^2-\sum_j \ndj^2)\\
&\phe+3\fpeh(N^2-N)+12\ssah\ssbh(N^2-\sum_i \nid^2-\sum_j \ndj^2+N).
\end{align*}

We compute the statistics~\eqref{eq:wstats} via
\begin{align}\label{eq:getwawbwe}
\begin{split}
W_a &= \sum_i\Bigl( \sum_j\zij(\yij-\byid)^4 + 3\nid^{-1}\sid^2\Bigr)\\
W_b &= \sum_j\Bigl( \sum_i\zij(\yij-\bydj)^4 + 3\ndj^{-1}\sdj^2\Bigr),\quad\text{and}\\
W_e &= N\sum_{ij}\zij(\yij-\bydd)^4 + 3\sdd^2,
\end{split}
\end{align}
where $\bydd = N^{-1}\sum_{ij}\zij\yij$ and $\sdd = \sum_{ij}\zij(\yij-\bydd)^2$.

Therefore, the kurtosis estimates $\hat\kappa$ requires $R+C+1$ new quantities
\begin{align}\label{eq:4thquantities}
\sum_j\zij(\yij-\byid)^4,\quad
\sum_i\zij(\yij-\bydj)^4,\quad\text{and}\quad
\sum_{ij}\zij(\yij-\bydd)^4
\end{align}
beyond those used to compute $\hat\theta$. These can be computed in a second pass over the data after $\byid$, $\bydj$ and $\bydd$ have been computed in the first pass. They can also be computed in the first pass using update formulas analogous to the second moment formulas~\eqref{eq:changolubleveque}. Such formulas are given by \cite{pebay2008formulas}, citing an unpublished paper by Terriberry.

Because the kurtosis estimates are used in formulas for $\wh\var(\hat\theta)$ and those formulas already require a second pass over the data, it is more convenient to compute the sample fourth moments via~\eqref{eq:4thquantities} in a second pass. By a similar argument as in Section~\ref{sec:ustatforvarcomp}, obtaining $\kah$, $\kbh$, and $\keh$ has space complexity $O(R+C)$ and time complexity $O(N)$, and is therefore scalable. 

\subsection{Algorithm summary}\label{sec:algosum}

\begin{figure}
	\centering
	\includegraphics[scale=0.7]{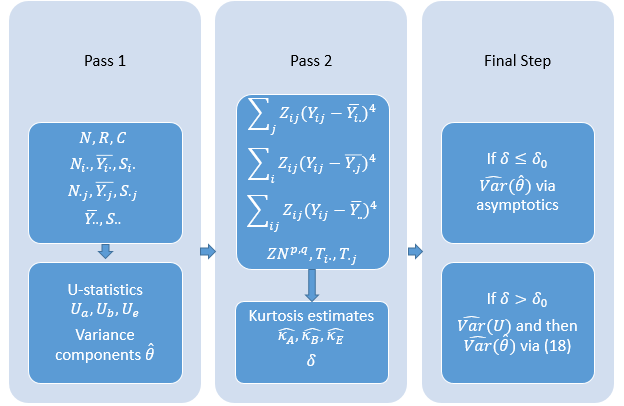}
	\caption{Schematic of our algorithm. The expressions in the smallest boxes are the values computed at each step. The threshold $\delta_0$ is chosen at the discretion of the data analyst and varies between applications.}
	\label{fig:algorithmsteps}
\end{figure}

For clarity of exposition, here we gather all of the steps in our algorithm to estimate $\ssa$, $\ssb$, and $\sse$ and the variances of those estimators. An outline is shown in Figure~\ref{fig:algorithmsteps}. We assume that all of the computations below can be done with large enough variable storage that overflow does not occur. This may require an extended precision representation beyond $64$ bit floating point, such as that in the python package mpmath \citep{mpmath}.

The first task is to compute $\hat\theta$. In a first pass over the data compute counts $N$, $R$, $C$, row values $\nid$, $\byid$, $\sid$ for all unique rows $i$ in the data set, and column values $\ndj$, $\bydj$, $\sdj$ for all unique columns $j$ in the data set as well as $\bydd$ and $\sdd$. Incremental updates are used as described in~\eqref{eq:changolubleveque}.

Then compute
\begin{align*}
U_a  = \sum_i \sid,\quad
U_b  = \sum_j \sdj,\quad\text{and}\quad
U_e = N\sdd,
\end{align*}
the matrix $M$ from~\eqref{eq:momvar} and then
$\hat\theta = (\ssah,\ssbh,\sseh)^\tran = M^{-1}(U_a, U_b, U_e)^\tran$ in time $O(R+C)$.

The second task is to compute approximately the variance of $\hat\theta$. A second pass over the data computes the centered fourth moments in~\eqref{eq:4thquantities}. Then one calculates the fourth order $U$-statistics of equation~\eqref{eq:getwawbwe}, solves~\eqref{eq:momkurt} for the centered fourth moments, and converts them to kurtosis estimates, all in time $O(R+C)$. 

In the second pass over the data, also compute
\begin{align}\label{eq:defzn}
\zn^{p,q} &\equiv 
\sum_{ij}\zij\nid^p\ndj^q
\end{align}
for
\begin{align*}
\begin{pmatrix}p\\q\end{pmatrix}
\in\biggl\{ 
\begin{pmatrix}-1\\-1\end{pmatrix},
\begin{pmatrix}1\\-1\end{pmatrix},
\begin{pmatrix}-1\\1\end{pmatrix},
\begin{pmatrix}1\\1\end{pmatrix},
\begin{pmatrix}-1\\2\end{pmatrix},
\begin{pmatrix}2\\-1\end{pmatrix}
\biggr\}
\end{align*}
as well as $\tid$ and $\tdj$ of equation~\eqref{eq:totdegrees} for all $i$ and $j$ in the data.

Now we may verify whether the limiting approximations in Theorem~\ref{thm:approx} hold. Specifically, compute \[\delta=\max\Bigl(\epsilon_R,\epsilon_C,\dfrac{R}{N},\dfrac{C}{N},\dfrac{N}{\sum_i\nid^2},\dfrac{N}{\sum_j\ndj^2},\dfrac{\sum_{ij}\zij\nid^{-1}\ndj}{\sum_i\nid^2},\dfrac{\sum_{ij}\zij\nid\ndj^{-1}}{\sum_j\ndj^2}\Bigr)\] If $\delta \leq \delta_0$, where $\delta_0$ is a user-specified threshold, then we may use
$$
\wh\var\begin{pmatrix}
\ssah\\
\ssbh\\
\sseh 
\end{pmatrix}
\doteq 
\frac1{N^2}
\begin{pmatrix}
\fpah(\kah+2) \sum_j\ndj^2 \\
&\fpbh(\kbh+2) \sum_i\nid^2 \\
&&\fpeh(\keh+2) N
\end{pmatrix}.
$$

Otherwise, then more work must be done in the second pass. Some of these next computations require even more bits per variable than are needed to avoid overflow, because they involve subtraction in a way that will lose precision.

In this case, estimate the variances of the $U$-statistics. To estimate the variances of $U_a$ and $U_b$, we apply the upper bounds discussed in Section~\ref{sec:vhatuapprox} to \eqref{eq:varuathm} and \eqref{eq:varubthm} and plug in $\ssah$, $\ssbh$, $\sseh$, $\kah$, $\kbh$, and $\keh$, calculating using time and space $O(R+C)$
\begin{align*}
\wh\var(U_a) & = \fpbh(\kbh+2)\Bigl(\sum_j\ndj^2-\zn^{-1,1}\Bigr)
+2\fpbh\Bigl( \zn^{-1,1} -R\sum_i\nid^{-1}\Bigr)\\
&\phe  + 4\ssbh\sseh(N-R)
+\fpeh(\keh+2)\sum_i\nid(1-\nid^{-1})^2+2\fpeh\sum_i(1-\nid^{-1})
\end{align*}
and 
\begin{align*}
\wh\var(U_b) & = \fpah(\kah+2)\Bigl(\sum_i\nid^2-\zn^{1,-1}\Bigr)
+2\fpah\Bigl( \zn^{1,-1} -C\sum_j\ndj^{-1}\Bigr)\\
&\phe  + 4\ssah\sseh(N-C)
+\fpeh(\keh+2)\sum_j\ndj(1-\ndj^{-1})^2+2\fpeh\sum_j(1-\ndj^{-1}).
\end{align*}

To estimate $\var(U_e)$ and the covariances of the $U$-statistics, we again plug in the variance component and kurtosis estimates into Theorem~\ref{thm:vu} without approximation. We get $\wh\var(U_e)$ from \eqref{eq:varuethm}, using $\zn^{1,1}$ from the second pass over the data. We get $\wh\cov(U_a,U_e)$ from~\eqref{eq:covuauethm} using $\zn^{-1,1}$, $\zn^{-1,2}$ and $\tid$, and $\wh\cov(U_b,U_e)$ from~\eqref{eq:covubuethm} using $\zn^{1,-1}$, $\zn^{2,-1}$ and $\tdj$. We get $\wh\cov(U_a,U_b)$ from~\eqref{eq:covuaubthm} using $\zn^{-1,-1}$. It can be easily verified that these calculations also take time and space $O(R+C)$.

The final plug-in estimator of variance is 
\begin{align}\label{eq:twopassplugin}
\wh\var\begin{pmatrix}
\ssah\\ \ssbh \\ \sseh
\end{pmatrix}
= M^{-1}
\begin{pmatrix}
\wh\var(U_a) & \wh\cov(U_a,U_b) & \wh\cov(U_a,U_e)\\
\wh\cov(U_b,U_a) & \wh\var(U_b) & \wh\cov(U_b,U_e)\\
\wh\cov(U_e,U_a) & \wh\cov(U_e,U_b) & \wh\var(U_e)
\end{pmatrix}
(M^{-1})^\tran
\end{align}
where $M$ is the matrix in~\eqref{eq:momvar}.

Aggregating the computation times and counting the number of intermediate values we must calculate, we see that our algorithm takes time $O(N)$ and space $O(R+C)$.

\section{Predictions}\label{sec:shrinkage}

Here we consider an application of variance component estimation to the prediction of a missing observation $\yij$ at given values of $i$ and $j$ in model~\eqref{eq:refmodel}. An equivalent problem is predicting the expected value at those levels of the factors, $\mu+\ai+\bj = \e(\yij\mid \ai,\bj)$.

\subsection{Best linear predictor}

A gold standard is the best linear predictor (BLP), \cite[Chapter 7.3]{SCM09}, which minimizes the MSE over the class of all predictors of the form $\yijh(\lambda) =\sum_{rs}\lrs\zrs\yrs$, where $\lambda$ is the vector of all $\lrs$. In this section, we characterize the weights $\lambda_{rs}^*$ of the BLP. We begin with the MSE
\begin{align}
L(\lambda) &= \e( (\yijh(\lambda)-\yij)^2)
\end{align}

\begin{lemma}\label{lem:mseofblp:main}
The MSEs for the linear predictor $\sum_{rs}\lambda_{rs}\zrs Y_{rs}$ are
\begin{align}\label{eq:predmuabgen}
\begin{split}
L(\lambda) &= \mu^2\Bigl(1-\sum_{rs}\lambda_{rs}\zrs\Bigr)^2 +\ssa+\ssb+\sse\\
&\phe 
+\ssa\sum_{rss'}\lrs\lrsp\zrs\zrsp 
+\ssb\sum_{rsr'}\lrs\lrps\zrs\zrps 
+\sse\sum_{rs}\lrs^2\zrs\\
&\phe -2\Bigl(\ssa\sum_{s}\lis\zis +\ssb\sum_{r}\lrj\zrj+\sse\lij^2\zij\Bigr). 
\end{split}
\end{align}
\end{lemma}
\begin{proof}
See Section~\ref{sec:prooflemmseofblp} of the supplement.
\end{proof}

The weights $\lrs^*$  of the BLP must satisfy the stationarity condition $\partial L(\lrs^*)/\partial\lambda =0.$ As shown in Section~\ref{sec:stationary} of the supplement, when $\zrs=0$, the condition holds no matter the value of $\lrs^*$. When $\zrs =1$, the condition becomes
\begin{align}\label{eq:stationary}
\sse\lrs^*
&=\mu^2\Bigl(1-\sum_{r's'}\lambda_{r's'}^*\zrpsp\Bigr)
+\ssa\Bigl(\oir-\sum_{s'}\lrsp^*\zrsp\Bigr)
+\ssb\Bigl(\ojs-\sum_{r'}\lrps^*\zrps\Bigr)
\end{align}

We can compute $\lambda_{rs}^*$ by solving an $N\times N$ system of equations but that ordinarily costs $O(N^3)$ time. Shortcuts are possible if there is a special pattern in the $\zij$, such as balanced data, but we don't know of any faster way to solve~\eqref{eq:stationary} for general $Z$. Therefore, we consider a smaller class of linear predictors called shrinkage predictors.

\subsection{Shrinkage predictors}

It is reasonable to suppose that the most important observations for predicting $\yij$ are those in its row and column. Therefore we consider predicting $Y_{ij}$ through a linear combination of the overall average, the average in row $i$, and the average in column $j$. We use estimators of the form
\begin{align}\label{eq:shrinkage}
\yijh(\lambda)=\lambda_0\sum_{rs}\zrs Y_{rs}+\lambda_a \sum_s \zis Y_{is}+\lambda_b \sum_r \zrj Y_{rj} 
\end{align} 
where $\lambda=\begin{pmatrix}\lambda_0 & \lambda_a & \lambda_b \end{pmatrix}^\tran$. Then t$\lambda_0$, $\lambda_a$, and $\lambda_b$ are chosen to minimize $L(\lambda)$. By writing~\eqref{eq:shrinkage} in terms of row and column totals we avoid complicated treatments for the situation where row or column means are unavailable because $\nid=0$ or $\ndj=0$ (or both). As an example, if $\min(\nid,\ndj)>0$, then the predictor $\yijh=\byid+\bydj-\bydd$ (from Theorem~\ref{thm:bignidndj:main} below) has $\lambda_0=-1/N$, $\lambda_a=1/\nid$ and $\lambda_b=1/\ndj$.

\begin{lemma}\label{lem:shrinkage:main}
The MSEs for the linear predictor~\eqref{eq:shrinkage} are 
\begin{align*}
L(\lambda)&= \mu^2\bigl(1-\lambda_0N-\lambda_a\nid-\lambda_b\ndj\bigr)^2 
+\lambda_0^2\Bigl(\ssa\sum_r\nrd^2+\ssb \sum_s\nds^2+\sse N\Bigr)\\
&\phe+\lambda_a^2\Bigl(\ssa\nid^2+\ssb \nid+\sse \nid\Bigr) 
+\lambda_b^2\Bigl(\ssa\ndj+\ssb \ndj^2+\sse \ndj\Bigr)+\ssa+\ssb+\sse\\
&\phe-2\lambda_0\Bigl(\ssa\nid+\ssb \ndj+\sse\zij\Bigr) 
-2\lambda_a\Bigl(\ssa\nid+\ssb \zij+\sse\zij\Bigr) \\
&\phe-2\lambda_b\Bigl(\ssa\zij+\ssb \ndj+\sse\zij\Bigr)
+2\lambda_0\lambda_a\Bigl(\ssa\nid^2 + \ssb\sum_s\zis\nds +\sse\nid\Bigr)\\
&\phe+2\lambda_0\lambda_b \Bigl(\ssa\sum_r\zrj\nrd + \ssb\ndj^2+\sse\ndj\Bigr)+2 \lambda_a\lambda_b \zij \Bigl( \ssa\nid + \ssb\ndj+ \sse\Bigr). 
\end{align*}
\end{lemma}
\begin{proof}
See Section~\ref{sec:prooflemshrinkage} of the supplement. 
\end{proof}

\begin{thm}\label{thm:shrinkage:main}
The $\lambda^*$ that minimizes the MSE $L = \e( (\yijh-\yij)^2)$ satisfies $H\lambda^*=c$, where 
$$
c = 
\begin{pmatrix}
N & \nid&\ndj&\zij\\
\nid &\nid&\zij&\zij\\
\ndj&\zij&\ndj&\zij\\
\end{pmatrix}
\begin{pmatrix}
\mu^2\\\ssa\\\ssb\\\sse 
\end{pmatrix},\quad\text{and}\quad H = 
\begin{pmatrix}
H_{11} & H_{12} & H_{13}\\
* & H_{22} & H_{23}\\
* & * & H_{33}\\
\end{pmatrix}
$$
is a symmetric matrix with upper triangular elements
\begin{align*}
H_{11} & = \mu^2N^2 + \ssa\sum_r\nrd^2 + \ssb\sum_s\nds^2+\sse N\\
H_{12} & = \mu^2N\nid + \ssa\nid^2+\ssb\tid+\sse\nid\\
H_{13} & = \mu^2N\ndj + \ssa\tdj+\ssb\ndj^2+\sse\ndj\\
H_{22} & = \mu^2\nid^2 + \ssa\nid^2+\ssb\nid+\sse\nid\\
H_{23} & = \mu^2\nid\ndj + \ssa\zij\nid+\ssb\zij\ndj+\sse\zij,\quad\text{and}\\
H_{33} & = \mu^2\ndj^2 + \ssa\ndj+\ssb\ndj^2+\sse\ndj. 
\end{align*}
\end{thm}
\begin{proof}
See Section \ref{sec:proofshrinkage} of the supplement.
\end{proof}

Given estimates of $\mu$ and $\theta$ we can plug them in to get estimates of the optimal $\lambda$ for prediction at $(i,j)$. Assuming that the algorithm to compute $\hat\theta$ and its variance has been executed, all of $c$ and most of $H$ can be computed using quantities found in the first pass over the data. All of the quantities~\eqref{eq:totdegrees} are available after a second pass.

Therefore, since solving $H\lambda^*=c$ takes time $O(1)$, $\lambda^*$ for predicting a given $\yij$ can be found in time $O(N)$. If we wanted to find $\lambda^*$ for $k$ different sets of $i$ and $j$, the computation cost is $O(N+k)$; we simply would have to store $k$ different $H$'s and $c$'s.

Predicting a missing $\yij$ using Theorem~\ref{thm:shrinkage:main} is simple. Next we look at some special cases to understand how it performs.

\subsubsection*{Special case: $\yij$ in new row and new column}

In this case, $N_{rj}=N_{is}=0$ for any $r,s$, and $\nid=\ndj=0$. The only nonzero entry of $H$ is $H_{11}=\mu^2 N^2+\ssa\sum_r \nrd^2+\ssb\sum_s \nds^2+\sse N$, and the only nonzero entry of $c$ is $c_{1}=\mu^2 N$. Hence $\lambda_a^*=\lambda_b^*=0$ and 
\begin{align*}
\lambda_0^* = \dfrac{\mu^2 N}{\mu^2 N^2+\ssa\sum_r \nrd^2+\ssb\sum_s \nds^2+\sse N}.
\end{align*}
The prediction $\yijh$ is then a shrinkage
$$
\lambda_0^*\ydd=N \lambda_0^*\bydd
=\dfrac{\mu^2 }{\mu^2+\ssa\sum_r \nrd^2/N^2+\ssb\sum_s \nds^2/N^2+\sse /N}\bydd.
$$
In practice we would plug in estimates of $\mu$ and the variance components. As we would expect, this estimate is very close to $\bydd$ for large $N$, when $\hat\mu\ne0$ and the limits~\eqref{eq:assum:bigavgtwo} hold. In that case, the corresponding MSE is $L\doteq\ssa+\ssb+\sse$, which can be verified to be approximately the same as the MSE of the BLP.

\subsubsection*{Special case: $\yij$ in new row but old column}

Suppose that $\zis =0$ for any $s$ but $\exists r$ where $\zrj=1$ , so $\nid=0$ and $\ndj>0$. We would expect most of the weight to be on $\bydj$, the average in the column containing $\yij$. This is indeed the case if $\tdj$ is not large compared to $N$, that is, if the rows that are co-observed with column $j$ do not comprise a large fraction of the data.

Let $c_k$ denote the $k$th entry of $c$ and $H_{k\ell}$ be the entry of $H$ in row $k$ and column $\ell$. In this case, $c_2$ is zero as is the second row and second column of $H$. Therefore, without loss of generality we can take $\lambda_a^*=0$ and $\tilde{\lambda}^*=\begin{pmatrix}\lambda_0^* & \lambda_b^*\end{pmatrix}^\tran$ 
can be computed by solving the system $\tilde{H}\tilde{\lambda}^*=\tilde{c}$, where
\begin{align*}
\tilde{H} &= \begin{pmatrix}H_{11} & H_{13} \\ H_{31} & H_{33}\end{pmatrix}\quad\text{and}\quad
\tilde{c}=\begin{pmatrix}c_{1} \\ c_{3}\end{pmatrix}.
\end{align*}
The following theorem describes the relative size of $\lambda_0^*$ and $\lambda_b^*$ in the big data limit.

\begin{thm}\label{thm:newrow:main}
Suppose that we are predicting $\yij$ where $\nid=0$ but $\ndj>0$. Assume that $0<\mu^2,\ssa,\ssb,\sse < \infty$ and that
$\tdj\equiv\sum_r\nrd\zrj \le \eta N$.
Then
$$\frac{\lambda_0^*}{\lambda_b^*} =\frac1N\frac{\ssa+\sse}\ssb(1+O(\eta))
$$
as $\eta\to0$.
\end{thm}
\begin{proof}
See Section~\ref{sec:proof:thm:newrow} of the Supplement.
\end{proof}

Note that $\lambda_0^*$ is the coefficient of a sum of $N$ observations, while $\lambda_b^*$ is the coefficient of a sum of $\ndj$ observations. Therefore, to more equitably compare the importances of the overall average and the column average for predicting $Y_{ij}$, we consider the ratio
\[\dfrac{N\lambda_0^*}{\ndj \lambda_b^*} \approx \dfrac{\ssa+\sse}{\ssb\ndj}.\]
We may interpret this as the column $j$ average being some multiple of $\ndj$ times as important as the overall average. This makes sense because the more data we have in column $j$, the better estimate we would be able to get of $\mu+b_j$; the overall average only tells us about $\mu$. Also, note that the larger $\sse$ is relative to $\ssb$, the more weight we put on the overall average; we do not trust using only the column average.

\subsubsection*{Special case: large $\nid$ and large $\ndj$}

Next we show that if both row $i$ and column $j$ have a very large number of observations, and the observation matrix $Z$ is not too extreme, then $\yijh$ is approximately $\byid+\bydj-\bydd$ as we might expect. As a result, the customized weights in Theorem~\ref{thm:shrinkage:main} are most useful for cases where one or both of $\nid$ and $\ndj$ are not very large.

\begin{thm}\label{thm:bignidndj:main}
Suppose that  $1/\eta \le \nid\le\eta N$
and $1/\eta \le \ndj\le\eta N$ both hold for some $\eta\in(0,1)$
and that
$0< \mu^2,\ssa,\ssb,\sse <\infty$. 
Then
$$
\yijh = (\byid+\bydj-\bydd) (1+O(\eta)),\quad\text{as $\eta\to0$.}
$$
\end{thm}
\begin{proof}
See Section~\ref{sec:proof:asywts} in the supplement.
\end{proof}

\section{Experimental Results}\label{sec:experiments}

\subsection{Simulations}

First, we compare the performance of our method of moments algorithm (`MoM'), described in Section~\ref{sec:algosum}, to the commonly used R package for mixed models, lme4. lme4 computes the maximum likelihood estimates of the parameters under an assumption of normality.

For our algorithm, we consider a range of data sizes, with $R=C$ ranging from $10$ to $500$. At each fixed value of $R=C$, for $100$ iterations, we generate data according to model~\eqref{eq:refmodel} with normally distributed random effects and $\ssa=2$, $\ssb=0.5$, and $\sse=1$. Exactly $25$ percent of the cells were randomly chosen to be observed. We measure the CPU time needed to obtain the variance component estimates $\ssah$, $\ssbh$, and $\sseh$ (labeled short) and the CPU time need to obtain the variance component estimates as well as upper bounds on the variances of those estimates (labeled long). In addition, we measure the mean squared errors of the variance component estimates. At the end, those five measurements were averaged over the $100$ iterations. 

With regard to lme4, our simulation steps are nearly the same, with the following differences. Due to the slowness of lme4, we only consider data sizes with $R=C$ up to $300$. In addition, because lme4 finds the maximum likelihood variance component estimates, the variances of those estimates were computed asymptotically using the inverse expected Fisher information matrix. The simulation results are shown in Figure~\ref{fig:simulation}.

\begin{figure}
	\begin{subfigure}{0.5\textwidth}
		\centering
		\includegraphics[scale=0.35]{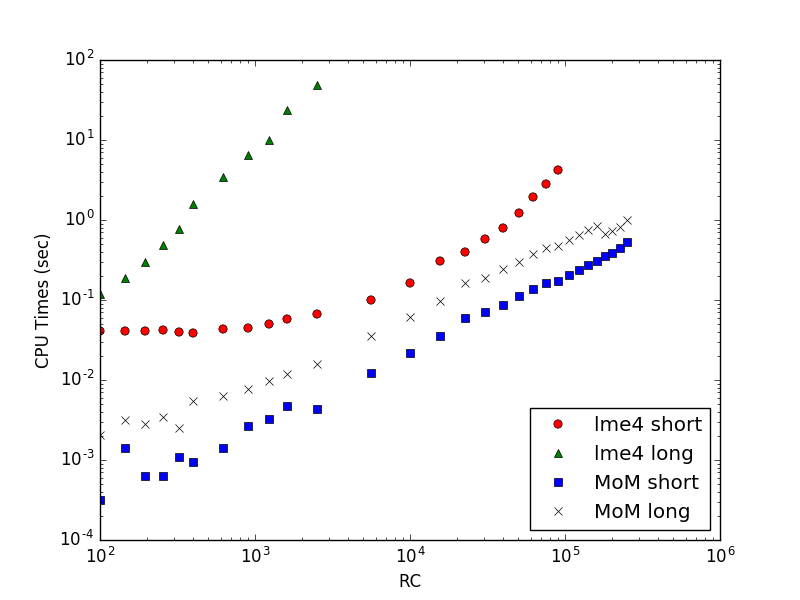}
		\caption{CPU}
		\label{fig:cpu}
	\end{subfigure}%
	\begin{subfigure}{0.5\textwidth}
		\centering
		\includegraphics[scale=0.35]{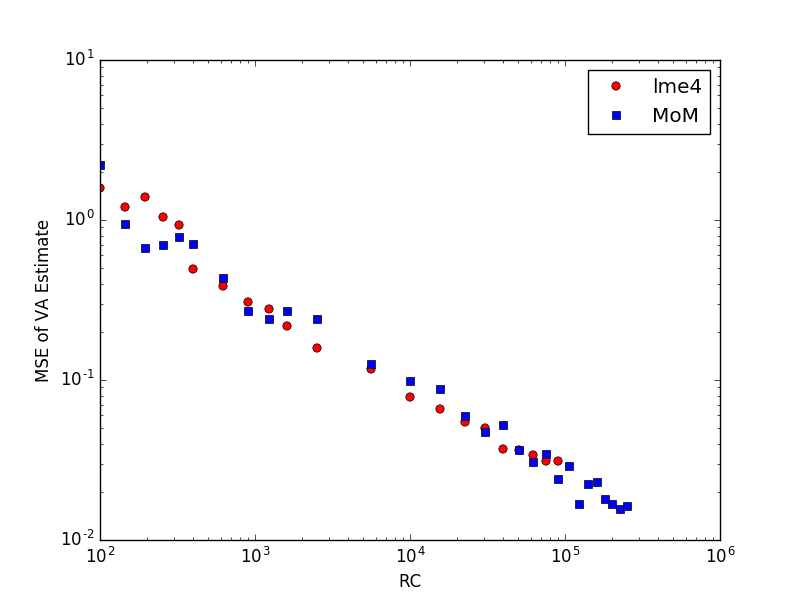}
		\caption{MSE of $\ssah$}
		\label{fig:mseva}
	\end{subfigure}
	
	\begin{subfigure}{0.5\textwidth}
		\centering
		\includegraphics[scale=0.35]{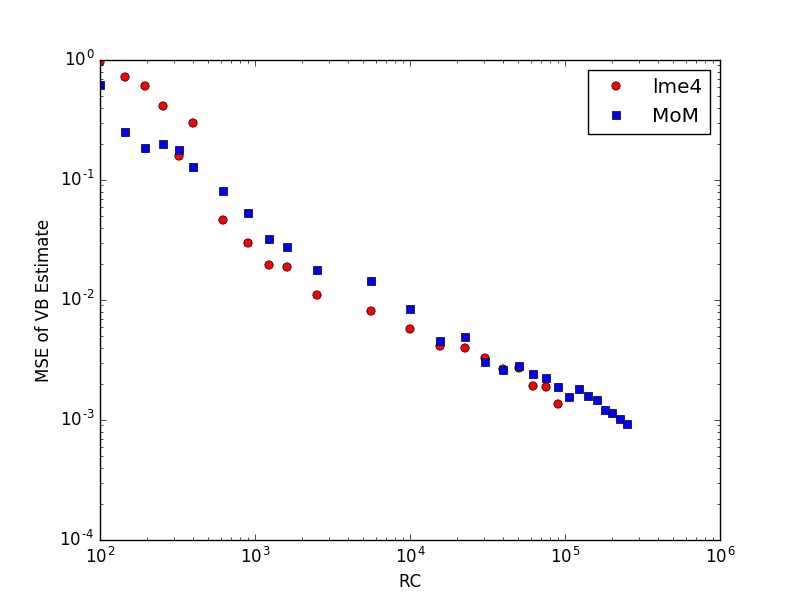}
		\caption{MSE of $\ssbh$}
		\label{fig:msevb}
	\end{subfigure}%
	\begin{subfigure}{0.5\textwidth}
		\centering
		\includegraphics[scale=0.35]{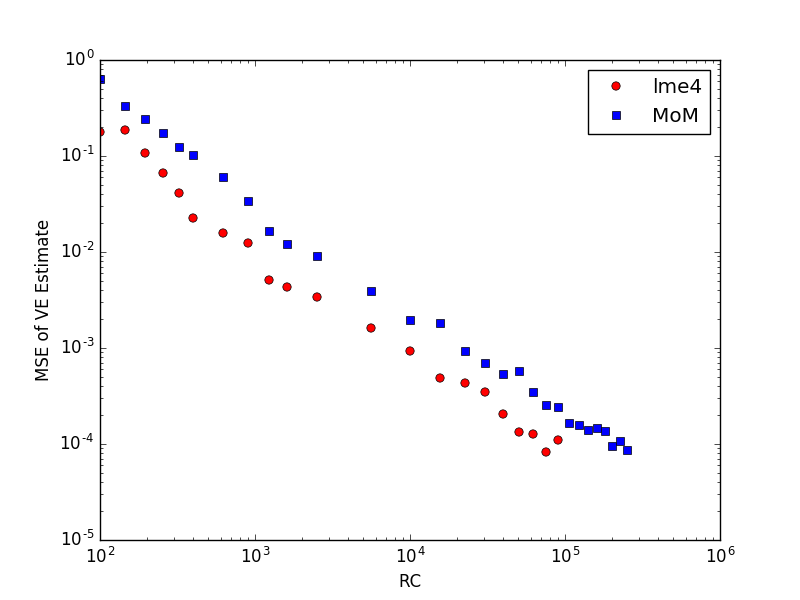}
		\caption{MSE of $\sseh$}
		\label{fig:mseve}
	\end{subfigure}
	\caption{Simulation results: log-log plots of the five recorded measurements against $R*C$, which is proportional to the number of observations. The slope of a fitted line through the scatterplot describes the effect of the x-axis quantity on the y-axis quantity; a slope of $1$ indicates a linear relationship, greater than $1$ a superlinear relationship, and less than $1$ a sublinear relationship.}
	\label{fig:simulation}
\end{figure}

Note that lme4 always takes more time than our algorithm. From Figure~\ref{fig:cpu}, we see that our method of moments algorithm takes time at most linear in the data size to compute both the variance component estimates and upper bounds on the variances of those estimates. For lme4 the computation time is clearly superlinear in the data size, for data sets large enough that the startup cost of the package is no longer dominant. 

The MSEs of $\ssah$ for our algorithm and lme4 are comparable. Moreover, both decrease at most linearly with the data size. The same is true for the MSEs of $\ssbh$. However, the MSE of $\sseh$ in lme4 is noticeably smaller than that of our algorithm; this appears to be the price we pay for the decreased computation time. In both cases, though, the MSE of $\sseh$ decreases approximately linearly with the data size.

\subsection{Real World Data}

We illustrate our algorithm, coded in Python, on three real world data sets that are too large for lme4 to handle in a timely manner. 

The first, from \cite{Webscope}, contains a random sample of ratings of movies by users, which are grades from A+ to F converted into a numeric scale. There are $211,231$ ratings by $7,642$ users on $11,916$ movies, filtered with the condition that each user rates at least ten movies. Only $0.23$ percent of the user-movie matrix is observed.

The estimated variances of the user random effect, the movie random effect, and the error are $2.57$, $2.86$, and $7.68$. The estimated kurtoses are $-2$, $-2$, and $6.56$. Estimated upper bounds on the variances of the estimated variance components are $0.0030$, $0.0018$, and $0.0060$.

The second data set, also from \cite{Webscope2}, contains ratings of $1000$ songs by $15400$ users, on a scale of $1$ to $5$. The first group of $10000$ users were randomly selected on the condition that they had rated at least $10$ of the $1000$ songs. The rest of the users were randomly selected from responders on a survey that asked them to rate a random subset of $10$ of the $1000$ songs. The songs were selected to have at least $500$ ratings. Here, about $2$ percent of the user-song pairs were observed.

The estimated variances of the user random effect, the song random effect, and the error are $0.97$, $0.24$, and $1.30$. The estimated kurtoses are $-2$, $-2$, and $3.31$. Estimated upper bounds on the variances of the estimated variance components are $4.5 \times 10^{-5}$, $10^{-5}$, and $5.8 \times 10^{-5}$. For determining the rating, the user effect is dominant over the song effect.

The third data set from \cite{lastfm} contains the numbers of times artists' songs are played by about $360,000$ users. Only the counts for the top $k$ (for some $k$) artists for each user is recorded. The users are randomly selected. This data set is extremely sparse; only about $0.03$ percent of user-artist pairs are observed. 

The estimated variances of the user random effect, the artist random effect, and the error are $1.65$, $0.22$, and $0.27$. The estimated kurtoses are $0.019$, $-2$, and $23.14$. Estimated upper bounds on the variances of the estimated variance components are $1.68 \times 10^{-5}$, $4.06 \times 10^{-7}$, and $1.37 \times 10^{-6}$. The biggest source of variation in the number of plays is the user, not the artist. The kurtosis of the row effect is nearly zero, indicating possible normality.

In all three data sets at least one of the estimated kurtoses was $-2$, which would be unexpected if the model is correctly specified. However, if model~\eqref{eq:refmodel} does not fit the data well, such behavior may occur. For example, the expected rating of a movie may not be additively decomposable into a movie effect, a user effect, and an error. 

\section{Conclusion}\label{sec:discussion}

When traditional maximum likelihood or MCMC methods are used, with both theory and simulations, we have found that fitting large two-factor crossed unbalanced random effects models has costs that are superlinear in the number of data points, $N$. With the method of moments it is possible to get, in linear time, parameter estimates and somewhat conservative estimates of their variance. The space requirements are proportional to the number of distinct levels of the factors entities; this will often be sublinear in $N$. We also developed shrinkage predictors of missing data that utilize our method of moments estimates.

Through simulations on normally distributed data, we show that our method of moments estimates are competitive with maximum likelihood estimates. We trade off a small increase in the MSE of one variance component for a dramatic decrease in computation time as $N$ gets large.

As stated in the introduction, the crossed random effects model we consider here is the simplest one for which we felt that there was no useful prior solution. We expect that richer models, which are the basis of our future work, will provide better fits to real world data.

In some cases we may be expecting a repeat observation in the $ij$-cell
and then it may be possible to
get a better estimate of $\mu+\ai+\bj$ than $\yij$ is.
Section \ref{sec:prooflemsmoothing} of the supplement considers this problem.

\subsection{Informative Missingness}\label{sec:missingness}

We have assumed throughout that the missingness pattern in $\zij$ is not informative. But in many applications the observed values are likely to differ in some way from the missing values. For instance, in movie ratings data people may be more likely to watch and rate movies they believe they will like, and so missing values could be lower on average than observed ones. In general, the observed ratings may have both high and low values oversampled relative to middling values. 

From observed values alone we cannot tell how different the missing values would be. To do so requires making untestable assumptions about the missingness mechanism. Even in cases where followup sampling can be made, e.g., giving some users incentives to make additional ratings, there will still be difficulties such as users refusing to make those ratings, or if forced, making inaccurate ratings. Methods to adjust for missingness have to be designed on a case by case basis, using whatever additional data and assumptions can be brought to bear. The uncertainties of the estimates from such methods can be quantified using, with further development, the techniques of this paper.

\section*{Acknowledgments}

This work was supported by US NSF under grant DMS-1407397.
KG was supported by US NSF Graduate Research Fellowship under grant DGE-114747. 
Any opinions, findings, and conclusions or recommendations expressed in this material are those of the 
authors and do not necessarily reflect the views of the National Science Foundation.

We would like to thank Brad Klingenberg for his motivation and encouragement during this project. We would also like to thank Rob Tibshirani for his suggestions about our experiments, and Lester Mackey and Norm Matloff for some helpful discussions. 

\bibliographystyle{apalike}
\bibliography{CRErefs}

\allowdisplaybreaks
\section{Appendix~\label{sec:appendix}}

\subsection{Proof of Theorem~\ref{thm:gibbsrate}}
\label{sec:proofthmgibbsrate}

In the balanced case we may assume that $i\in\{1,2,\dots,R\}$
and $j\in \{1,2,\dots,C\}$.
The posterior distribution of the parameters is given by
\begin{align*}
p(\mu,a,b,\ssa,\ssb,\sse \mid Y) &\propto \prod_{i=1}^R\dfrac{1}{\sqrt{2\pi\ssa}}
\exp\Bigl(-\dfrac{a_i^2}{2\ssa}\Bigr)
\prod_{j=1}^C \dfrac{1}{\sqrt{2\pi\ssb}}\exp\Bigl(-\dfrac{b_j^2}{2\ssb}\Bigr) \\
&\phe\times\prod_{i=1}^R\prod_{j=1}^C \dfrac{1}{\sqrt{2\pi\sse}}
\exp\Bigl(-\dfrac{(Y_{ij}-\mu-a_i-b_j)^2}{2\sse}\Bigr) \\
&\propto \sigma_A^{-R}\sigma_B^{-C}\sigma_E^{-RC}\exp\Bigl(-\dfrac{\sum_i a_i^2}{2\ssa}-\dfrac{\sum_j b_j^2}{2\ssb}-\dfrac{\sum_{ij}(Y_{ij}-\mu-a_i-b_j)^2}{2\sse}\Bigr)
\end{align*}

Then, $\phi$ is given by
\begin{align*}
p(a,b \mid \mu,\ssa,\ssb,\sse,Y) &\propto \exp\Bigl(-\dfrac{\sum_i a_i^2}{2}\Bigl(\dfrac{1}{\ssa}+\dfrac{C}{\sse}\Bigr)-\dfrac{\sum_j b_j^2}{2}
\Bigl(\dfrac{1}{\ssb}+\dfrac{R}{\sse}\Bigr)-\dfrac{\sum_{ij}a_i b_j}{\sse}\Bigr).
\end{align*}
Therefore, the posterior distribution of $a$ and $b$ 
is a joint normal with precision matrix 
\begin{align*}
Q &= \begin{pmatrix}\dfrac{\sse+C\ssa}{\ssa\sse}I_R & \dfrac{1}{\sse}1_R 1_C^\tran \\ \dfrac{1}{\sse}1_C 1_R^\tran & \dfrac{\sse+R\ssb}{\ssb\sse}I_C\end{pmatrix}.
\end{align*}

From Theorem 1 of \cite{RS97}, for the Gibbs sampler described in Section~\ref{sec:gibbssampling}, we have the following result. Let $A=I-\diag(Q_{11}^{-1},Q_{22}^{-1})Q$, where $Q_{11}$ denotes the upper left block of $Q$ and $Q_{22}$ denotes the lower right block. Let $L$ be the block lower triangular part of $A$, and $U=A-L$. Then, the convergence rate $\rho$ is given by the spectral radius of the matrix $B=(I-L)^{-1}U$. Now, we compute $\rho$.
First
\begin{align*}
A &= I-\begin{pmatrix}\dfrac{\ssa\sse}{\sse+C\ssa}I_R & 0 \\ 0 & \dfrac{\ssb\sse}{\sse+R\ssb}I_C\end{pmatrix}Q  
=\begin{pmatrix}0 & -\dfrac{\ssa}{\sse+C\ssa}1_R 1_C^\tran\\ 
-\dfrac{\ssb}{\sse+R\ssb}1_C 1_R^\tran & 0 \end{pmatrix}. 
\end{align*}
Next
\begin{align*}
L &= \begin{pmatrix}0 & 0 \\ -\dfrac{\ssb}{\sse+R\ssb}1_C 1_R^\tran & 0 \end{pmatrix} \quad\text{and}\quad
U = \begin{pmatrix}0 & -\dfrac{\ssa}{\sse+C\ssa}1_R 1_C^\tran\\ 0 & 0 \end{pmatrix} 
\end{align*}
from which
\begin{align*}
B &= \begin{pmatrix}I_R & 0 \\ \dfrac{\ssb}{\sse+R\ssb}1_C 1_R^\tran & I_C \end{pmatrix}^{-1}U =\begin{pmatrix}I_R & 0 \\ -\dfrac{\ssb}{\sse+R\ssb}1_C 1_R^\tran & I_C \end{pmatrix}U \\
&= \begin{pmatrix} 0 & -\dfrac{\ssa}{\sse+C\ssa}1_R 1_C^\tran \\ 0 & \dfrac{R\ssa\ssb}{(\sse+C\ssa)(\sse+R\ssb)}1_C 1_C^\tran\end{pmatrix}.
\end{align*}
Clearly, $B$ has rank one. Then, its spectral radius must be equal to its nonzero eigenvalue, which is also the trace of $B$. Hence, 
$$\rho=\dfrac{RC\ssa\ssb}{(\sse+C\ssa)(\sse+R\ssb)}$$ 

\subsection{Simulation results}

The results of our simulations described in Section~\ref{sec:mcmc} are presented here in Tables~\ref{tab:cpu} through~\ref{tab:ssesim}. 

\begin{table}
\small 
\centering 
\resizebox{\textwidth}{!}{\begin{tabular}{lccccccccc}
\toprule
Method           & Gibbs & Block & Reparam. & Lang. & MALA  & Indp. &  RWM & RWM Sub. & pCN  \\ \hline 
\toprule
R=10 \\ C=10     & 20    & 9     & 23       &    20 &   27  &    21 &   19 &    21    & 21   \\ \hline 
R=20 \\ C=20     & 33    & 10    & 37       &    35 &   45  &    34 &   32 &    33    & 33   \\ \hline 
R=50 \\ C=50     & 71    & 17    & 80       &    79 &  101  &    71 &   68 &    75    & 70   \\ \hline 
R=100 \\ C=100   & 143   & 361   & 159      &   156 &  199  &   139 &  133 &   141    & 136  \\ \hline 
R=200 \\ C=200   & 326   & 984   & 351      &   323 &  462  &   300 &  279 &   303    & 280  \\ \hline 
R=500 \\ C=500   & 1157  & 2356  & 1205     &   955 & 1786  &   952 &  851 &  1019    & 817  \\ \hline 
R=1000 \\ C=1000 & 3432  & 15046 & 4099     &  2302 & 4760  &  2513 & 2141 &  2635    & 1966 \\ \hline 
R=2000 \\ C=2000 & 10348 & 88756 & 11434    &  6991 & 15836 &  7815 & 5712 &  9274    & 6006 \\ \hline 
R=50 \\ C=100    & 105   & 287   &   121    &   112 &  151  &   103 &  101 &   107    & 102  \\ \hline 
R=10 \\ C=200    & 138   & 316   &    167   &   139 &  200  &   138 &  137 &   142    & 138  \\ \hline 
R=100 \\ C=1000  & 898   & 5148  &     964  &   807 & 1179  &   795 &  748 &   822    & 760  \\
\end{tabular}}
\caption{\label{tab:cpu}
Median CPU time in seconds.}
\end{table}

\begin{table}
\small 
\centering 
\resizebox{\textwidth}{!}{\begin{tabular}{lccccccccc}
\toprule
Method   & Gibbs & Block & Reparam. & Lang. & MALA & Indp. & RWM  & RWM Sub. & pCN  \\ \hline 
\toprule
R=10     & 0.72  & 0.94  & 1.27     &  1.07 & 1.18 &  2.40 & 0.76 &   0.74   & 1.51 \\ 
C=10     & 26    & 29    & 24       &   178 &  689 &  1604 & 1252 &   1522   & 1392 \\ \hline 
R=20     & 0.81  & 1.02  & 1.01     &  1.07 & 0.94 &  2.89 & 1.69 &   1.08   & 1.47 \\
C=20     & 34    & 43    & 26       &    75 &  841 &  1019 & 1674 &   1720   & 1765 \\ \hline 
R=50     & 1.09  & 0.91  & 0.98     &  0.98 & 1.04 &  2.97 & 1.66 &   1.70   & 1.58 \\
C=50     & 83    & 84    & 75       &     8 &  610 & 5000+ & 1158 &   1681   & 1104 \\ \hline 
R=100    & 0.98  & 1.02  & 1.13     &  0.99 & 0.85 &  2.73 & 1.57 &   1.61   & 1.49 \\
C=100    & 123   & 185   & 144      &     2 &  398 & 5000+ & 1145 &   1713   & 1522 \\ \hline 
R=200    & 1.01  & 1.02  & 1.03     &  1.01 & 0.95 &  3.22 & 1.60 &   1.31   & 1.52 \\
C=200    & 257   & 346   & 272      &     1 &    1 &  1278 & 1508 &   1692   & 807  \\ \hline 
R=500    & 0.99  & 1.01  & 0.99     &  0.99 & 1.00 &  2.26 & 1.58 &   1.15   & 1.55 \\
C=500    & 536   & 617   & 576      &     9 &    4 &  1572 &  924 &   1687   & 1613 \\ \hline 
R=1000   & 0.97  & 1.02  & 1.04     &  0.99 & 0.96 &  2.39 & 1.55 &   1.07   & 1.53 \\
C=1000   & 801   & 790   & 694      &     1 & 2501 & 5000+ & 1133 &   1656   & 1008 \\ \hline 
R=2000   & 0.98  & 1.01  &  1.00    &  1.01 & 1.00 &  2.57 & 1.55 &   1.03   & 1.55 \\
C=2000   & 672   & 721   &     771  &     1 & 5000+ &  1086 & 1176 &   1716   & 799  \\ \hline 
R=50     & 0.89  & 1.03  &    0.95  &  1.01 & 1.06 &  2.70 & 1.57 &   1.61   & 1.45 \\
C=100    & 144   & 155   &     118  &     7 & 1095 & 5000+ & 1219 &   1725   & 1371 \\ \hline 
R=10     & 0.86  & 1.08  &   0.84   &  0.94 & 0.80 &  2.40 & 1.41 &   1.36   & 1.23 \\ 
C=200    & 329   & 244   &     299  &   120 &  944 &  3339 & 1518 &   1657   & 1437 \\ \hline 
R=100    & 1.06  & 1.06  &    1.02  &  1.01 & 1.03 &  2.73 & 1.57 &   1.11   & 1.55 \\
C=1000   & 573   & 536   &    672   &     1 &    1 &  3330 & 1161 &   1681   & 3333 \\
\end{tabular}}
\caption{\label{tab:musim}
Median estimates of $\mu$ and lag when ACF$(\hat\mu)\leq 0.5$.}
\end{table}

\begin{table}
\small 
\centering 
\resizebox{\textwidth}{!}{\begin{tabular}{lccccccccc}
\toprule
Method   & Gibbs & Block & Reparam. & Lang. & MALA & Indp. & RWM  & RWM Sub. & pCN  \\ \hline 
\toprule
R=10     & 2.76  & 2.49  &  2.05    &  2.07 & 2.45 &  2.39 & 1.88 &  2.05    & 1.38 \\ 
C=10     & 1     & 1     &     1    &   898 &  768 &  1604 &  759 &   606    & 1232 \\ \hline 
R=20     & 2.00  & 2.06  &  1.65    &  1.89 & 2.32 &  1.48 & 1.96 &  1.76    & 2.00 \\
C=20     & 1     & 1     &     1    &   930 &  829 &   850 &  873 &   822    & 1083 \\ \hline 
R=50     & 1.94  & 1.96  &  2.17    &  1.77 & 2.21 &  1.44 & 2.06 &  2.03    & 1.95 \\
C=50     & 1     & 1     &     1    &   797 &  720 & 5000+ & 1035 &  1032    & 1079 \\ \hline 
R=100    & 2.21  & 2.14  &  2.23    &  1.88 & 1.87 &  1.11 & 2.19 &  1.92    & 1.95 \\
C=100    & 1     & 1     &     1    &   649 &  398 & 5000+ &  994 &   917    & 1522 \\ \hline 
R=200    & 2.09  & 2.09  &  2.10    &  2.08 & 1.99 &  1.16 & 2.02 &  2.12    & 2.01 \\
C=200    & 1     & 1     &     1    &   410 &  437 &  1281 & 1598 &   673    & 1135 \\ \hline 
R=500    & 1.97  & 2.12  &  1.99    &  1.64 & 1.96 &  1.07 & 2.02 &  2.01    & 1.97 \\
C=500    & 1     & 1     &     1    &   407 &  197 &  1572 &  895 &   826    & 1599 \\ \hline 
R=1000   & 1.96  & 1.99  &  2.02    &  1.90 & 1.95 &  1.78 & 2.01 &  1.96    & 1.99 \\
C=1000   & 1     & 1     &     1    &   122 & 2656 & 5000+ & 1133 &   989    & 912  \\ \hline 
R=2000   & 1.97  & 2.00  &    2.03  &  1.94 & 1.99 &  1.04 & 2.01 &  2.00    & 1.99 \\
C=2000   & 1     &    1  &     1    &    69 & 5000+ &  1086 & 1181 &  1262    & 1161 \\ \hline 
R=50     & 2.22  & 2.29  &  2.05    &  2.24 & 1.98 &  1.10 & 2.00 &  1.96    & 2.09 \\
C=100    & 1     & 1     &     1    &   948 &  672 & 5000+ & 1103 &   787    & 1005 \\ \hline 
R=10     & 2.34  & 1.74  &  3.05    &  2.70 & 2.72 &  0.88 & 1.89 &  1.43    & 1.16 \\ 
C=200    & 1     & 1     &      1   &   891 & 1023 &  3309 & 1492 &   724    & 988  \\ \hline 
R=100    & 2.04  & 2.03  &  2.14    &  1.98 & 1.98 &  1.46 & 1.90 &  1.87    & 2.05 \\
C=1000   & 1     & 1     &     1    &   512 &  450 &  3329 &  985 &  1086    & 3333 \\
\end{tabular}}
\caption{\label{tab:ssasim}
Median estimates of $\ssa$ and lag when ACF($\hat\ssa$) $\leq 0.5$}
\end{table}

\begin{table}
\small 
\centering 
\resizebox{\textwidth}{!}{\begin{tabular}{lccccccccc}
\toprule
Method   & Gibbs & Block & Reparam. & Lang. & MALA & Indp. & RWM  & RWM Sub. & pCN  \\ \hline 
\toprule
R=10     & 0.66  & 0.81  & 0.88     &  0.46 & 0.89 &  1.47 & 0.45 &   0.43   & 0.45 \\ 
C=10     & 1     &   1   &    1     &   382 &  638 &  1604 & 1214 &    956   & 1297 \\ \hline 
R=20     & 0.54  & 0.45  & 0.44     &  0.43 & 0.44 &  1.55 & 0.49 &   0.46   & 0.57 \\
C=20     & 1     &    1  &    1     &   261 &  410 &  978  &  937 &   1217   & 704  \\ \hline 
R=50     & 0.49  & 0.49  & 0.49     &  0.49 & 0.53 &  1.35 & 0.49 &   0.43   & 0.48 \\
C=50     & 1     &   1   &    1     &   123 &  138 & 5000+ & 1308 &    786   & 1463 \\ \hline 
R=100    & 0.51  & 0.54  & 0.49     &  0.46 & 0.48 &  0.84 & 0.52 &   0.47   & 0.49 \\
C=100    & 1     &   1   &    1     &    65 &   66 & 5000+ &  691 &   1169   & 1522 \\ \hline 
R=200    & 0.49  & 0.51  & 0.51     &  0.47 & 0.50 &  1.67 & 0.51 &   0.49   & 0.50 \\
C=200    & 1     &  1    &    1     &    36 &   37 &  1266 & 1497 &   1241   & 831  \\ \hline 
R=500    & 0.51  & 0.49  & 0.50     &  0.28 & 0.47 &  1.56 & 0.50 &   0.48   & 0.47 \\
C=500    & 1     &   1   &    1     &   770 &   16 &  1572 &  696 &    993   & 1619 \\ \hline 
R=1000   &  0.51 & 0.50  & 0.50     &  0.40 & 0.50 &  2.94 & 0.51 &   0.50   & 0.49 \\
C=1000   & 1     &   1   &    1     &   477 & 2514 & 5000+ & 1133 &    855   & 556  \\ \hline 
R=2000   & 0.50  & 0.50  &    0.49  &  0.39 & 0.50 &  1.65 & 0.48 &   0.49   & 0.50 \\
C=2000   & 1     &    1  &     1    &   224 & 5000+ &  1086 & 1220 &    830   & 1253 \\ \hline 
R=50     & 0.50  & 0.51  &   0.53   &  0.48 & 0.54 &  1.93 & 0.53 &   0.49   & 0.49 \\
C=100    & 1     &    1  &      1   &    69 &   85 & 5000+ & 1378 &    910   & 1419 \\ \hline 
R=10     & 0.47  & 0.51  & 0.51     &  0.40 & 0.52 &  1.65 & 0.61 &   0.59   & 0.55 \\ 
C=200    & 1     &   1   &     1    &    23 &   52 &  3332 & 1289 &   1004   & 1408 \\ \hline 
R=100    &  0.50 & 0.49  &  0.50    &  0.47 & 0.49 &  2.95 & 0.50 &   0.49   & 0.50 \\
C=1000   & 1     &    1  &      1   &     6 &    8 &  3328 & 1345 &    962   & 3333 \\
\end{tabular}}
\caption{\label{tab:ssbsim}
Median estimates of $\ssb$ and lag when ACF($\hat\ssb$) $\leq 0.5$}
\end{table}

\begin{table}
\small 
\centering 
\resizebox{\textwidth}{!}{\begin{tabular}{lccccccccc}
\toprule
Method   & Gibbs & Block & Reparam. & Lang.  & MALA   & Indp. & RWM  & RWM Sub. & pCN  \\ \hline 
\toprule
R=10     & 1.02  & 0.99  & 0.96     &  0.91  & 1.17   &  0.17 & 0.76 &  0.80    & 0.75 \\ 
C=10     & 1     &    1  &    1     &   196  &  334   &  1604 & 1354 &  1329    & 1504 \\ \hline 
R=20     & 0.97  & 0.98  & 1.00     &  0.91  & 1.00   &  0.17 & 0.48 &  0.45    & 0.37 \\
C=20     & 1     &    1  &    1     &    61  &   75   &  1218 & 1649 &  1614    & 1827 \\ \hline 
R=50     & 1.00  & 1.01  & 0.98     &  0.96  & 0.99   &  0.17 &    0 &  0.01    &    0 \\
C=50     & 1     &    1  &    1     &    10  &   12   & 5000+ & 1107 &  1616    & 1466 \\ \hline 
R=100    &  1.00 & 1.00  & 1.00     &  0.98  & 1.00   &  0.16 &    0 &  0.38    &    0 \\
C=100    & 1     &   1   &    1     &     3  &    3   & 5000+ & 1199 &  1714    & 1532 \\ \hline 
R=200    & 1.00  & 1.00  & 1.00     &  1.01  & 1.01   &  0.21 &    0 &  0.66    &    0 \\
C=200    & 1     &    1  &    1     &     1  &    1   &  1266 & 1626 &  1691    &  636 \\ \hline 
R=500    &  1.00 & 1.00  & 1.00     & 118.45 & 52.70  &  0.14 &    0 &  0.87    &    0 \\
C=500    & 1     &    1  &    1     &    545 &   138  &  1572 &  834 &  1702    & 1616 \\ \hline 
R=1000   & 1.00  & 1.00  & 1.00     &  65.22 & 2.66   &  0.15 &    0 &  0.93    &    0 \\
C=1000   & 1     &   1   &    1     &    385 & 3062   & 5000+ & 1518 &  1724    &  621 \\ \hline 
R=2000   & 1.00  & 1.00  &   1.00   & 115.59 & 1.05   &  0.18 &    0 &  0.97    &    0 \\
C=2000   & 1     &    1  &      1   &     10 & 5000+  &  1021 & 1194 &  1702    & 1014 \\ \hline 
R=50     &  1.01 & 0.99  &  1.00    &   0.98 & 1.01   &  0.15 &    0 &  0.19    &    0 \\
C=100    & 1     &   1   &      1   &      5 &    6   & 5000+ & 1676 &  1774    & 1442 \\ \hline 
R=10     & 0.99  & 0.99  &   1.01   &   0.92 & 0.99   &  0.17 &    0 &  0.55    &    0 \\ 
C=200    & 1     &   1   &     1    &     12 &   15   &  3309 & 1570 &  1678    & 1279 \\ \hline 
R=100    & 1.00  & 1.00  &   1.00   &   3.50 & 3.46   &  0.19 &    0 &  0.87    &    0 \\
C=1000  & 1      &   1   &      1   &      3 &    3   &  3330 & 1454 &  1699    & 3333 \\
\end{tabular}}
\caption{\label{tab:ssesim}
Median estimates of $\sse$ and lag when ACF($\hat\sse$) $\leq 0.5$}
\end{table}







 \newcommand{\zispp}{Z_{is''}}
 \newcommand{\zrppj}{Z_{r''j}}
 \newcommand{\oirpp}{1_{i=r''}}

 \newcommand{\ojspp}{1_{j=s''}}

 \newcommand{\yis}{Y_{is}}
 \newcommand{\yisp}{Y_{is'}}

 \newcommand{\yrj}{Y_{rj}}
 \newcommand{\real}{\mathbb{R}}













 \newcommand{\yijt}{{\tilde Y}_{ij}}

 \newcommand{\lrppspp}{\lambda_{r''s''}}
 \newcommand{\zrppspp}{Z_{r''s''}}
 \newcommand{\lrppsp}{\lambda_{r''s'}}
 \newcommand{\zrppsp}{Z_{r''s'}}
 \newcommand{\lrpspp}{\lambda_{r's''}}
 \newcommand{\zrpspp}{Z_{r's''}}

\vfill\eject
\begin{center}
Supplementary material for:\\
 Efficient moment calculations for variance components 
 in large unbalanced crossed random effects models\\
by Katelyn Gao and Art B. Owen, Stanford University 
\end{center}




\maketitle

\begin{abstract}
This is a supplementary document containing proofs
for some results in the main document. 
The section numbers continue where that document left off.
Some contextual material is repeated for clarity.  Also
as this is a supplementary document, material that is
traditionally left out as being `tedious algebra' is included
in full detail, making the numerous steps easier to follow and check. 
\end{abstract}

\setcounter{section}{8} 
\setcounter{equation}{99}
\section{Partially observed random effects model}\label{sec:supplfirst}
The random effects model is
\begin{align}\label{eq:reffect}
\yij = \mu + \ai + \bj +\eij,\quad i,j\in\natu
\end{align}
for $a_i \simiid F_a$, $b_j\simiid F_b$ and $e_{ij}\simiid F_e$ independent
of each other. These random variables have mean $0$,
variances $\ssa$, $\ssb$, $\sse$ and kurtoses $\ka$, $\kb$, $\ke$,
respectively. We will not need their skewnesses.

We use letters $i,i',r,r'$ to index rows. Letters $j,j',s,s'$ are used for columns.
In internet applications, the actual indices may be people rating items, 
items being rated, cookies, URLs,
IP addresses, query strings, image identifiers and so on.
We simplify the index set to $\natu$ for notational convenience.  One feature
of these variables is that we fully expect future data to bring hitherto unseen
levels. That is why a countable index set is appropriate.

We will want to estimate $\ssa$, $\ssb$, $\sse$ and
get a formula for the  variance of those estimates.
Many, perhaps most, of the $\yij$ values are missing.
Here we assume that the missingness is not informative.
For further discussion see Section~\ref{sec:missingness} of the main document.

The variable $\zij\in\{0,1\}$ takes the value $1$ if $\yij$ is available
and $0$ otherwise.
The total sample size is $N=\sum_{ij}\zij$. We assume that $1\le N<\infty$.
We also need $\nid = \sum_j\zij$ and $\ndj=\sum_i\zij$.
The number of unique observed rows and columns are, respectively,
$$
R\equiv\sum_i1_{\nid>0},\quad\text{and}\quad 
C\equiv\sum_j1_{\ndj>0}.
$$
In the sum above, only finitely many summands are nonzero.
When we sum over $i,i',r,r'$,
the sum is over the set $\{i\mid\nid>0\}$.  Similarly sums
over column indices $j,j',s,s'$ are over the set $\{j\mid \ndj>0\}$.
These ranges are what one would naturally get in a pass over data
logs showing all records.

We frequently need the number of columns jointly 
observed in two rows such as $i$ and $i'$. 
This is $\sum_j\zij\zipj=(ZZ^\tran)_{ii'}$. 
Similarly, columns $j$ and $j'$ are jointly observed 
in $\sum_i\zij\zijp = (Z^\tran Z)_{jj'}$ rows.

The matrix $Z$ encodes several different measurement regimes as
special cases. These include crossed designs, nested designs and IID
sampling, as follows.
A crossed design with an $R\times C$ matrix of completely observed data
can be represented via $\zij=1_{1\le i\le R}1_{1\le j\le C}$.
If $\max_i \nid=1$ and $\max_j\ndj>1$ then the data 
have a nested structure, with $\ndj$ distinct rows in column $j$
and $(Z^\tran Z)_{jj'}=0$ for $j\ne j'$.
Similarly $\max_j\ndj=1$ with $\max_i\nid>1$ yields
columns nested in rows.  If $\max_i\nid=\max_j\ndj=1$
then we have $N$ IID observations.

We note some identities:
\begin{align}
\sum_{ir}\zzt_{ir}&=\sum_{ijr}\zij\zrj=\sum_j\ndj^2,\quad\text{and}\label{eq:sscount}\\
\sum_{ir}\nid^{-1}\zzt_{ir}&=\sum_{ijr}\nid^{-1}\zij\zrj=\sum_{ij}
\zij\nid^{-1}\ndj. \label{eq:avgcolbyrowcount}
\end{align}

We need some notation for equality among index sets. 
The notation $1_{ij=rs}$ means $1_{i=r}1_{j=s}$.
It is different from $1_{\{i,j\}=\{r,s\}}$ which we also use.
Additionally, $1_{ij\ne rs}$ means $1-1_{ij=rs}$. 

\section{Weighted U statistics}

We will work with weighted U-statistics
\begin{align*} 
U_a & = \frac12\sum_{ijj'} \ui\zij\zijp(\yij-\yijp)^2\\ 
U_b & = \frac12\sum_{iji'} \vj\zij\zipj(\yij-\yipj)^2,\quad\text{and}\\ 
U_e & = \frac12\sum_{iji'j'} \wij\zij\zipjp(\yij-\yipjp)^2,
\end{align*} 
for weights $u_i$, $v_j$ and $w_{ij}$ chosen below.

We can write $U_a=\sum_iu_i\nid(\nid-1)s^2_{i\sumdot}$
where $s^2_{i\sumdot}$ is an unbiased estimate of
$\ssb+\sse$ from within any row $i$ with $\nid\ge2$.
Under our model the values in row $i$ are IID with mean
$\mu+a_i$ and variance $\ssb+\sse$, and so
\begin{align*}
\var( s^2_{i\sumdot}) 
&= (\ssb+\sse)^2\Bigl( \frac2{\nid-1} + \frac{\kappa(b_j+e_{ij})}{\nid}\Bigr)
\end{align*}
where $\kappa(b_j+e_{ij}) = (\kb\fpb+\ke\fpe)/(\ssb+\sse)^2$ is the kurtosis of $\yij$
for the given $i$ and any $j$.
Thus
\begin{align}\label{eq:varsrowi}
\var( s^2_{i\sumdot}) 
&= 
\frac{2 (\ssb+\sse)^2}{\nid-1} + \frac{\kb\fpb}{\nid}+ \frac{\ke\fpe}{\nid}.
\end{align}
Inverse variance weighting
then suggests that we weight $s^2_{i\sumdot}$ proportionally to a value between
$\nid$ and $\nid-1$.
Weighting proportional to $\nid-1$ has the advantage of zeroing out
rows with $\nid=1$. 
This consideration motivates us to take $\ui=1/\nid$, and similarly
$\vj=1/\ndj$.

If $U_e$ is dominated by contributions from $\eij$ then the
observations enter symmetrically and there is no reason to
not take $\wij=1$.  Even if the $\eij$ do not dominate, the statistic
$U_e$ compares more data pairs than the others.  It is unlikely
to be the information limiting statistic.  So $\wij=1$ is a reasonable
default.

If the data are IID then only $U_e$ above is nonzero. This is appropriate
as only the sum $\ssa+\ssb+\sse$ can be identified in that case.
For data that are nested but not IID, only two of the U-statistics
above are nonzero and in that case only one of $\ssa$ and $\ssb$
can be identified separately from $\sse$.

The U-statistics we use are then
\begin{align}\label{eq:defu}
\begin{split}
U_a & = \frac12\sum_{ijj'} \nid^{-1}\zij\zijp(\yij-\yijp)^2\\
U_b & = \frac12\sum_{iji'} \ndj^{-1}\zij\zipj(\yij-\yipj)^2,\quad\text{and}\\
U_e & = \frac12\sum_{iji'j'} \zij\zipjp(\yij-\yipjp)^2.
\end{split}
\end{align}
Because we only sum over $i$ with $\nid>0$ and $j$ with $\ndj>0$,  
our sums never include $0/0$.  

\subsection{Expected $U$-statistics}\label{sec:proof:lem:eu}

Here we find the expected values for our three $U$-statistics.
\begin{lemma}\label{lem:eu}
Under the random effects model~\eqref{eq:reffect}, the U-statistics in~\eqref{eq:defu}
satisfy
\begin{align}\label{eq:eu}
\begin{pmatrix}
\e(U_a)\\[.5ex]
\e(U_b)\\[.5ex]
\e(U_e) 
\end{pmatrix}
&=
\begin{pmatrix}
0 & N-R&N-R\\[.5ex]
N-C&0&N-C\\[.5ex]
N^2-\sum_i\nid^2 & N^2-\sum_j\ndj^2 &N^2-N 
\end{pmatrix}
\begin{pmatrix}
\ssa\\[.5ex]
\ssb\\[.5ex]
\sse 
\end{pmatrix}.
\end{align}
\end{lemma}
\begin{proof}
First we note that
\begin{align*}
\e((\ai-\aip)^2) & = 2\ssa(1-\oiip)\\
\e((\bj-\bjp)^2) & = 2\ssb(1-\ojjp),\quad\text{and}\\
\e((\eij-\eipjp)^2) & = 2\sse(1-\oiip\ojjp).
\end{align*}
Now $\yij-\yijp=\bj-\bjp+\eij-\eijp$, and so
\begin{align*}
\e(U_a)&=\frac12\sum_{ijj'}\nid^{-1}\zij\zijp 
\bigl( 2\ssb(1-\ojjp)+2\sse(1-\oii\ojjp)\bigr)\\
& = (\ssb+\sse)\sum_{ijj'}\nid^{-1}\zij\zijp(1-\ojjp)\\
& = (\ssb+\sse)\sum_{ij'}\zijp(1-\ojjp)\\
& = (\ssb+\sse)\sum_{i}(\nid-1)\\
& = (\ssb+\sse)(N-R).
\end{align*}
The same argument give $\e(U_b) = (\ssa+\sse)(N-C)$.
\end{proof}

The matrix in~\eqref{eq:eu} is
\begin{align}\label{eq:defm}
M&\equiv
\begin{pmatrix}
0 & N-R&N-R\\[.5ex]
N-C&0&N-C\\[.5ex]
N^2-\sum_i\nid^2 & N^2-\sum_j\ndj^2 &N^2-N 
\end{pmatrix}.
\end{align}
Our moment based estimates are
\begin{align}\label{eq:sshat}
\begin{pmatrix}
\ssah\\[.5ex]
\ssbh\\[.5ex]
\sseh 
\end{pmatrix}
= M^{-1}
\begin{pmatrix}
U_a\\[.5ex]
U_b\\[.5ex]
U_e 
\end{pmatrix}. 
\end{align}
They are only well defined when $M$ is nonsingular.
The determinant of $M$ is
\begin{align*}
&-(N-R)
\bigl[(N-C)(N^2-N)-(N-C)(N^2-\sum_i\nid^2)
\bigr]\\
&+(N-R)\bigl[(N-C)(N^2-\sum_j\ndj^2)\bigr]\\
=\,&
-(N-R)
\bigl[(N-C)(\sum_i\nid^2-N)\bigr]
+(N-R)\bigl[(N-C)(N^2-\sum_j\ndj^2)\bigr]\\
=\,&(N-R)(N-C)[N^2-\sum_i\nid^2-\sum_j\ndj^2+N]. 
\end{align*}

The first factor is positive so long as $\max_i\nid>1$,
and the second factor requires $\max_j\ndj>1$.  We already knew
that we needed these conditions in order to have all three U-statistics
depend on the $\yij$.
It is still of interest to know when the third factor is positive.
It is sufficient that no row or column has over half 
of the data. 


\section{The variance}
From equation~\eqref{eq:sshat} we get
$$
\var
\begin{pmatrix}
\ssah\\[.5ex]
\ssbh\\[.5ex]
\sseh 
\end{pmatrix} 
=M^{-1}
\var
\begin{pmatrix}
U_a\\[.5ex]
U_b\\[.5ex]
U_e 
\end{pmatrix}
M^{-1}
$$
where $M$ is given at~\eqref{eq:defm}.
So we need the variances and covariances of the three $U$ statistics.

To find variances, 
we will work out $\e(U^2)$ for our $U$-statistics.
Those involve
\begin{align*}
&\phe\e( (\yij-\yipjp)^2(\yrs-\yrpsp)^2)\\
&=
\e\bigl(
(\ai-\aip+\bj-\bjp+\eij-\eipjp)^2
(\ar-\arp+\bs-\bsp+\ers-\erpsp)^2
\bigr)\\
&=
\e\Bigl[\bigl((\ai-\aip)^2+(\bj-\bjp)^2+(\eij-\eipjp)^2\\
&\quad+2(\ai-\aip)(\bj-\bjp)+2(\ai-\aip)(\eij-\eipjp)+2(\bj-\bjp)(\eij-\eipjp)\bigr)\\
&\quad\times\bigl((\ar-\arp)^2+(\bs-\bsp)^2+(\ers-\erpsp)^2\\
&\quad+2(\ar-\arp)(\bs-\bsp)+2(\ar-\arp)(\ers-\erpsp)+2(\bs-\bsp)(\ers-\erpsp)\bigr)\Bigr].
\end{align*}

This expression involves $8$ indices and it has
$36$ terms.  Some of those terms simplify due
to independence and some vanish due to zero means.
To shorten some expressions we use 
\begin{align*}
\baiiprrp&\equiv \e( (\ai-\aip)(\ar-\arp))\\
\daiip&\equiv \e( (\ai-\aip)^2) ,\quad\text{and},\\
\qaiiprrp&\equiv \e( (\ai-\aip)^2(\ar-\arp)^2)
\end{align*}
with mnemonics bilinear, diagonal and quartic.
There are similarly defined terms for component $B$. 
For the error term we have
\begin{align*}
\beijipjprsrpsp&\equiv \e( (\eij-\eipjp)(\ers-\erpsp))\\
\deijipjp&\equiv \e( (\eij-\eipjp)^2) ,\quad\text{and},\\
\qeijipjprsrpsp&\equiv \e( (\eij-\eipjp)^2(\ers-\erpsp)^2).
\end{align*}

The generic contribution
$\e( (\yij-\yipjp)^2(\yrs-\yrpsp)^2)$
to the mean square of a $U$-statistic equals
\begin{equation}\label{eq:meanprodsquarediff}
\begin{split}
&\qaiiprrp + \qbjjpssp + \qeijipjprsrpsp\\
&+\daiip\dbssp +\daiip\dersrpsp\\
&+\dbjjp\darrp+\dbjjp\dersrpsp\\
&+\deijipjp\darrp+\deijipjp\dbssp\\
&+4\baiiprrp\bbjjpssp +4\baiiprrp\beijipjprsrpsp+4\bbjjpssp\beijipjprsrpsp. 
\end{split}
\end{equation}
The other $24$ terms are zero.

\subsection{Variance parts}

Here we collect expressions for the quantities appearing
in the generic term of our squared $U$-statistics.

\begin{lemma}\label{lem:meanproddiff}
In the random effects model~\eqref{eq:reffect},
\begin{align*}
\baiiprrp &= \ssa\bigl(\oir-\oirp-\oipr+\oiprp\bigr),\\
\bbjjpssp &= \ssb\bigl(\ojs-\ojsp-\ojps+\ojpsp\bigr),\quad\text{and} \\
\beijipjprsrpsp &= \sse\bigl(1_{ij=rs}-1_{ij=r's'}-1_{i'j'=rs}+1_{i'j'=r's'}\bigr).
\end{align*}
\end{lemma}
\begin{proof}
The first one follows by expanding and using 
$\e(\ai\ar)=\ssa\oir$, et cetera.  The other two
use the same argument.
\end{proof}

\begin{lemma}\label{lem:meansquarediff}
In the random effects model~\eqref{eq:reffect},
\begin{align*}
\daiip &= 2\ssa(1-\oiip)\\
\dbjjp &= 2\ssb(1-\ojjp)\\
\deijipjp &= 2\sse(1-1_{ij=i'j'}).
\end{align*}
\end{lemma}
\begin{proof}
Take $i=r$ and $i'=r'$ in Lemma~\ref{lem:meanproddiff}.
\end{proof}

\begin{lemma}\label{lem:meanprodsqdiff}
In the random effects model~\eqref{eq:reffect},
\begin{align*}
\qaiiprrp
&=1_{i\ne i'}1_{r\ne r'}\fpa\Bigl(4+(\ka+2)(1_{i\in\{r,r'\}}+1_{i'\in\{r,r'\}})+4\times1_{\{i,i'\}=\{r,r'\}}\Bigr)\\
\qbjjpssp
&=1_{j\ne j'}1_{s\ne s'}\fpb\Bigl(4+(\kb+2)(1_{j\in\{s,s'\}}+1_{j'\in\{s,s'\}})+4\times1_{\{j,j'\}=\{s,s'\}}\Bigr)\\
\qeijipjprsrpsp
&=1_{ij\ne i'j'}1_{rs\ne r's'}\fpe\Bigl(4+(\ke+2)(1_{ij\in\{rs,r's'\}}+1_{i'j'\in\{rs,r's'\}})+4\times1_{\{ij,i'j'\}=\{rs,r's'\}}\Bigr).
\end{align*}
\end{lemma}
\begin{proof}
We prove the first one; the others are similar.
This quantity is $0$ if $i=i'$ or $r=r'$.
When $i\ne i'$ and $r\ne r'$,
there are $3$ cases to consider:
$|\{i,i'\}\cap\{r,r'\}|=0$,
$|\{i,i'\}\cap\{r,r'\}|=1$ and
$|\{i,i'\}\cap\{r,r'\}|=2$.
The kurtosis is defined via
$\ka = \e(a^4)/\fpa-3$, so $\e(a^4) = (\ka+3)\fpa$.

For no overlap, we find
\begin{align*}
\e( (a_1-a_2)^2(a_3-a_4)^2) & = \e((a_1-a_2)^2)^2 = 4\fpa.
\end{align*}
For a single overlap,
\begin{align*}
\e( (a_1-a_2)^2(a_1-a_3)^2) 
& = \e( (a_1^2-2a_1a_2+a_2^2) (a_1^2-2a_1a_3+a_3^2))\\
& = \e(a_1^4)+3\fpa = \fpa(\ka+6). 
\end{align*}
For a double overlap,
\begin{align*}
\e( (a_1-a_2)^4)
& = \e(a_1^4-4a_1a_2^3+6a_1^2a_2^2-4a_1^3a_2+a_2^4)\\
& = 2\e(a_1^4)+6\fpa = \fpa(2\ka+12). 
\end{align*}

As a result,
\begin{align*}
\e(  (\ai-\aip)^2(\ar-\arp)^2 ) = \begin{cases}
4\fpa,&|\{i,i'\}\cap\{r,r'\}|=0,\\
\fpa(\ka+6),&|\{i,i'\}\cap\{r,r'\}|=1,\\
\fpa(2\ka+12),&|\{i,i'\}\cap\{r,r'\}|=2,
\end{cases}
\end{align*}
and so
$\e(  (\ai-\aip)^2(\ar-\arp)^2 )$ equals
\begin{align*}
1_{i\ne i'}1_{r\ne r'}
\fpa\Bigl(4+(\ka+2)(1_{i\in\{r,r'\}}+1_{i'\in\{r,r'\}})
+4\times1_{\{i,i'\}=\{r,r'\}}\Bigr).\qquad\qedhere
\end{align*}
\end{proof}

\subsection{Variance of $U_a$}\label{sec:varua}
We will work out $\e(U_a^2)$ and then subtract
$\e(U_a)^2$. First we write
\begin{align*}
U_a^2 = \frac14\sum_{ijj'}\sum_{rss'}\nid^{-1}\nrd^{-1}
\zij\zijp\zrs\zrsp(\yij-\yijp)^2(\yrs-\yrsp)^2 .
\end{align*}
For $\e(U_a^2)$ we use the special case $i=i'$ and $r=r'$
of~\eqref{eq:meanprodsquarediff},
\begin{align*}
\e(U_a^2)
&=\frac14\sum_{ijj'}\sum_{rss'}\nid^{-1}\nrd^{-1}\zij\zijp\zrs\zrsp\Bigl[\\
&\quad\phe\qaiirr + \qbjjpssp + \qeijijprsrsp\\
&\quad+\daii\dbssp +\daii\dersrsp\\
&\quad+\dbjjp\darr+\dbjjp\dersrsp\\
&\quad+\deijijp\darr+\deijijp\dbssp\\
&\quad+4\baiirr\bbjjpssp +4\baiirr\beijijprsrsp+4\bbjjpssp\beijijprsrsp\Bigr]\\
&=\frac14\sum_{ijj'}\sum_{rss'}\nid^{-1}\nrd^{-1}\zij\zijp\zrs\zrsp\Bigl[
\underbrace{\qbjjpssp}_{\text{Term 1}}
+ \underbrace{\qeijijprsrsp}_{\text{Term 2}}\\
&+\underbrace{\dbjjp\dersrsp}_{\text{Term 3}}
+\underbrace{\deijijp\dbssp}_{\text{Term 4}}
+\underbrace{4\bbjjpssp\beijijprsrsp}_{\text{Term 5}}\Bigr]
\end{align*} 
after eliminating terms that are always $0$.
We handle these five sums in the next subsubsections.

\subsubsection{$U_a^2$ term 1}
\begin{align*}
&\phe\frac14\sum_{ijj'}\sum_{rss'}\nid^{-1}\nrd^{-1}\zij\zijp\zrs\zrsp\qbjjpssp\\
&=\frac\fpb4\sum_{ijj'}\sum_{rss'}\nid^{-1}\nrd^{-1}\zij\zijp\zrs\zrsp1_{j\ne j'}1_{s\ne s'}\\
&\phe\Bigl(4+(\kb+2)(1_{j\in\{s,s'\}}+1_{j'\in\{s,s'\}})+4\times1_{\{j,j'\}=\{s,s'\}}\Bigr)\\
&=\frac\fpb4\sum_{ijj'}\sum_{rss'}\nid^{-1}\nrd^{-1}\zij\zijp\zrs\zrsp(1-\ojjp)(1-\ossp) \\
&\phe
\Bigl(\underbrace{4}_{1.1}
+\underbrace{(\kb+2)(1_{j\in\{s,s'\}}+1_{j'\in\{s,s'\}})}_{\text{1.2 and 1.3}}
+\underbrace{4\times1_{\{j,j'\}=\{s,s'\}}}_{\text{1.4}}\Bigr).
\end{align*}

Term 1 is now a sum of four terms, 1.1 through 1.4.  Term 1.1 is $\fpb$ times
\begin{align*}
&
\phe\frac14\sum_{ijj'}\sum_{rss'}\nid^{-1}\nrd^{-1}\zij\zijp\zrs\zrsp 
4(1-\ojjp-\ossp+\ojjp\ossp)\\
&=\sum_{ijj'}\sum_{rss'}\nid^{-1}\nrd^{-1}\zij\zijp\zrs\zrsp\\
&\phe-\sum_{ij}\sum_{rss'}\nid^{-1}\nrd^{-1}\zij\zrs\zrsp\\
&\phe-\sum_{ijj'}\sum_{rs}\nid^{-1}\nrd^{-1}\zij\zijp\zrs\\
&\phe+\sum_{ij}\sum_{rs}\nid^{-1}\nrd^{-1}\zij\zrs\\
&=\sum_{ir}(\nid\nrd-\nrd-\nid+1)\\
&=(N-R)^2.
\end{align*}

Term 1.2 is $\fpb(\kb+2)/4$ times
\begin{align*}
&\phe\sum_{ijj'}\sum_{rss'}\nid^{-1}\nrd^{-1}\zij\zijp\zrs\zrsp
(1-\ojjp-\ossp+\ojjp\ossp)1_{j\in\{s,s'\}}\\
&=\sum_{ijj'}\sum_{rss'}\nid^{-1}\nrd^{-1}\zij\zijp\zrs\zrsp (\ojs+\ojsp-\ojs\ojsp)\\
&\phe-\sum_{ij}\sum_{rss'}\nid^{-1}\nrd^{-1}\zij\zrs\zrsp (\ojs+\ojsp-\ojs\ojsp)\\ 
&\phe-\sum_{ijj'}\sum_{rs}\nid^{-1}\nrd^{-1}\zij\zijp\zrs \ojs\\
&\phe+\sum_{ij}\sum_{rs}\nid^{-1}\nrd^{-1}\zij\zrs \ojs\\
&=2\sum_{ir}\zzt_{ir}-\sum_{ir}\nrd^{-1}(ZZ^\tran)_{ir}\\
&\phe-2\sum_{ir}\nid^{-1}(ZZ^\tran)_{ir}  +\sum_{ir}\nid^{-1}\nrd^{-1}(ZZ^\tran)_{ir}\\
&\phe-\sum_{ir}\nrd^{-1}(ZZ^\tran)_{ir} + \sum_{ir}\nid^{-1}\nrd^{-1}(ZZ^\tran)_{ir}\\
&=
2\sum_{ir}\zzt_{ir}(1-\nid^{-1})(1-\nrd^{-1}).
\end{align*}
The expression $\sum_{ir}(ZZ^\tran)_{ir}$ simplifies to $\sum_j\ndj^2$,
changing it from a `row quantity' to a `column quantity'.
But the other parts of this expression are equivalent to
sums of terms such as $\nid^{-1}\zij\ndj$ making the column version 
less convenient to work with.
Term 1.3 is the same as term 1.2 by symmetry of indices.

Term 1.4 is $\fpb$ times
\begin{align*}
&\phe\sum_{ijj'}\sum_{rss'}\nid^{-1}\nrd^{-1}\zij\zijp\zrs\zrsp
(1-\ojjp)(1-\ossp) 1_{\{j,j'\}=\{s,s'\}}\\
&=\sum_{ijj'}\sum_{rss'}\nid^{-1}\nrd^{-1}\zij\zijp\zrs\zrsp
1_{j\ne j'}1_{s\ne s'} 1_{\{j,j'\}=\{s,s'\}}\\
&=2\sum_{ijj'}\sum_{rss'}\nid^{-1}\nrd^{-1}\zij\zijp\zrs\zrsp
1_{j\ne j'}1_{s\ne s'} \ojs\ojpsp\\
&=2\sum_{ijj'}\sum_{r}\nid^{-1}\nrd^{-1}\zij\zijp\zrj\zrjp1_{j\ne j'} \\
&=2\sum_{ijj'}\sum_{r}\nid^{-1}\nrd^{-1}\zij\zijp\zrj\zrjp -2\sum_{ij}\sum_{r}\nid^{-1}\nrd^{-1}\zij\zrj\\
&=
2\sum_{ir}\nid^{-1}\nrd^{-1}\zzt_{ir}^2 -2\sum_{ir}\nid^{-1}\nrd^{-1}\zzt_{ir}.
\end{align*}

Summing terms 1.1 through 1.4 yields
\begin{align*}
&\phe\fpb\Bigl(
(N-R)^2+(\kb+2)\sum_{ir}\zzt_{ir}(1-\nid^{-1})(1-\nrd^{-1} ) 
\\
&\phe+
2\sum_{ir}\nid^{-1}\nrd^{-1}\zzt_{ir}(\zzt_{ir}-1) 
\Bigr). 
\end{align*}

\subsubsection{$U_a^2$ term 2}

\begin{align*}
&\phe\frac14\sum_{ijj'}\sum_{rss'}\nid^{-1}\nrd^{-1}\zij\zijp\zrs\zrsp\qeijijprsrsp\\
&=
\frac\fpe4\sum_{ijj'}\sum_{rss'}\nid^{-1}\nrd^{-1}\zij\zijp\zrs\zrsp 
1_{ij\ne ij'}1_{rs\ne rs'}\\
&\phe\times\Bigl(4 
+{(\ke+2)(1_{ij\in\{rs,rs'\}}+1_{ij'\in\{rs,rs'\}})}
+{4\times1_{\{ij,ij'\}=\{rs,rs'\}}}\Bigr)\\
&=
\frac\fpe4\sum_{ijj'}\sum_{rss'}\nid^{-1}\nrd^{-1}\zij\zijp\zrs\zrsp 
1_{j\ne j'}1_{s\ne s'}\\
&\phe\times\Bigl(\underbrace{4}_{2.1}
+\underbrace{(\ke+2)\oir(1_{j\in\{s,s'\}}+1_{j'\in\{s,s'\}})}_{\text{2.2 and 2.3}}
+\underbrace{4\oir1_{\{j,j'\}=\{s,s'\}}}_{2.4}\Bigr). 
\end{align*}

Term 2.1 is $\fpe$ times
\begin{align*}
&\phe\sum_{ijj'}\sum_{rss'}\nid^{-1}\nrd^{-1}\zij\zijp\zrs\zrsp 1_{j\ne j'}1_{s\ne s'}
=(N-R)^2
\end{align*}
by the same process that evaluated term 1.1.

Term 2.2 is $\fpe(\ke+2)/4$ times
\begin{align*}
&\phe\sum_{ijj'}\sum_{rss'}\nid^{-1}\nrd^{-1}\zij\zijp\zrs\zrsp 1_{j\ne j'}1_{s\ne s'}\oir1_{j\in\{s,s'\}}\\
&=\sum_{ijj'}\sum_{ss'}\nid^{-2}\zij\zijp\zis\zisp 1_{j\in\{s,s'\}}\\
&\phe-\sum_{ijj'}\sum_{ss'}\nid^{-2}\zij\zijp\zis\zisp \ojjp1_{j\in\{s,s'\}}\\
&\phe-\sum_{ijj'}\sum_{ss'}\nid^{-2}\zij\zijp\zis\zisp \ossp1_{j\in\{s,s'\}}\\
&\phe+\sum_{ijj'}\sum_{ss'}\nid^{-2}\zij\zijp\zis\zisp \ojjp\ossp1_{j\in\{s,s'\}}\\
\end{align*}
which reduces to
\begin{align*}
&\phe\sum_{ijj'}\sum_{ss'}\nid^{-2}\zij\zijp\zis\zisp 1_{j\in\{s,s'\}}\\
&\phe-\sum_{ijj'}\sum_{ss'}\nid^{-2}\zij\zijp\zis\zisp \ojjp1_{j\in\{s,s'\}}\\
&\phe-\sum_{ijj'}\sum_{ss'}\nid^{-2}\zij\zijp\zis\zisp \ossp\ojs\\
&\phe+\sum_{ijj'}\sum_{ss'}\nid^{-2}\zij\zijp\zis\zisp \ojjp\ojs\ojsp\\
&=\sum_{ijj'}\sum_{ss'}\nid^{-2}\zij\zijp\zis\zisp (\ojs+\ojsp-\ojs\ojsp)\\  
&\phe-\sum_{ijj'}\sum_{ss'}\nid^{-2}\zij\zijp\zis\zisp \ojjp (\ojs+\ojsp-\ojs\ojsp)\\   
&\phe-\sum_{ijj'}\sum_{ss'}\nid^{-2}\zij\zijp\zis\zisp \ossp\ojs\\
&\phe+\sum_{ijj'}\sum_{ss'}\nid^{-2}\zij\zijp\zis\zisp \ojjp\ojs\ojsp\\
&=2\sum_{ijj'}\sum_{s}\nid^{-2}\zij\zijp\zis -\sum_{ijj'}\nid^{-2}\zij\zijp\\
&\phe-2\sum_{ij}\sum_{s}\nid^{-2}\zij\zis +\sum_{ij}\nid^{-2}\zij\\
&\phe-\sum_{ijj'}\nid^{-2}\zij\zijp+\sum_{ij}\nid^{-2}\zij \\
&=2\sum_i\nid -R-2R +\sum_{i}\nid^{-1}-R+\sum_i\nid^{-1}\\
&=2N-4R+2\sum_i \nid^{-1}\\
&=2\sum_i\nid(1-\nid^{-1})^2.
\end{align*}
The last expression resembles the diagonal part of term 1.2.
Term 2.3 is the same is the same as term 2.2.

Term 2.4 is $\fpe$ times
\begin{align*}
&\phe\sum_{ijj'}\sum_{rss'}\nid^{-1}\nrd^{-1}\zij\zijp\zrs\zrsp\oir1_{j\ne  j}1_{s\ne s'}1_{\{j,j'\}=\{s,s'\}}\\
\end{align*}
This is the same sum as the coefficient in term 1.4 has except that it
has the additional constraint $i=r$. Imposing $i=r$ on that quantity yields
\begin{align*}
2\sum_i\nid^{-2}(ZZ^\tran)_{ii}^2-2\sum_i\nid^{-2}(ZZ^\tran)_{ii}
&=2\sum_i(1-\nid^{-1}).
\end{align*}

Term 2 is thus
\begin{align*}
&\fpe\Bigl( (N-R)^2 
+(\ke+2)\sum_i\nid(1-\nid^{-1})^2+2\sum_{i}(1-\nid^{-1})\Bigr). 
\end{align*}

\subsubsection{$U_a^2$ terms 3 and 4}
These terms are equal by symmetry.  We evaluate term 3.
\begin{align*}
&\phe\frac14\sum_{ijj'}\sum_{rss'}\nid^{-1}\nrd^{-1}\zij\zijp\zrs\zrsp\dbjjp\dersrsp\\
&=\frac14 
\biggl(\sum_{ijj'}\nid^{-1}\zij\zijp\dbjjp\biggr) 
\biggl(
\sum_{rss'}\nrd^{-1}\zrs\zrsp\dersrsp 
\biggr). 
\end{align*}

Now 
\begin{align*}
\sum_{ijj'}\nid^{-1}\zij\zijp\dbjjp 
&=2\ssb\sum_{ijj'}\nid^{-1}\zij\zijp(1-\ojjp)\\
&=2\ssb\sum_i(\nid-1) = 2\ssb(N-R) 
\end{align*}
and 
\begin{align*}
\sum_{rss'}\nrd^{-1}\zrs\zrsp\dersrsp 
&=2\sse\sum_{rss'}\nrd^{-1}\zrs\zrsp(1-\ossp)\\
&=2\sse(N-R) 
\end{align*}
by the same steps. Therefore term 3 of $\e(U_a^2)$ equals 
$\ssb\sse(N-R)^2$ and the sum of terms 3 and 4 is
$2\ssb\sse(N-R)^2$.

\subsubsection{$U_a^2$ term 5}
The term equals
\begin{align*}
&\phe\sum_{ijj'}\sum_{rss'}\nid^{-1}\nrd^{-1}\zij\zijp\zrs\zrsp\bbjjpssp\beijijprsrsp\\
&=\sum_{ijj'}\sum_{ss'}\nid^{-1}\zij\zijp\bbjjpssp\sum_r \nrd^{-1}\zrs\zrsp\beijijprsrsp.
\end{align*}

Now
\begin{align*}
\sum_r\nrd^{-1}\zrs\zrsp\beijijprsrsp
&= \sse\sum_r\nrd^{-1}\zrs\zrsp\bigl(1_{ij=rs}-1_{ij=rs'}-1_{ij'=rs}+1_{ij'=rs'}\bigr)\\
&= \sse\nid^{-1}\zis\zisp\bigl(1_{j=s}-1_{j=s'}-1_{j'=s}+1_{j'=s'}\bigr).
\end{align*}
Term 5 is then 
\begin{align*}
&\phe\sse\sum_{ijj'}\sum_{ss'}\nid^{-2}\zij\zijp
\zis\zisp\bigl(1_{j=s}-1_{j=s'}-1_{j'=s}+1_{j'=s'}\bigr)\bbjjpssp\\
&=\sse\ssb\sum_{ijj'}\sum_{ss'}\nid^{-2}\zij\zijp\zis\zisp
\bigl(\ojs-\ojsp-\ojps+\ojpsp\bigr)^2\\
&=\sse\ssb\sum_{ijj'}\sum_{ss'}\nid^{-2}\zij\zijp\zis\zisp
\ojs\bigl(\ojs-\ojsp-\ojps+\ojpsp\bigr)\\
&\phe-\sse\ssb\sum_{ijj'}\sum_{ss'}\nid^{-2}\zij\zijp\zis\zisp
\ojsp\bigl(\ojs-\ojsp-\ojps+\ojpsp\bigr)\\
&\phe-\sse\ssb\sum_{ijj'}\sum_{ss'}\nid^{-2}\zij\zijp\zis\zisp
\ojps\bigl(\ojs-\ojsp-\ojps+\ojpsp\bigr)\\
&\phe+\sse\ssb\sum_{ijj'}\sum_{ss'}\nid^{-2}\zij\zijp\zis\zisp
\ojpsp\bigl(\ojs-\ojsp-\ojps+\ojpsp\bigr),
\end{align*}
which we call terms 5.1, 5.2, 5.3 and 5.4.
Next we find the coefficients of $\ssb\sse$ in these four terms.

For term 5.1, we get
\begin{align*}
&
\sum_{ijj'}\sum_{ss'}\nid^{-2}\zij\zijp\zis\zisp
\ojs\bigl(\ojs-\ojsp-\ojps+\ojpsp\bigr)\\
&=
\sum_{ijj'}\sum_{s'}\nid^{-2}\zij\zijp\zisp
\bigl(1-\ojsp-\ojjp+\ojpsp\bigr)\\
&=\sum_i(\nid-1)\\
&=N-R.
\end{align*}
For term 5.2, we get
\begin{align*}
&
-\sum_{ijj'}\sum_{ss'}\nid^{-2}\zij\zijp\zis\zisp 
\ojsp\bigl(\ojs-\ojsp-\ojps+\ojpsp\bigr)\\
&=
-\sum_{ijj'}\sum_{s}\nid^{-2}\zij\zijp\zis 
\bigl(\ojs-1-\ojps+\ojjp\bigr)\\
&=N-R
\end{align*}
as well.
Terms 5.3 and 5.4 are also $\nid-1$, by the steps used
for terms 5.2 and 5.1 respectively.
As a result term 5 equals $4\ssb\sse(N-R)$.

\subsection{Combination}

Combining the results of the previous sections, we have
\begin{align*}
\e(U_a^2)&=
\fpb\Bigl((N-R)^2+(\kb+2)\sum_{ir}\zzt_{ir}(1-\nid^{-1})(1-\nrd^{-1} ) \\
&\phe+
2\sum_{ir}\nid^{-1}\nrd^{-1}\zzt_{ir}(\zzt_{ir}-1) 
\Bigr)\\
&\phe+2\ssb\sse(N-R)^2 + 4\ssb\sse(N-R)\\
&\phe+\fpe\Bigl( (N-R)^2 
+(\ke+2)\sum_i\nid(1-\nid^{-1})^2+2\sum_{i}(1-\nid^{-1})\Bigr). 
\end{align*}
Subtracting $\e(U_a)^2=(N-R)^2(\ssb+\sse)^2$ we find
\begin{align}\label{eq:varua}
\begin{split}
\var(U_a) 
&=\fpb\Bigl(
(\kb+2)\sum_{ir}\zzt_{ir}(1-\nid^{-1})(1-\nrd^{-1} ) 
+2\sum_{ir}\nid^{-1}\nrd^{-1}\zzt_{ir}(\zzt_{ir}-1) \Bigr)\\
&\phe+4\ssb\sse(N-R) 
+\fpe\Bigl(
(\ke+2)\sum_i\nid(1-\nid^{-1})^2+2\sum_{i}(1-\nid^{-1})\Bigr). 
\end{split}
\end{align}
 
\subsection{Checks}
We can check some special cases of this formula.

\subsubsection{Rows nested in columns}
If for instance rows are nested within columns, then 
$N=R$, and all
$\nid=\nrd=1$ and in this case $U_a=0$.
The above formula gives $\var(U_a)=0$ for this case.

\subsubsection{Columns nested in rows}
If columns are nested in rows, then $\zzt_{ir}=1_{i=r}\nid$ and 
equation~\eqref{eq:varua} yields
\begin{align}
\var(U_a)
&=\fpb\Bigl(
(\kb+2)\sum_{ir}\nid\oir(1-\nid^{-1})(1-\nrd^{-1} ) 
+2\sum_{ir}\nid^{-1}\nrd^{-1}\oir\nid(\nid-1) \Bigr) \notag\\
&\phe+4 \ssb\sse (N-R)+\fpe\Bigl(
(\ke+2)\sum_i\nid(1-\nid^{-1})^2+2\sum_{i}(1-\nid^{-1})\Bigr) \notag\\
&=
\Bigl(\fpb (\kb+2)+\fpe (\ke+2)\Bigr) \sum_{i}\nid(1-\nid^{-1})^2 
+2(\fpb+\fpe) \sum_{i}(1-\nid^{-1}) +4\ssb\sse(N-R) \notag\\
&=
(\kb\fpb+\ke\fpe)  \sum_{i}\nid(1-\nid^{-1})^2 
+(\fpb+\fpe)\sum_i\Bigl(
2\nid(1-\nid^{-1})^2  + 2(1-\nid^{-1})\Bigr) 
+4\ssb\sse(N-R) \notag\\
&=
(\kb\fpb+\ke\fpe)  \sum_{i}\nid(1-\nid^{-1})^2 
+2(\fpb+\fpe)\sum_i(\nid-1)
+4\ssb\sse(N-R) \notag\\
&=
(\kb\fpb+\ke\fpe)  \sum_{i}\nid(1-\nid^{-1})^2 +2(N-R)(\ssb+\sse)^2.\label{eq:nestfromvarua}
\end{align}

When columns are nested in rows, then
$U_a = \sum_i (\nid-1)s^2_{i\sumdot}$
and because the rows are then independent, $U_a$ has variance
$$
(\ssb+\sse)^2
\sum_i (\nid-1)^2\Bigl(
\frac{2}{\nid-1}
+\frac{\kappa(b_1+e_{11})}{\nid}
\Bigr).
$$
The kurtosis of $b_j+e_{ij}$ is
$$
\kappa_{B+E} = 
\kb\Bigl( \frac{\sigma^2_B}{\sigma^2_B+\sigma^2_E}\Bigr)^2
+\ke\Bigl( \frac{\sigma^2_E}{\sigma^2_B+\sigma^2_E}\Bigr)^2.
$$
Therefore for columns nested in rows
\begin{align*}
\var(U_a)&=
(\ssb+\sse)^2\sum_i\Bigl(
2(\nid-1)
+\frac{(\nid-1)^2}{\nid}
\Bigl(\kb\Bigl( \frac{\sigma^2_B}{\sigma^2_B+\sigma^2_E}\Bigr)^2
+\ke\Bigl( \frac{\sigma^2_E}{\sigma^2_B+\sigma^2_E}\Bigr)^2\Bigr)\Bigr)\\
&=
2(\ssb+\sse)^2\sum_i(\nid-1) 
+(\kb\fpb+\ke\fpe)\sum_i\nid(1-\nid^{-1})^2 \\
&=
2(N-R)(\ssb+\sse)^2
+(\kb\fpb+\ke\fpe)\sum_i\nid(1-\nid^{-1})^2,
\end{align*}
matching the expression~\eqref{eq:nestfromvarua} that comes from equation~\eqref{eq:varua}
for $\var(U_a)$.

\subsubsection{$\ssb=0$}
If $\ssb=0$ then $\var(U_a)$ should be the same as it is for columns nested
in rows.  In this case
equation~\eqref{eq:varua} reduces to
\begin{align*} 
\begin{split}
\var(U_a) 
&=\fpe\Bigl((\ke+2)\sum_i\nid(1-\nid^{-1})^2+2\sum_{i}(1-\nid^{-1})\Bigr)\\
&=\fpe\Bigl(
\ke\sum_i\nid(1-\nid^{-1})^2+2\sum_{i}\Bigl(
\nid(1-\nid^{-1})^2+(1-\nid^{-1})\Bigr)\\
&=\fpe\Bigl(
\ke\sum_i\nid(1-\nid^{-1})^2+2(N-R)\Bigr).
\end{split}
\end{align*}
If instead we first take the columns nested in rows special case from
equation~\eqref{eq:nestfromvarua} and then substitute $\ssb=0$, we get the same expression.

\subsubsection{$\sse=0$ and $\kb=-2$}

In this special case we take $\sse=0$ and take $\bj\sim U(\pm1)$.
Then $\ssb=1$ and $\kb=-2$.
Then
\begin{align*}
\var(U_a)  &=2\fpb\sum_{ir}\nid^{-1}\nrd^{-1}\zzt_{ir}(\zzt_{ir}-1)
=2\sum_{ir}\nid^{-1}\nrd^{-1}\zzt_{ir}(\zzt_{ir}-1).
\end{align*}
In this case
\begin{align*}
U_a 
& = \frac12\sum_{ijj'}\nid^{-1}\zij\zijp(\bj-\bjp)^2\\
& = \frac12\sum_{ijj'}\nid^{-1}\zij\zijp(\bj^2-2\bj\bjp+\bjp^2)\\
& = \sum_{ijj'}\nid^{-1}\zij\zijp(1-\bj\bjp)\\
& = \sum_i\nid-\sum_{ijj'}\nid^{-1}\zij\zijp\bj\bjp
\end{align*}
and so $\var(U_a)=\var(\wt U_a)$ where $\wt U_a=\sum_{ijj'}\nid^{-1}\zij\zijp\bj\bjp$.
We easily find that 
$$\e(\wt U_a) = \sum_{ijj'}\nid^{-1}\zij\zijp\ojjp=\sum_i 1=R.$$

To get the variance of $\wt U_a$ we need 
$$\e(\bj\bjp\bs\bsp)
=\ojjp\ossp + \ojs\ojpsp + \ojsp\ojps -2\times\ojsp\ojpsp\ossp.
$$
Now 
\begin{align*}
\e(\wt U_a^2) 
& = \sum_{ijj'}\sum_{rss'}\nid^{-1}\nrd^{-1}\zij\zijp\zrs\zrsp \e(\bj\bjp\bs\bsp)\\
& = \sum_{ijj'}\sum_{rss'}\nid^{-1}\nrd^{-1}\zij\zijp\zrs\zrsp \ojjp\ossp\\
&\phe+\sum_{ijj'}\sum_{rss'}\nid^{-1}\nrd^{-1}\zij\zijp\zrs\zrsp \ojs\ojpsp\\
&\phe+\sum_{ijj'}\sum_{rss'}\nid^{-1}\nrd^{-1}\zij\zijp\zrs\zrsp \ojsp\ojps\\
&\phe-2\sum_{ijj'}\sum_{rss'}\nid^{-1}\nrd^{-1}\zij\zijp\zrs\zrsp \ojsp\ojpsp\ossp\\
& = \sum_{ij}\sum_{rs}\nid^{-1}\nrd^{-1}\zij\zrs
+\sum_{ijj'}\sum_{r}\nid^{-1}\nrd^{-1}\zij\zijp\zrj\zrjp \\
&\phe+\sum_{ijj'}\sum_{r}\nid^{-1}\nrd^{-1}\zij\zijp\zrjp\zrj -2\sum_{ij}\sum_{r}\nid^{-1}\nrd^{-1}
\zij\zrj\\
& = \sum_{i}\sum_{r}1
+2\sum_{i}\sum_{r}\nid^{-1}\nrd^{-1}\zzt_{ir}^2
 -2\sum_{i}\sum_{r}\nid^{-1}\nrd^{-1}\zzt_{ir}\\
& = R^2 +2\sum_i\sum_r \nid^{-1}\nrd^{-1}\zzt_{ir}(\zzt_{ir}-1).
\end{align*}
In this case we get
\begin{align*}
\var(\wt U_a) = R^2 +2\sum_i\sum_r\zzt_{ir}(\zzt_{ir}-1) -R^2
\end{align*}
matching the result from~\eqref{eq:varua}.

\subsubsection{Crossed design}
In a crossed design $\nid=C$ for all $i$ and $\zzt_{ir}=C$ for all $i$ and $r$.
Here the variance is
\begin{align}\label{eq:varuacross}
\begin{split}
\var(U_a) 
&=\fpb\Bigl(
(\kb+2)\sum_{ir} C(1-C^{-1})^2 
+2\sum_{ir} C^{-1} (C-1)\Bigr)\\
&\phe+\fpe\Bigl(
(\ke+2)\sum_i C(1-C^{-1})^2+2\sum_{i}(1-C^{-1})\Bigr)+4\ssb\sse(N-R) \\
&=\fpb\Bigl(
(\kb+2)C(1-C^{-1})^2 +2(1-C^{-1})\Bigr)R^2\\
&\phe+\fpe\Bigl(
(\ke+2)C(1-C^{-1})^2+2(1-C^{-1})\Bigr)R+4\ssb\sse(R-1)C.
\end{split}
\end{align}

\subsection{Variance of $U_b$}

This case is exactly symmetric to the one above
with $\var(U_a)$ given by~\eqref{eq:varua}.
Therefore
\begin{align}\label{eq:varub}
\begin{split}
\var(U_b) &=\fpb\Bigl(
(\ka+2)\sum_{js}\ztz_{js}(1-\ndj^{-1})(1-\nds^{-1} ) 
+2\sum_{js}\ndj^{-1}\nds^{-1}\ztz_{js}(\ztz_{js}-1) \Bigr)\\
&\phe+\fpe\Bigl(
(\ke+2)\sum_j\ndj(1-\ndj^{-1})^2+2\sum_{j}(1-\ndj^{-1})\Bigr). 
\end{split}
\end{align}

\subsection{Variance of $U_e$}\label{sec:varue}

As before, we find $\e(U_e^2)$ and then subtract $\e(U_e)^2$.
Now
\begin{align*}
U_e^2 = \frac14\sum_{ii'jj'}\sum_{rr'ss'}
\zij\zipjp\zrs\zrpsp(\yij-\yipjp)^2(\yrs-\yrpsp)^2 .
\end{align*}
From~\eqref{eq:meanprodsquarediff},
\begin{align*}
\e(U_e^2)
&=\frac14\sum_{ii'jj'}\sum_{rr'ss'}\zij\zipjp\zrs\zrpsp\Bigl[\underbrace{\qaiiprrp}_{\text{Term 1}} + \underbrace{\qbjjpssp}_{\text{Term 2}} + \underbrace{\qeijipjprsrpsp}_{\text{Term 3}}\\
&+\underbrace{\daiip\dbssp}_{\text{Term 4}} +\underbrace{\daiip\dersrpsp}_{\text{Term 5}}+\underbrace{\dbjjp\darrp}_{\text{Term 6}}+\underbrace{\dbjjp\dersrpsp}_{\text{Term 7}}\\
&+\underbrace{\deijipjp\darrp}_{\text{Term 8}}+\underbrace{\deijipjp\dbssp}_{\text{Term 9}}\\
&+\underbrace{4\baiiprrp\bbjjpssp}_{\text{Term 10}} +\underbrace{4\baiiprrp\beijipjprsrpsp}_{\text{Term 11}}+\underbrace{4\bbjjpssp\beijipjprsrpsp}_{\text{Term 12}}\Bigr].
\end{align*} 
We handle the twelve sums in the next subsections.

\subsubsection{$U_e^2$ Term 1}

As before, we split term 1 into four parts.
\begin{align*}
&\frac14\sum_{ii'jj'}\sum_{rr'ss'}\zij\zipjp\zrs\zrpsp\qaiiprrp \\
&=\frac14\sum_{ii'jj'}\sum_{rr'ss'}\zij\zipjp\zrs\zrpsp 1_{i\ne i'}1_{r\ne r'}\fpa\Bigl(4+(\ka+2)(1_{i\in\{r,r'\}}+1_{i'\in\{r,r'\}})+4\times1_{\{i,i'\}=\{r,r'\}}\Bigr)\\
&=\frac\fpa4\sum_{ii'jj'}\sum_{rr'ss'}\zij\zipjp\zrs\zrpsp (1-1_{i=i'})(1-1_{r=r'})\\
&\phe\Bigl(\underbrace{4}_{1.1}+\underbrace{(\ka+2)(1_{i\in\{r,r'\}}+1_{i'\in\{r,r'\}})}_{1.2\text{ and }1.3}+\underbrace{4\times1_{\{i,i'\}=\{r,r'\}}}_{1.4}\Bigr).
\end{align*}

For term 1.1, we have $\fpa$ times
\begin{align*}
&\sum_{ii'jj'}\sum_{rr'ss'}\zij\zipjp\zrs\zrpsp (1-1_{i=i'}-1_{r=r'}+1_{i=i'}1_{r=r'}) \\
&=\sum_{ii'jj'}\sum_{rr'ss'}\zij\zipjp\zrs\zrpsp-\sum_{ijj'}\sum_{rr'ss'}\zij\zijp\zrs\zrpsp-\sum_{ii'jj'}\sum_{rss'}\zij\zipjp\zrs\zrsp+\sum_{ijj'}\sum_{rss'}\zij\zijp\zrs\zrsp \\
&=N^4-N^2\sum_{i}\nid^2-N^2\sum_{r}\nrd^2+\Bigl(\sum_{i}\nid^2\Bigr)\Bigl(\sum_{r}\nid^2\Bigr) \\
&=\Bigl(N^2-\sum_{i}\nid^2\Bigr)^2.
\end{align*}

Term 1.2 is $\fpa(\ka+2)/4$ times
\begin{align*}
&\sum_{ii'jj'}\sum_{rr'ss'}\zij\zipjp\zrs\zrpsp (1-1_{i=i'}-1_{r=r'}+1_{i=i'}1_{r=r'})1_{i\in\{r,r'\}} \\
&=\sum_{ii'jj'}\sum_{rr'ss'}\zij\zipjp\zrs\zrpsp (1_{i=r}+1_{i=r'}-1_{i=r}1_{i=r'})\\
&\phe-\sum_{ijj'}\sum_{rr'ss'}\zij\zijp\zrs\zrpsp (1_{i=r}+1_{i=r'}-1_{i=r}1_{i=r'}) \\
&\phe-\sum_{ii'jj'}\sum_{rss'}\zij\zipjp\zrs\zrsp 1_{i=r}\\
&\phe+\sum_{ijj'}\sum_{rss'}\zij\zijp\zrs\zrsp 1_{i=r}\\
&=\sum_{ii'jj'}\sum_{r'ss'}\zij\zipjp\zis\zrpsp+\sum_{ii'jj'}\sum_{rss'}\zij\zipjp\zrs\zisp-\sum_{ii'jj'}\sum_{ss'}\zij\zipjp\zis\zisp \\
&\phe-\sum_{ijj'}\sum_{r'ss'}\zij\zijp\zis\zrpsp-\sum_{ijj'}\sum_{rss'}\zij\zijp\zrs\zisp+\sum_{ijj'}\sum_{ss'}\zij\zijp\zis\zisp \\
&\phe-\sum_{ii'jj'}\sum_{ss'}\zij\zipjp\zis\zisp+\sum_{ijj'}\sum_{ss'}\zij\zijp\zis\zisp\\
&=N^2\sum_i \nid^2+N^2\sum_i \nid^2-N\sum_i \nid^3 \\
&\phe-N\sum_i \nid^3-N\sum_i \nid^3+\sum_i \nid^4 \\
&\phe-N\sum_i \nid^3+\sum_i \nid^4\\
&=2N^2\sum_i \nid^2-4N\sum_i \nid^3+2\sum_i \nid^4.
\end{align*}
By symmetry of indices, term 1.3 is the same as term 1.2.

For term 1.4, we have $\fpa$ times
\begin{align*}
&\sum_{ii'jj'}\sum_{rr'ss'}\zij\zipjp\zrs\zrpsp 1_{i \ne i'}1_{r \ne r'} 1_{\{i,i'\}=\{r,r'\}}\\
&=2\sum_{ii'jj'}\sum_{rr'ss'}\zij\zipjp\zrs\zrpsp 1_{i \ne i'}1_{r \ne r'}1_{i=r}1_{i'=r'}\\
&=2\sum_{ii'jj'}\sum_{ss'}\zij\zipjp\zis\zipsp 1_{i \ne i'}\\
&=2\sum_{ii'}\nid\nipd\nid\nipd (1-1_{i=i'})\\
&=2\Bigl(\sum_i \nid^2\Bigr)^2-2\sum_i \nid^4.
\end{align*}

Summing terms 1.1 to 1.4 gives
\begin{align*}
&\fpa\Bigl( N^4-2N^2\sum_i \nid^2+3\Bigl(\sum_i \nid^2\Bigr)^2-2\sum_i \nid^4\Bigr) \\
&+\fpa(\ka+2)\Bigl(N^2\sum_i \nid^2-2N\sum_i \nid^3+\sum_i \nid^4\Bigr).
\end{align*}

\subsubsection{$U_e^2$ Term 2}

We can use the symmetry of the roles of $A$ and $B$ and their indices. Therefore, term 2 is equal to
\begin{align*}
&\fpb\Bigl( N^4-2N^2\sum_j \ndj^2+3\Bigl(\sum_j \ndj^2\Bigr)^2-2\sum_j \ndj^4\Bigr) \\
&+\fpb(\kb+2)\Bigl(N^2\sum_j \ndj^2-2N\sum_j \ndj^3+\sum_j \ndj^4\Bigr).
\end{align*}

\subsubsection{$U_e^2$ Term 3}

As before, we split term 3 into four parts.
\begin{align*}
&\frac14\sum_{ii'jj'}\sum_{rr'ss'}\zij\zipjp\zrs\zrpsp\qeijipjprsrpsp \\
&=\frac14\sum_{ii'jj'}\sum_{rr'ss'}\zij\zipjp\zrs\zrpsp 1_{ij\ne i'j'}1_{rs\ne r's'}\fpe\Bigl(4+(\ke+2)(1_{ij\in\{rs,r's'\}}+1_{i'j'\in\{rs,r's'\}})+4\times1_{\{ij,i'j'\}=\{rs,r's'\}}\Bigr)\\
&=\frac\fpe4\sum_{ii'jj'}\sum_{rr'ss'}\zij\zipjp\zrs\zrpsp (1-1_{ij=i'j'})(1-1_{rs=r's'})\\
&\phe\Bigl(\underbrace{4}_{3.1}+\underbrace{(\ke+2)(1_{ij\in\{rs,r's'\}}+1_{i'j'\in\{rs,r's'\}})}_{3.2\text{ and }3.3}+\underbrace{4\times1_{\{ij,i'j'\}=\{rs,r's'\}}}_{3.4}\Bigr).
\end{align*}

For term 3.1, we have $\fpe$ times
\begin{align*}
&\sum_{ii'jj'}\sum_{rr'ss'}\zij\zipjp\zrs\zrpsp (1-1_{ij=i'j'}-1_{rs=r's'}+1_{ij=i'j'}1_{rs=r's'}) \\
&=\sum_{ii'jj'}\sum_{rr'ss'}\zij\zipjp\zrs\zrpsp-\sum_{ij}\sum_{rr'ss'}\zij\zrs\zrpsp-\sum_{ii'jj'}\sum_{rs}\zij\zipjp\zrs+\sum_{ij}\sum_{rs}\zij\zrs \\
&=N^4-N^3-N^3+N^2 \\
&=N^2(N-1)^2.
\end{align*}

Term 3.2 is $\fpe(\ke+2)/4$ times
\begin{align*}
&\sum_{ii'jj'}\sum_{rr'ss'}\zij\zipjp\zrs\zrpsp (1-1_{ij=i'j'}-1_{rs=r's'}+1_{ij=i'j'}1_{rs=r's'})1_{ij\in\{rs,r's'\}} \\
&=\sum_{ii'jj'}\sum_{rr'ss'}\zij\zipjp\zrs\zrpsp (1_{ij=rs}+1_{ij=r's'}-1_{ij=rs}1_{ij=r's'})\\
&\phe-\sum_{ij}\sum_{rr'ss'}\zij\zrs\zrpsp (1_{ij=rs}+1_{ij=r's'}-1_{ij=rs}1_{ij=r's'}) \\
&\phe-\sum_{ii'jj'}\sum_{rs}\zij\zipjp\zrs 1_{ij=rs}\\
&\phe+\sum_{ij}\sum_{rs}\zij\zrs 1_{ij=rs}\\
&=\sum_{ii'jj'}\sum_{r's'}\zij\zipjp\zrpsp+\sum_{ii'jj'}\sum_{rs}\zij\zipjp\zrs-\sum_{ii'jj'}\zij\zipjp\\
&\phe-\sum_{ij}\sum_{r's'}\zij\zrpsp-\sum_{ij}\sum_{rs}\zij\zrs+\sum_{ij}\zij\\
&\phe-\sum_{ii'jj'}\zij\zipjp+\sum_{ij}\zij\\
&=N^3+N^3-N^2\\
&\phe-N^2-N^2+N\\
&\phe-N^2+N\\
&= 2N^3-4N^2+2N \\
&= 2N(N-1)^2.
\end{align*}
By symmetry of indices, term 3.3 is the same as term 3.2.

For term 3.4, we have $\fpe$ times
\begin{align*}
&\sum_{ii'jj'}\sum_{rr'ss'}\zij\zipjp\zrs\zrpsp 1_{ij \ne i'j'}1_{rs \ne r's'} 1_{\{ij,i'j'\}=\{rs,r's'\}}\\
&=2\sum_{ii'jj'}\sum_{rr'ss'}\zij\zipjp\zrs\zrpsp 1_{ij \ne i'j'}1_{rs \ne r's'}1_{ij=rs}1_{i'j'=r's'}\\
&=2\sum_{ii'jj'}\zij\zipjp 1_{ij \ne i'j'}\\
&=2N(N-1).
\end{align*}

Summing terms 3.1 to 3.4, we get
\begin{align*}
&
\fpe N(N-1)[N(N-1)+2]+\fpe(\ke+2)N(N-1)^2.
\end{align*}

\subsubsection{$U_e^2$ Term 4}

\begin{align*}
&\frac14\sum_{ii'jj'}\sum_{rr'ss'}\zij\zipjp\zrs\zrpsp \daiip\dbssp\\
&=\frac14\Bigl(\sum_{ii'jj'}\zij\zipjp\daiip\Bigr) \Bigl(\sum_{rr'ss'}\zrs\zrpsp\dbssp\Bigr).
\end{align*}
The first factor is
\begin{align*}
\sum_{ii'jj'}\zij\zipjp\daiip &= 2\ssa\sum_{ii'jj'}\zij\zipjp(1-\oiip) \\
&= 2\ssa(N^2-\sum_{ijj'}\zij\zijp)\\
&= 2\ssa(N^2-\sum_i \nid^2).
\end{align*}
By the same argument, the second factor is
\begin{align*}
\sum_{rr'ss'}\zrs\zrpsp\dbssp &= 2\ssb(N^2-\sum_s \nds^2),
\end{align*}
and so term 4 is
\begin{align*}
\ssa\ssb(N^2-\sum_i \nid^2)(N^2-\sum_j \ndj^2).
\end{align*}

\subsubsection{$U_e^2$ Term 5}

\begin{align*}
&\frac14\sum_{ii'jj'}\sum_{rr'ss'}\zij\zipjp\zrs\zrpsp \daiip\dersrpsp\\
&=\frac14\Bigl(\sum_{ii'jj'}\zij\zipjp\daiip\Bigr) \Bigl(\sum_{rr'ss'}\zrs\zrpsp\dersrpsp\Bigr).
\end{align*}
The first factor is computed in the previous section.  The second factor is
\begin{align*}
\sum_{rr'ss'}\zrs\zrpsp\dersrpsp &= 2\sse\sum_{rr'ss'}\zrs\zrpsp(1-1_{rs=r's'})\\
&= 2\sse (N^2-\sum_{rs}\zrs)\\
&= 2\sse N(N-1).
\end{align*}
Thus, term 5 is
\begin{align*}
\ssa\sse N(N-1)(N^2-\sum_i \nid^2).
\end{align*}

\subsubsection{$U_e^2$ Term 6}

By symmetry of indices, this is the same as Term 4:
\begin{align*}
\ssa\ssb(N^2-\sum_i \nid^2)(N^2-\sum_j \ndj^2).
\end{align*}

\subsubsection{$U_e^2$ Term 7}

This is like term 5 with factors $A$ and $B$ interchanged.
Thus, term 7 is equal to
\begin{align*}
\ssb\sse N(N-1)(N^2-\sum_j \ndj^2).
\end{align*}

\subsubsection{$U_e^2$ Term 8}

By symmetry of indices, this is the same as term 5:
\begin{align*}
\ssa\sse N(N-1)(N^2-\sum_i \nid^2).
\end{align*}

\subsubsection{$U_e^2$ Term 9}

By symmetry of indices, this is the same as term 7:
\begin{align*}
\ssb\sse N(N-1)(N^2-\sum_j \ndj^2).
\end{align*}

\subsubsection{$U_e^2$ Term 10}

\begin{align*}
&\sum_{ii'jj'}\sum_{rr'ss'}\zij\zipjp\zrs\zrpsp\baiiprrp\bbjjpssp \\
&=\ssa\ssb\sum_{ii'jj'}\sum_{rr'ss'}\zij\zipjp\zrs\zrpsp\bigl(\oir-\oirp-\oipr+\oiprp\bigr)
\bigl(\ojs-\ojsp-\ojps+\ojpsp\bigr) \\
&=\ssa\ssb\sum_{ii'jj'}\sum_{rr'ss'}\zij\zipjp\zrs\zrpsp(\oir\ojs-\oir\ojsp-\oir\ojps+\oir\ojpsp \\[-2ex]
&\phantom{\ts\ssa\ssb\sum_{ii'jj'}\sum_x\zij\zipjp\zrs\zrpsp}
-\oirp\ojs+\oirp\ojsp+\oirp\ojps-\oirp\ojpsp \\
&\phantom{\ts\ssa\ssb\sum_{ii'jj'}\sum_x\zij\zipjp\zrs\zrpsp}
-\oipr\ojs+\oipr\ojsp+\oipr\ojps-\oipr\ojpsp \\
&\phantom{\ts\ssa\ssb\sum_{ii'jj'}\sum_x\zij\zipjp\zrs\zrpsp}
+\oiprp\ojs-\oiprp\ojsp-\oiprp\ojps+\oiprp\ojpsp) \\
&=\ssa\ssb\Bigl(\sum_{ii'jj'}\sum_{r's'}\zij\zipjp\zrpsp-\sum_{ii'jj'}\sum_{r's}\zij\zipjp\zis\zrpj-\sum_{ii'jj'}\sum_{r's'}\zij\zipjp\zijp\zrpsp+\sum_{ii'jj'}\sum_{r's}\zij\zipjp\zis\zrpjp \\
&\hspace{15mm}-\sum_{ii'jj'}\sum_{rs'}\zij\zipjp\zrj\zisp+\sum_{ii'jj'}\sum_{rs}\zij\zipjp\zrs+\sum_{ii'jj'}\sum_{rs'}\zij\zipjp\zrjp\zisp-\sum_{ii'jj'}\sum_{rs}\zij\zipjp\zrs\zijp \\
&\hspace{15mm}-\sum_{ii'jj'}\sum_{r's'}\zij\zipjp\zipj\zrpsp+\sum_{ii'jj'}\sum_{r's}\zij\zipjp\zips\zrpj+\sum_{ii'jj'}\sum_{r's'}\zij\zipjp\zrpsp-\sum_{ii'jj'}\sum_{r's}\zij\zipjp\zips\zrpjp \\
&\hspace{15mm}+\sum_{ii'jj'}\sum_{rs'}\zij\zipjp\zrj\zipsp-\sum_{ii'jj'}\sum_{rs}\zij\zipjp\zrs\zipj-\sum_{ii'jj'}\sum_{rs'}\zij\zipjp\zrjp\zipsp+\sum_{ii'jj'}\sum_{rs}\zij\zipjp\zrs\Bigr) \\
&=\ssa\ssb\Bigl(N^3-N\sum_{ij}\zij\nid\ndj-N\sum_{ij'}\nid\ndjp\zijp+\sum_{ij'}\nid^2\ndjp^2 \\
&\hspace{15mm}-N\sum_{ij}\zij\ndj\nid+N^3+\sum_{ij'}\ndjp^2\nid^2-N\sum_{ij'}\ndjp\nid\zijp \\
&\hspace{15mm}-N\sum_{i'j}\ndj\nipd\zipj+\sum_{i'j}\nipd^2\ndj^2+N^3-N\sum_{i'j'}\zipjp\nipd\ndjp \\
&\hspace{15mm}+\sum_{i'j}\ndj^2\nipd^2-N\sum_{i'j}\ndj\nipd\zipj-N\sum_{i'j'}\zipjp\ndjp\nipd+N^3\Bigr) \\
&=4\ssa\ssb\Bigl(N^3-2N\sum_{ij}\zij\nid\ndj+\sum_{ij}\nid^2\ndj^2\Bigr).
\end{align*}

\subsubsection{$U_e^2$ Term 11}

\begin{align*}
&\sum_{ii'jj'}\sum_{rr'ss'}\zij\zipjp\zrs\zrpsp\baiiprrp\beijipjprsrpsp\\
&=\ssa\sse\sum_{ii'jj'}\sum_{rr'ss'}\zij\zipjp\zrs\zrpsp\bigl(\oir-\oirp-\oipr+\oiprp\bigr)\bigl(1_{ij=rs}-1_{ij=r's'}-1_{i'j'=rs}+1_{i'j'=r's'}\bigr) \\
&=\ssa\sse\sum_{ii'jj'}\sum_{rr'ss'}\zij\zipjp\zrs\zrpsp\bigl(1_{ij=rs}-\oir1_{ij=r's'}-\oir1_{i'j'=rs}+\oir1_{i'j'=r's'}\\[-2ex]
&\hspace{54mm}-\oirp1_{ij=rs}+1_{ij=r's'}+\oirp1_{i'j'=rs}-\oirp1_{i'j'=r's'}\\
&\hspace{54mm}-\oipr1_{ij=rs}+\oipr1_{ij=r's'}+1_{i'j'=rs}-\oipr1_{i'j'=r's'}\\
&\hspace{54mm}+\oiprp1_{ij=rs}-\oiprp1_{ij=r's'}-\oiprp1_{i'j'=rs}+1_{i'j'=r's'}\bigr)\\
&=\ssa\sse\bigl(\sum_{ii'jj'}\sum_{r's'}\zij\zipjp\zrpsp-\sum_{ii'jj'}\sum_{s}\zij\zipjp\zis-\sum_{ijj'}\sum_{r's'}\zij\zijp\zrpsp+\sum_{ii'jj'}\sum_{s}\zij\zipjp\zis\\
&\hspace{15mm}-\sum_{ii'jj'}\sum_{s'}\zij\zipjp\zisp+\sum_{ii'jj'}\sum_{rs}\zij\zipjp\zrs+\sum_{ii'jj'}\sum_{s'}\zij\zipjp\zisp-\sum_{ijj'}\sum_{rs}\zij\zijp\zrs\\
&\hspace{15mm}-\sum_{ijj'}\sum_{r's'}\zij\zijp\zrpsp+\sum_{ii'jj'}\sum_{s}\zij\zipjp\zips+\sum_{ii'jj'}\sum_{r's'}\zij\zipjp\zrpsp-\sum_{ii'jj'}\sum_{s}\zij\zipjp\zips\\
&\hspace{15mm}+\sum_{ii'jj'}\sum_{s'}\zij\zipjp\zipsp-\sum_{ijj'}\sum_{rs}\zij\zijp\zrs-\sum_{ii'jj'}\sum_{s'}\zij\zipjp\zipsp+\sum_{ii'jj'}\sum_{rs}\zij\zipjp\zrs\bigr)\\
&=\ssa\sse\bigl(2\sum_{ii'jj'}\sum_{r's'}\zij\zipjp\zrpsp-2\sum_{ijj'}\sum_{r's'}\zij\zijp\zrpsp
+2\sum_{ii'jj'}\sum_{rs}\zij\zipjp\zrs-2\sum_{ijj'}\sum_{rs}\zij\zijp\zrs\bigr)\\
&=\ssa\sse\bigl(4N^3-4N\sum_i \nid^2\bigr).
\end{align*}

\subsubsection{$U_e^2$ Term 12}

We can use the symmetry with term 11, interchanging rows columns. Thus, term 12 is
\[
\ssb\sse\bigl(4N^3-4N\sum_j \ndj^2\bigr).
\]

\subsection{Combination}\label{sec:varue}

Summing up the results of the previous twelve sections, we have
\begin{align*}
\e(U_e^2) &= \fpa N^4-2\fpa N^2\sum_i \nid^2+3\fpa \Bigl(\sum_i \nid^2\Bigr)^2-2\fpa \sum_i \nid^4+\fpa(\ka+2)\Bigl(N^2\sum_i \nid^2-2N\sum_i \nid^3+\sum_i \nid^4\Bigr)\\
&+\fpb N^4-2\fpb N^2\sum_j \ndj^2+3\fpb \Bigl(\sum_j \ndj^2\Bigr)^2-2\fpb \sum_j \ndj^4+\fpb(\kb+2)\Bigl(N^2\sum_j \ndj^2-2N\sum_j \ndj^3+\sum_j \ndj^4\Bigr)\\
&+\fpe\Bigl(N^4-2N^3+3N^2-2N\Bigr)+\fpe(\ke+2)N(N-1)^2+\ssa\ssb(N^2-\sum_i \nid^2)(N^2-\sum_j \ndj^2)\\
&+\ssa\sse N(N-1)(N^2-\sum_i \nid^2)+\ssa\ssb(N^2-\sum_i \nid^2)(N^2-\sum_j \ndj^2)\\
&+\ssb\sse N(N-1)(N^2-\sum_j \ndj^2)+\ssa\sse N(N-1)(N^2-\sum_i \nid^2)+\ssb\sse N(N-1)(N^2-\sum_j \ndj^2)\\
&+4\ssa\ssb\Bigl(N^3-2N\sum_{ij}\zij\nid\ndj+\sum_{ij}\nid^2\ndj^2\Bigr)+4\ssa\sse\bigl(N^3-N\sum_i \nid^2\bigr)+\ssb\sse\bigl(4N^3-4N\sum_j \ndj^2\bigr).
\end{align*}


Then, we have
\begin{align*}
\var(U_e) &= \e(U_e^2)-\e(U_e)^2 \\
&= \e(U_e^2)-\fpa(N^2-\sum_i \nid^2)^2-\fpb(N^2-\sum_j \ndj^2)^2-\fpe(N^2-N)^2 \\
&\phe-2\ssa\ssb(N^2-\sum_i \nid^2)(N^2-\sum_j \ndj^2)-2\ssa\sse N(N-1)(N^2-\sum_i \nid^2)\\
&\phe-2\ssb\sse N(N-1)(N^2-\sum_j \ndj^2)\\
&=2\fpa \Bigl(\sum_i \nid^2\Bigr)^2-2\fpa \sum_i \nid^4+\fpa(\ka+2)
\Bigl(N^2\sum_i \nid^2-2N\sum_i \nid^3+\sum_i \nid^4\Bigr)\\
&\phe+2\fpb \Bigl(\sum_j \ndj^2\Bigr)^2-2\fpb \sum_j \ndj^4+\fpb(\kb+2)\Bigl(N^2\sum_j \ndj^2-2N\sum_j \ndj^3+\sum_j \ndj^4\Bigr)\\
&\phe+2\fpe(N^2-N)+\fpe(\ke+2)N(N-1)^2+4\ssa\ssb\Bigl(N^3-2N\sum_{ij}\zij\nid\ndj+\sum_{ij}\nid^2\ndj^2\Bigr)\\
&\phe+4\ssa\sse\bigl(N^3-N\sum_i \nid^2\bigr)+\ssb\sse\bigl(4N^3-4N\sum_j \ndj^2\bigr). 
\end{align*}

Next, we simplify the form of this expression.
The coefficient of $(\ka+2)\fpa$ is
$$
N^2\sum_i \nid^2-2N\sum_i \nid^3+\sum_i \nid^4
=\sum_i\nid^2(N-\nid)^2
$$
and similarly for that of $(\kb+2)\fpb$.
The coefficient of $(\ke+2)\fpe$ is $N(N-1)^2$.
The remaining multiple of $\fpa$ is
$$
2\Bigl( (\sum_i\nid^2)^2-\sum_i\nid^4\Bigr)
$$
and similarly for $\fpb$. The remaining multiple of $\fpe$ is $2N(N-1)$.
The coefficient of $\ssa\ssb$ is
\begin{align*}
4\Bigl(N^3-2N\sum_{ij}\zij\nid\ndj+\sum_{ij}\nid^2\ndj^2\Bigr)
&=
4\sum_{ij}(\nid^2\ndj^2-2N\zij\nid\ndj+N^2\zij)
\end{align*}
because $N^2\sum_{ij}\zij=N^3$. Therefore the coefficient of $\ssa\ssb$ is
$$4\sum_{ij}(\nid\ndj-N\zij)^2.$$
Applying these simplifications
\begin{align}\label{eq:varue}
\begin{split}
\var(U_e) & =  
2\fpa\Bigl( \Bigl(\sum_i\nid^2\Bigr)^2-\sum_i\nid^4\Bigr)  
+2\fpb\Bigl( \Bigl(\sum_j\ndj^2\Bigr)^2-\sum_j\ndj^4\Bigr)  
+2\fpe N(N-1)\\
&\phe +(\ka+2)\fpa\sum_i\nid^2(N-\nid)^2  
+(\kb+2)\fpb\sum_j\ndj^2(N-\ndj)^2  
+(\ke+2)\fpe N(N-1)^2\\
&\phe+4\ssa\ssb\sum_{ij}(\nid\ndj-N\zij)^2  
+4\ssa\sse N\bigl(N^2-\sum_i \nid^2\bigr)+4\ssb\sse N\bigl(N^2-\sum_j \ndj^2\bigr).  
\end{split}
\end{align}
 
The coefficient of $\ssa\ssb$ is a measure of how close to a regular $R\times C$ grid
the data are.  

\subsection{Check}

\subsubsection{$\ssa=\ssb=0$}
If $\ssa=\ssb=0$ then $\yij$ are IID with variance $\sse$
and kurtosis $\ke$.
Then
$$U_e = \frac12\sum_{iji'j'}(\yij-\yipjp)^2 = N(N-1)s^2_e$$
where $s_e$ is the usual sample standard deviation applied to all $N$
of the $\yij$.
Thus 
\begin{align*}
\var(U_e)  = \fpe N^2(N-1)^2\Bigl( \frac{2}{N-1} + \frac{\ke}{N}\Bigr)
 = \fpe \bigl( 2N^2(N-1)+\ke N(N-1)^2\bigr).
\end{align*}

Substituting $\ssa=\ssb=0$ in $\var(U_e)$ from Section~\ref{sec:varue}
yields
$$
2\fpe(N^2-N)+\fpe(\ke+2)N(N-1)^2
=2N^2(N-1)\fpe
+N(N-1)^2\ke\fpe 
$$
which matches the formula.
Equation~\eqref{eq:varue} becomes
\begin{align*}
\begin{split}
\var(U_e) & = 
2\fpe N(N-1) + (\ke+2)\fpe N(N-1)^2
\end{split}
\end{align*}
which also matches.

\subsubsection{IID sampling}

If $\max_i\nid=\max_j\ndj =1$, then the observations are IID 
with variance $\ssy= \ssa+\ssb+\sse$ and kurtosis 
$$\ky 
=\frac{\ka\fpa+\kb\fpb+\ke\fpe}{\fpy}. 
$$
Now $U_e$ is $N(N-1)$ times the sample standard deviation of all $N$ observations. 
Thus 
\begin{align}\label{eq:iidcheck}
\begin{split}
\var( U_e) & = 2N^2(N-1)\fpy + N(N-1)^2\ky\fpy\\
& = 2N^2(N-1)(\fpa+\fpb+\fpe+2\ssa\ssb+2\ssa\sse+2\ssb\sse) \\
&\phe+ N(N-1)^2 (\ka\fpa+\kb\fpb+\ke\fpe). 
\end{split}
\end{align}

In this case, the formula gives
\begin{align*}
&2\fpa \Bigl(\sum_i \nid^2\Bigr)^2-2\fpa \sum_i \nid^4+\fpa(\ka+2)\Bigl(N^2\sum_i \nid^2-2N\sum_i \nid^3+\sum_i \nid^4\Bigr)\\
&\phe+2\fpb \Bigl(\sum_j \ndj^2\Bigr)^2-2\fpb \sum_j \ndj^4+\fpb(\kb+2)\Bigl(N^2\sum_j \ndj^2-2N\sum_j \ndj^3+\sum_j \ndj^4\Bigr)\\
&\phe+2\fpe(N^2-N)+\fpe(\ke+2)N(N-1)^2+4\ssa\ssb\Bigl(N^3-2N\sum_{ij}\zij\nid\ndj+\sum_{ij}\nid^2\ndj^2\Bigr)\\
&\phe+4\ssa\sse\bigl(N^3-N\sum_i \nid^2\bigr)+\ssb\sse\bigl(4N^3-4N\sum_j \ndj^2\bigr).
\end{align*}

If we set all positive $\nid=1$ and all positive $\ndj=1$
then $\sum_i\nid^2=N$ because there are now $R=N$
rows in the data. Similarly $\nid^3$ and $\nid^4$ sum to $N$
and these powers of $\ndj$ also sum to $N$.
Next $\sum_{ij}\zij\nid\ndj=\sum_{ij}\zij=N$.
The most subtle of these sums is $\sum_{ij}\nid^2\ndj^2
=\sum_{ij}1=N^2$ because the indices run over all $i$ with $\nid>0$
and all $j$ with $\ndj>0$.

Making these substitutions we get 
\begin{align*}
\var(U_e)&=2\fpa N^2-2\fpa N+\fpa(\ka+2)(N^3-2N^2+N)\\
&\phe+2\fpb N^2-2\fpb N+\fpb(\kb+2)(N^3-2N^2+N)\\
&\phe+2\fpe(N^2-N)+\fpe(\ke+2)N(N-1)^2+4\ssa\ssb(N^3-2N^2+N^2)\\
&\phe+4\ssa\sse(N^3-N^2)+\ssb\sse(4N^3-4N^2)\\ 
&=2\fpa\bigl(N^2-N+N(N-1)^2\bigr)\\
&\phe+2\fpb\bigl(N^2-N+N(N-1)^2\bigr)\\
&\phe+2\fpe\bigl( N(N-1) + N(N-1)^2\bigr)+4\ssa\ssb N^2(N-1)\\
&\phe+4\ssa\sse N^2(N-1)+4\ssb\sse N^2(N-1)\\
&\phe+N(N-1)^2(\ka\fpa+\kb\fpb+\ke\fpe)\\
&=2N^2(N-1)(\fpa+\fpb+\fpe)\\
&\phe+4N^2(N-1)\bigl(\ssa\ssb+ \ssa\sse+ \ssa\sse \bigr)\\
&\phe+N(N-1)^2(\ka\fpa+\kb\fpb+\ke\fpe)
\end{align*}
which matches equation~\eqref{eq:iidcheck}.

Equation~\eqref{eq:varue} gives
\begin{align*}
\begin{split}
\var(U_e) 
& =  
2\fpa\Bigl( \Bigl(\sum_i\nid^2\Bigr)^2-\sum_i\nid^4\Bigr)  
+2\fpb\Bigl( \Bigl(\sum_j\ndj^2\Bigr)^2-\sum_j\ndj^4\Bigr)  
+2\fpe N(N-1)\\
&\phe +(\ka+2)\fpa\sum_i\nid^2(N-\nid)^2  
+(\kb+2)\fpb\sum_j\ndj^2(N-\ndj)^2  
+(\ke+2)\fpe N(N-1)^2\\
&\phe+4\ssa\ssb\sum_{ij}(\nid\ndj-N\zij)^2  
+4\ssa\sse N\bigl(N^2-\sum_i \nid^2\bigr)+4\ssb\sse N\bigl(N^2-\sum_j \ndj^2\bigr)\\
& =  
2\fpa( N^2-N)+2\fpb(N^2-N)+2\fpe N(N-1)\\
&\phe +(\ka+2)\fpa N(N-1)^2+(\kb+2)\fpb  N(N-1)^2+(\ke+2)\fpe N(N-1)^2\\
&\phe+4\ssa\ssb N^2(N-1)^2
+4\ssa\sse N(N^2-N)+4\ssb\sse N(N^2-N).
\end{split}
\end{align*}

\subsubsection{IID sampling again}

If $\max_i\nid=1$ and $\ssb=0$, then once again the observations are IID and 
\begin{align}
\var( U_e) & = 2N^2(N-1)\fpy + N(N-1)^2\ky\fpy \notag\\
& = 2N^2(N-1)(\fpa+\fpb+\fpe+2\ssa\ssb+2\ssa\sse+2\ssb\sse) \notag\\
&\phe+ N(N-1)^2 (\ka\fpa+\kb\fpb+\ke\fpe) \notag\\ 
& = 2N^2(N-1)(\fpa+\fpe+2\ssa\sse) 
+ N(N-1)^2 (\ka\fpa+\ke\fpe). \label{eq:veiidcheck2}
\end{align}

The formula gives
\begin{align*}
\var(U_e)
&=2\fpa \Bigl(\sum_i \nid^2\Bigr)^2-2\fpa \sum_i \nid^4+\fpa(\ka+2)\Bigl(N^2\sum_i \nid^2-2N\sum_i \nid^3+\sum_i \nid^4\Bigr)\\
&\phe+2\fpb \Bigl(\sum_j \ndj^2\Bigr)^2-2\fpb \sum_j \ndj^4+\fpb(\kb+2)\Bigl(N^2\sum_j \ndj^2-2N\sum_j \ndj^3+\sum_j \ndj^4\Bigr)\\
&\phe+2\fpe(N^2-N)+\fpe(\ke+2)N(N-1)^2+4\ssa\ssb\Bigl(N^3-2N\sum_{ij}\zij\nid\ndj+\sum_{ij}\nid^2\ndj^2\Bigr)\\
&\phe+4\ssa\sse\bigl(N^3-N\sum_i \nid^2\bigr)+\ssb\sse\bigl(4N^3-4N\sum_j \ndj^2\bigr)\\
&=2\fpa \Bigl(\sum_i \nid^2\Bigr)^2-2\fpa \sum_i \nid^4+\fpa(\ka+2)\Bigl(N^2\sum_i \nid^2-2N\sum_i \nid^3+\sum_i \nid^4\Bigr)\\
&\phe+2\fpe(N^2-N)+\fpe(\ke+2)N(N-1)^2 +4\ssa\sse(N^3-N\sum_i \nid^2)\\
&=2\fpa N^2-2\fpa N+\fpa(\ka+2)(N^3-2N^2+N)\\
&\phe+2\fpe(N^2-N)+\fpe(\ke+2)N(N-1)^2 +4\ssa\sse(N^3-N^2)\\
&=2N^2(N-1)(\fpa+\fpe+2\ssa\sse)+N(N-1)^2(\ka\fpa+\ke\fpe),
\end{align*}
matching~\eqref{eq:veiidcheck2}.

\section{Covariance of $U_a$ and $U_b$}\label{sec:covuaub}

We use the formula $\cov(U_a,U_b) = \e(U_aU_b)-\e(U_a)\e(U_b)$, so we just need to compute $\e(U_aU_b)$. Using our preferred normalization,
\begin{align*}
U_aU_b &= \frac14\bigl(\sum_{ijj'} \nid^{-1}\zij\zijp(\yij-\yijp)^2 \bigr)\bigl(\sum_{rr's} \nds^{-1}\zrs\zrps(\yrs-\yrps)^2\bigr) \\
&= \frac14\sum_{ijj'}\sum_{rr's} \nid^{-1}\nds^{-1}\zij\zijp\zrs\zrps(\yij-\yijp)^2(\yrs-\yrps)^2
\end{align*}

Then,
\begin{align*}
\e(U_aU_b) &= \frac14\sum_{ijj'}\sum_{rr's} \nid^{-1}\nds^{-1}\zij\zijp\zrs\zrps \bigl(\,\underbrace{\qeijijprsrps}_{\text{Term 1}}\\
&+\underbrace{\dbjjp\darrp}_{\text{Term 2}}+\underbrace{\dbjjp\dersrps}_{\text{Term 3}}+\underbrace{\deijijp\darrp}_{\text{Term 4}}\bigr).
\end{align*}

We consider each term separately.

\subsection{$U_aU_b$ Term 1}

\begin{align*}
&\frac14\sum_{ijj'}\sum_{rr's} \nid^{-1}\nds^{-1}\zij\zijp\zrs\zrps\qeijijprsrps\\
&=\frac\fpe4\sum_{ijj'}\sum_{rr's} \nid^{-1}\nds^{-1}\zij\zijp\zrs\zrps 1_{ij\ne ij'}1_{rs \ne r's}\\
&\phe\Bigl(4+(\ke+2)(1_{ij\in\{rs,r's\}}+1_{ij'\in\{rs,r's\}})+4\times1_{\{ij,ij'\}=\{rs,r's\}}\Bigr)\\
&=\frac\fpe4\sum_{ijj'}\sum_{rr's} \nid^{-1}\nds^{-1}\zij\zijp\zrs\zrps 1_{j\ne j'}1_{r \ne r'}\\
&\phe\Bigl(\underbrace{4}_{\text{1.1}}+\underbrace{(\ke+2)(1_{ij\in\{rs,r's\}}+1_{ij'\in\{rs,r's\}})}_{\text{1.2 and 1.3}}+\underbrace{4\times1_{\{ij,ij'\}=\{rs,r's\}}}_{\text{1.4}}\Bigr).
\end{align*}

For 1.1, we have
\begin{align*}
&\fpe\sum_{ijj'}\sum_{rr's} \nid^{-1}\nds^{-1}\zij\zijp\zrs\zrps 1_{j\ne j'}1_{r \ne r'}\\
&=\fpe\bigl(\sum_{ijj'}\nid^{-1}\zij\zijp(1-1_{j=j'})\bigr)\bigl(\sum_{rr's}\nds^{-1}\zrs\zrps(1-1_{r=r'})\bigr)\\
&=\fpe\bigl(\sum_i \nid-\sum_{ij}\nid^{-1}\zij\bigr)\bigl(\sum_s \nds-\sum_{rs}\nds^{-1}\zrs\bigr)\\
&=\fpe\bigl(N-R\bigr)\bigl(N-C\bigr).
\end{align*}

For 1.2, we have $\fpe(\ke+2)/4$ times
\begin{align*}
&\sum_{ijj'}\sum_{rr's}\nid^{-1}\nds^{-1}\zij\zijp\zrs\zrps 1_{j\ne j'}1_{r \ne r'}1_{ij\in\{rs,r's\}}\\
&=2\sum_{ijj'}\sum_{rr's}\nid^{-1}\nds^{-1}\zij\zijp\zrs\zrps 1_{j\ne j'}1_{r \ne r'}1_{ij=rs}\\
&=2\sum_{ijj'}\sum_{r'}\nid^{-1}\ndj^{-1}\zij\zijp\zrpj 1_{j\ne j'}1_{i \ne r'} \\
&=2\sum_{ijj'}\sum_{r'}\nid^{-1}\ndj^{-1}\zij\zijp\zrpj (1-1_{j=j'}-1_{i=r'}+1_{j=j'}1_{i=r'})\\
&=2\sum_{ijj'}\nid^{-1}\zij\zijp-2\sum_{ij}\sum_{r'}\nid^{-1}\ndj^{-1}\zij\zrpj-2\sum_{ijj'}\nid^{-1}\ndj^{-1}\zij\zijp+2\sum_{ij}\nid^{-1}\ndj^{-1}\zij \\
&=2\sum_{ij} \zij-2\sum_{ij}\nid^{-1}\zij-2\sum_{ij}\ndj^{-1}\zij+2\sum_{ij}\nid^{-1}\ndj^{-1}\zij \\
&=2\sum_{ij} \zij(1-\nid^{-1})(1-\ndj^{-1}).
\end{align*}
Term 1.3 is the same as 1.2 by symmetry of indices.

For 1.4, we have
\begin{align*}
&\fpe\sum_{ijj'}\sum_{rr's}\nid^{-1}\nds^{-1}\zij\zijp\zrs\zrps 1_{j\ne j'}1_{r \ne r'}1_{\{ij,ij'\}=\{rs,r's\}}=0,
\end{align*}
since the last indicator implies $r=i$ and $r'=i$ but the second one is $1_{r\ne r'}$.

Summing up, term 1 is equal to
\begin{align*}
&\fpe\bigl(N-R\bigr)\bigl(N-C\bigr)+\fpe(\ke+2)\sum_{ij} \zij(1-\nid^{-1})(1-\ndj^{-1})\\
&=\fpe\bigl(N-R\bigr)\bigl(N-C\bigr)+\fpe(\ke+2)(N-R-C+\sum_{ij}\nid^{-1}\ndj^{-1}\zij).
\end{align*}

\subsection{$U_aU_b$ Term 2}

\begin{align*}
&\frac14\sum_{ijj'}\sum_{rr's} \nid^{-1}\nds^{-1}\zij\zijp\zrs\zrps\dbjjp\darrp\\
&=\frac14\sum_{ijj'}\sum_{rr's} \nid^{-1}\nds^{-1}\zij\zijp\zrs\zrps 2\ssb(1-\ojjp)2\ssa(1-\orrp)\\
&=\ssa\ssb\bigl(\sum_{ijj'}\nid^{-1}\zij\zijp(1-\ojjp)\bigr)\bigl(\sum_{rr's}\nds^{-1}\zrs\zrps(1-\orrp)\bigr)\\
&=\ssa\ssb\bigl(\sum_{i}\nid-\sum_{ij}\nid^{-1}\zij\bigr)\bigl(\sum_{s}\nds-\sum_{rs}\nds^{-1}\zrs\bigr)\\
&=\ssa\ssb\bigl(N-R\bigr)\bigl(N-C\bigr).
\end{align*}

\subsection{$U_aU_b$ Term 3}

\begin{align*}
&\frac14\sum_{ijj'}\sum_{rr's} \nid^{-1}\nds^{-1}\zij\zijp\zrs\zrps\dbjjp\dersrps\\
&=\frac14\sum_{ijj'}\sum_{rr's}\nid^{-1}\nds^{-1}\zij\zijp\zrs\zrps 2\ssb(1-\ojjp)2\sse(1-\orrp)\\
&=\ssb\sse\bigl(N-R\bigr)\bigl(N-C\bigr)
\end{align*}
using the previous section.

\subsection{$U_aU_b$ Term 4}

\begin{align*}
&\frac14\sum_{ijj'}\sum_{rr's} \nid^{-1}\nds^{-1}\zij\zijp\zrs\zrps\deijijp\darrp\\
&=\frac14\sum_{ijj'}\sum_{rr's}\nid^{-1}\nds^{-1}\zij\zijp\zrs\zrps 2\sse(1-\ojjp)2\ssa(1-\orrp)\\
&=\ssa\sse\bigl(N-R\bigr)\bigl(N-C\bigr)
\end{align*}
using the previous section.

\subsection{Combination}

Adding up the four terms, we have
\begin{align*}
\e(U_aU_b)&=\fpe\bigl(N-R\bigr)\bigl(N-C\bigr)+\fpe(\ke+2)\sum_{ij} \zij(1-\nid^{-1})(1-\ndj^{-1})\\
&\hspace{4mm}+\ssa\ssb\bigl(N-R\bigr)\bigl(N-C\bigr)+\ssb\sse\bigl(N-R\bigr)\bigl(N-C\bigr)+\ssa\sse\bigl(N-R\bigr)\bigl(N-C\bigr),
\end{align*}
and so
\begin{align*}
\cov(U_a,U_b) &= \e(U_aU_b)-\e(U_a)\e(U_b)\\
&= \e(U_aU_b)-(\ssb+\sse)(\ssa+\sse)(N-R)(N-C)\\
&=\fpe(\ke+2)\sum_{ij} \zij(1-\nid^{-1})(1-\ndj^{-1}). 
\end{align*}

Notice that $\cov(U_a,U_b)=0$ when $\sse=0$. This can be verified
by noting that when $\sse=0$ then 
$U_a$ is a function only of $a_i$ while $U_b$ is a function only
of $b_j$.  Therefore $U_a$ and $U_b$ are independent when $\sse=0$.

\section{Covariance of $U_a$ and $U_e$}\label{sec:covuaue}

We use the formula $\cov(U_a,U_e) = \e(U_aU_e)-\e(U_a)\e(U_e)$, so we just need to compute $\e(U_aU_e)$. 
First,
\begin{align*}
U_aU_e &= \frac14\Bigl(\sum_{ijj'} \nid^{-1}\zij\zijp(\yij-\yijp)^2 \Bigr)\Bigl(\sum_{rr'ss'} \zrs\zrpsp(\yrs-\yrpsp)^2\Bigr) \\
&= \frac14\sum_{ijj'}\sum_{rr'ss'} \nid^{-1}\zij\zijp\zrs\zrpsp(\yij-\yijp)^2(\yrs-\yrpsp)^2.
\end{align*}
Then,
\begin{align*}
\e(U_aU_e) &= \frac14\sum_{ijj'}\sum_{rr'ss'} \nid^{-1}\zij\zijp\zrs\zrpsp \bigl(\,\underbrace{\qbjjpssp}_{\text{Term 1}} + \underbrace{\qeijijprsrpsp}_{\text{Term 2}}+\underbrace{\dbjjp\darrp}_{\text{Term 3}}+\underbrace{\dbjjp\dersrpsp}_{\text{Term 4}}\\
&+\underbrace{\deijijp\darrp}_{\text{Term 5}}+\underbrace{\deijijp\dbssp}_{\text{Term 6}}+\underbrace{4\bbjjpssp\beijijprsrpsp}_{\text{Term 7}}\,\bigr).
\end{align*}
We consider each term separately.

\subsection{$U_aU_e$ Term 1}

\begin{align*}
&\frac14\sum_{ijj'}\sum_{rr'ss'} \nid^{-1}\zij\zijp\zrs\zrpsp\qbjjpssp \\
&=\frac14\sum_{ijj'}\sum_{rr'ss'} \nid^{-1}\zij\zijp\zrs\zrpsp 1_{j\ne j'}1_{s\ne s'}\fpb\Bigl(\underbrace{4}_{1.1}+\underbrace{(\kb+2)(1_{j\in\{s,s'\}}+1_{j'\in\{s,s'\}})}_{\text{1.2 and 1.3}}+\underbrace{4\times1_{\{j,j'\}=\{s,s'\}}}_{1.4}\Bigr).
\end{align*}

Term 
1.1 is equal to $\fpb$ times
\begin{align*}
&\sum_{ijj'}\sum_{rr'ss'} \nid^{-1}\zij\zijp\zrs\zrpsp 1_{j\ne j'}1_{s\ne s'} \\
&=\bigl(\sum_{ijj'}\nid^{-1}\zij\zijp(1-\ojjp)\bigr)\bigl(\sum_{rr'ss'}\zrs\zrpsp(1-\ossp)\bigr) \\
&=\bigl(\sum_i \nid-\sum_{ij}\nid^{-1}\zij\bigr)\bigl(N^2-\sum_{rr's}\zrs\zrps\bigr)\\
&=\bigl(N-R\bigr)\bigl(N^2-\sum_{s}\nds^2\bigr).
\end{align*}

Term 1.2 is equal to $\fpb(\kb+2)/4$ times
\begin{align*}
&\sum_{ijj'}\sum_{rr'ss'} \nid^{-1}\zij\zijp\zrs\zrpsp 1_{j\ne j'}1_{s\ne s'}1_{j\in\{s,s'\}}\\
&=2\sum_{ijj'}\sum_{rr'ss'} \nid^{-1}\zij\zijp\zrs\zrpsp 1_{j\ne j'}1_{s\ne s'}1_{j=s}\\
&=2\sum_{ijj'}\sum_{rr's'} \nid^{-1}\zij\zijp\zrj\zrpsp (1-\ojjp)(1-1_{j= s'})\\
&=2\sum_{ijj'}\sum_{rr's'} \nid^{-1}\zij\zijp\zrj\zrpsp 
-2\sum_{ij}\sum_{rr's'} \nid^{-1}\zij\zrj\zrpsp 
\\&\phe 
-2\sum_{ijj'}\sum_{rr'} \nid^{-1}\zij\zijp\zrj\zrpj 
+2\sum_{ij}\sum_{rr'} \nid^{-1}\zij\zrj\zrpj \\
&=2N\sum_{ijj'}\sum_{r} \nid^{-1}\zij\zijp\zrj-2N\sum_{ij}\sum_{r} \nid^{-1}\zij\zrj
\\&\phe 
-2\sum_{ijj'}\sum_{r} \nid^{-1}\ndj \zij\zijp\zrj
+2\sum_{ij}\sum_{r} \nid^{-1}\ndj\zij\zrj \\
&=2N\sum_{ij}\zij\ndj-2N\sum_{ij}\nid^{-1}\zij\ndj-2\sum_{ij}\zij\ndj^2+2\sum_{ij}\nid^{-1}\zij\ndj^2\\
&=2\sum_{ij}\zij(N\ndj-N\nid^{-1}\ndj-\ndj^2+\nid^{-1}\ndj^2)\\
&=2\sum_{ij}\zij(N-\ndj)\ndj(1-\nid^{-1}).
\end{align*}
Term 1.3 is equal to term 1.2 by symmetry of indices.

Term 1.4 is equal to $\fpb$ times
\begin{align*}
&\sum_{ijj'}\sum_{rr'ss'} \nid^{-1}\zij\zijp\zrs\zrpsp 1_{j\ne j'}1_{s\ne s'}1_{\{j,j'\}=\{s,s'\}} \\
&=2\sum_{ijj'}\sum_{rr'ss'} \nid^{-1}\zij\zijp\zrs\zrpsp 1_{j\ne j'}1_{s\ne s'}1_{j=s}1_{j'=s'} \\
&=2\sum_{ijj'}\sum_{rr'} \nid^{-1}\zij\zijp\zrj\zrpjp (1-\ojjp) \\
&=2\sum_{ijj'}\nid^{-1}\zij\zijp\ndj\ndjp-2\sum_{ij}\sum_{rr'}\nid^{-1}\zij\zrj\zrpj\\
&=2\sum_i\nid^{-1}\Bigl( \sum_j\zij\ndj\Bigr)^2-2\sum_{ij}\nid^{-1}\zij\ndj^2.
\end{align*}

Summing the four terms, we find that term 1 is equal to
\begin{align*}
&\fpb\bigl(N-R\bigr)\bigl(N^2-\sum_{j}\ndj^2\bigr)+
2\fpb\Bigl(\sum_i\nid^{-1}\Bigl( \sum_j\zij\ndj\Bigr)^2-\sum_{ij}\nid^{-1}\zij\ndj^2\Bigr)\\
&+\fpb(\kb+2)\sum_{ij}\zij(N-\ndj)\ndj(1-\nid^{-1}).
\end{align*}

\subsection{$U_aU_e$ Term 2}

\begin{align*}
&\frac14\sum_{ijj'}\sum_{rr'ss'} \nid^{-1}\zij\zijp\zrs\zrpsp\qeijijprsrpsp \\
&=\frac14\sum_{ijj'}\sum_{rr'ss'} \nid^{-1}\zij\zijp\zrs\zrpsp 1_{j\ne j'}1_{rs\ne r's'}\fpe\\
&\phe\Bigl(\underbrace{4}_{2.1}+\underbrace{(\ke+2)(1_{ij\in\{rs,r's'\}}+1_{ij'\in\{rs,r's'\}})}_{\text{2.2 and 2.3}}+\underbrace{4\times1_{\{ij,ij'\}=\{rs,r's'\}}}_{2.4}\Bigr).
\end{align*}

For 2.1, we get $\fpe$ times
\begin{align*}
&\sum_{ijj'}\sum_{rr'ss'} \nid^{-1}\zij\zijp\zrs\zrpsp 1_{j\ne j'}1_{rs\ne r's'}\\
&=N(N-1)\sum_{ijj'}\nid^{-1}\zij\zijp (1-\ojjp)\\
&=N(N-1)\bigl(\sum_i \nid-\sum_{ij}\nid^{-1}\zij\bigr)\\
&=N(N-1)(N-R).
\end{align*}

For 2.2, we get $\fpe(\ke+2)/4$ times
\begin{align*}
&\sum_{ijj'}\sum_{rr'ss'} \nid^{-1}\zij\zijp\zrs\zrpsp 1_{j\ne j'}1_{rs\ne r's'}1_{ij\in\{rs,r's'\}}\\
&=2\sum_{ijj'}\sum_{rr'ss'} \nid^{-1}\zij\zijp\zrs\zrpsp 1_{j\ne j'}1_{rs\ne r's'}1_{ij=rs}\\
&=2\sum_{ijj'}\sum_{r's'} \nid^{-1}\zij\zijp\zrpsp (1-\ojjp)(1-1_{ij=r's'})\\
&=2\sum_{ijj'}\sum_{r's'} \nid^{-1}\zij\zijp\zrpsp (1-\ojjp-1_{ij=r's'}+\ojjp 1_{ij=r's'})\\
&=2N\sum_i \nid-2N\sum_{ij}\nid^{-1}\zij-2\sum_{ijj'}\nid^{-1}\zij\zijp+2\sum_{ij}\nid^{-1}\zij\\
&=2N^2-2NR-2N+2R\\
&=2(N-R)(N-1).
\end{align*}
Term 2.3 is the same as 2.2 by symmetry of indices.

For term 2.4, we get $\fpe$ times
\begin{align*}
&\sum_{ijj'}\sum_{rr'ss'} \nid^{-1}\zij\zijp\zrs\zrpsp 1_{j\ne j'}1_{rs\ne r's'}1_{\{ij,ij'\}=\{rs,r's'\}}\\
&=2\sum_{ijj'}\sum_{rr'ss'} \nid^{-1}\zij\zijp\zrs\zrpsp 1_{j\ne j'}1_{rs\ne r's'}1_{ij=rs}1_{ij'=r's'}\\
&=2\sum_{ijj'}\nid^{-1}\zij\zijp 1_{j \ne j'}\\
&=2\sum_i \nid-2\sum_{ij}\nid^{-1}\zij \\
&=2(N-R).
\end{align*}

Adding up the four terms, we find that term 2 equals
\begin{align*}
\fpe N(N-1)(N-R)+2\fpe (N-R)+\fpe(\ke+2)(N-R)(N-1).
\end{align*}

\subsection{$U_aU_e$ Term 3}

\begin{align*}
&\frac14\sum_{ijj'}\sum_{rr'ss'} \nid^{-1}\zij\zijp\zrs\zrpsp \dbjjp\darrp\\
&=\frac14\sum_{ijj'}\sum_{rr'ss'} \nid^{-1}\zij\zijp\zrs\zrpsp 2\ssb(1-\ojjp)2\ssa(1-\orrp)\\
&=\ssa\ssb\bigl(\sum_{ijj'}\nid^{-1}\zij\zijp(1-\ojjp)\bigr)\bigl(\sum_{rr'ss'}\zrs\zrpsp(1-\orrp)\bigr)\\
&=\ssa\ssb\bigl(\sum_i \nid-\sum_{ij}\nid^{-1}\zij\bigr)\bigl(N^2-\sum_{rss'}\zrs\zrsp\bigr)\\
&=\ssa\ssb\bigl(N-R\bigr)\bigl(N^2-\sum_{r}\nrd^2\bigr).
\end{align*}

\subsection{$U_aU_e$ Term 4}

\begin{align*}
&\frac14\sum_{ijj'}\sum_{rr'ss'} \nid^{-1}\zij\zijp\zrs\zrpsp \dbjjp\dersrpsp\\
&=\frac14\sum_{ijj'}\sum_{rr'ss'} \nid^{-1}\zij\zijp\zrs\zrpsp 2\ssb(1-\ojjp)2\sse(1-\orrp\ossp)\\
&=\ssb\sse\bigl(\sum_{ijj'}\nid^{-1}\zij\zijp(1-\ojjp)\bigr)\bigl(\sum_{rr'ss'}\zrs\zrpsp(1-\orrp\ossp)\bigr)\\
&=\ssb\sse\bigl(N-R\bigr)\bigl(N^2-\sum_{rs}\zrs\bigr)\\
&=\ssb\sse\bigl(N-R\bigr)\bigl(N^2-N).
\end{align*}

\subsection{$U_aU_e$ Term 5}

\begin{align*}
&\frac14\sum_{ijj'}\sum_{rr'ss'} \nid^{-1}\zij\zijp\zrs\zrpsp \deijijp\darrp\\
&=\frac14\sum_{ijj'}\sum_{rr'ss'} \nid^{-1}\zij\zijp\zrs\zrpsp 2\sse(1-\ojjp)2\ssa(1-\orrp)\\
&=\ssa\sse\bigl(\sum_{ijj'}\nid^{-1}\zij\zijp(1-\ojjp)\bigr)\bigl(\sum_{rr'ss'}\zrs\zrpsp(1-\orrp)\bigr)\\
&=\ssa\sse\bigl(N-R\bigr)\bigl(N^2-\sum_{r}\nrd^2\bigr)
\end{align*}
using the result for term 3.

\subsection{$U_aU_e$ Term 6}

\begin{align*}
&\frac14\sum_{ijj'}\sum_{rr'ss'} \nid^{-1}\zij\zijp\zrs\zrpsp \deijijp\dbssp\\
&=\frac14\sum_{ijj'}\sum_{rr'ss'} \nid^{-1}\zij\zijp\zrs\zrpsp 2\sse(1-\ojjp)2\ssb(1-\ossp)\\
&=\ssb\sse\bigl(\sum_{ijj'}\nid^{-1}\zij\zijp(1-\ojjp)\bigr)\bigl(\sum_{rr'ss'}\zrs\zrpsp(1-\ossp)\bigr)\\
&=\ssb\sse\bigl(N-R\bigr)\bigl(N^2-\sum_{rr's}\zrs\zrps\bigr)\\
&=\ssb\sse\bigl(N-R\bigr)\bigl(N^2-\sum_{s}\nds^2\bigr).
\end{align*}

\subsection{$U_aU_e$ Term 7}

\begin{align*}
&\sum_{ijj'}\sum_{rr'ss'} \nid^{-1}\zij\zijp\zrs\zrpsp\bbjjpssp\beijijprsrpsp \\
&=\sum_{ijj'}\sum_{rr'ss'} \nid^{-1}\zij\zijp\zrs\zrpsp\ssb\bigl(\ojs-\ojsp-\ojps+\ojpsp\bigr)\\
&\hspace{52mm}\sse\bigl(1_{ij=rs}-1_{ij=r's'}-1_{ij'=rs}+1_{ij'=r's'}\bigr)\\
&=\ssb\sse\sum_{ijj'}\sum_{rr'ss'} \nid^{-1}\zij\zijp\zrs\zrpsp\bigl(
\hspace{4mm}1_{ij=rs}-\ojs 1_{ij=r's'}-\ojs 1_{ij'=rs}+\ojs 1_{ij'=r's'}\\[-1.75ex]
&\hspace{60mm}-\ojsp 1_{ij=rs}+1_{ij=r's'}+\ojsp 1_{ij'=rs}-\ojsp 1_{ij'=r's'}\\
&\hspace{60mm}-\ojps 1_{ij=rs}+\ojps 1_{ij=r's'}+1_{ij'=rs}-\ojps 1_{ij'=r's'}\\
&\hspace{60mm}+\ojpsp 1_{ij=rs}-\ojpsp 1_{ij=r's'}-\ojpsp 1_{ij'=rs}+1_{ij'=r's'}\bigr)\\
&=\ssb\sse\bigl(\sum_{ijj'}\sum_{r's'} \nid^{-1}\zij\zijp\zrpsp-\sum_{ijj'}\sum_{r} \nid^{-1}\zij\zijp\zrj-\sum_{ij}\sum_{r's'} \nid^{-1}\zij\zrpsp+\sum_{ijj'}\sum_{r} \nid^{-1}\zij\zijp\zrj\\
&\hspace{14mm}-\sum_{ijj'}\sum_{r'} \nid^{-1}\zij\zijp\zrpj+\sum_{ijj'}\sum_{rs} \nid^{-1}\zij\zijp\zrs+\sum_{ijj'}\sum_{r'} \nid^{-1}\zij\zijp\zrpj-\sum_{ij}\sum_{rs} \nid^{-1}\zij\zrs\\
&\hspace{14mm}-\sum_{ij}\sum_{r's'} \nid^{-1}\zij\zrpsp+\sum_{ijj'}\sum_{r} \nid^{-1}\zij\zijp\zrjp+\sum_{ijj'}\sum_{r's'} \nid^{-1}\zij\zijp\zrpsp-\sum_{ijj'}\sum_{r} \nid^{-1}\zij\zijp\zrjp\\
&\hspace{14mm}+\sum_{ijj'}\sum_{r'} \nid^{-1}\zij\zijp\zrpjp-\sum_{ij}\sum_{rs} \nid^{-1}\zij\zrs-\sum_{ijj'}\sum_{r'} \nid^{-1}\zij\zijp\zrpjp+\sum_{ijj'}\sum_{rs} \nid^{-1}\zij\zijp\zrs\bigr)\\
&=4\ssb\sse\bigl(N\sum_{ijj'}\nid^{-1}\zij\zijp-N\sum_{ij}\nid^{-1}\zij\bigr)\\
&=4\ssb\sse N\bigl(\sum_i \nid-\sum_i 1\bigr)\\
&=4\ssb\sse N(N-R).
\end{align*}

\subsection{Combination}

We add up the seven terms, replacing some $\nrd$ and $\nds$
expressions by equivalents using $\nid$ and $\ndj$, getting
\begin{align*}
\e(U_aU_e) &= 
\fpb\bigl(N-R\bigr)\bigl(N^2-\sum_{j}\ndj^2\bigr)+
2\fpb\Bigl(\sum_i\nid^{-1}\Bigl( \sum_j\zij\ndj\Bigr)^2-\sum_{ij}\nid^{-1}\zij\ndj^2\Bigr)\\
&\phe+\fpb(\kb+2)\sum_{ij}\zij(N-\ndj)\ndj(1-\nid^{-1})\\
&\phe+\fpe N(N-1)(N-R)+2\fpe (N-R)+\fpe(\ke+2)(N-R)(N-1)\\
&\phe+\ssa\ssb\bigl(N-R\bigr)\bigl(N^2-\sum_{i}\nid^2\bigr) 
+\ssb\sse\bigl(N-R\bigr)\bigl(N^2-N)\\
&\phe+\ssa\sse\bigl(N-R\bigr)\bigl(N^2-\sum_{i}\nid^2\bigr) 
+\ssb\sse\bigl(N-R\bigr)\bigl(N^2-\sum_{j}\ndj^2\bigr)\\
&\phe+4\ssb\sse N(N-R). 
\end{align*}

Now
\begin{align*}
\e(U_a)\e(U_e) & = (N-R)(\ssb+\sse)\Bigl(
\ssa(N^2-\sum_i\nid^2)
+\ssb(N^2-\sum_j\ndj^2)
+\sse(N^2-N)
\Bigr)
\end{align*}
which contains terms equalling several of those in $\e(U_aU_e)$ above.
Subtracting those term from $\e(U_aU_e)$ yields
\begin{align*}
\cov(U_a,U_e) &= 
2\fpb\Bigl(\sum_i\nid^{-1}\Bigl( \sum_j\zij\ndj\Bigr)^2-\sum_{ij}\nid^{-1}\zij\ndj^2\Bigr)\\
&\phe+\fpb(\kb+2)\sum_{ij}\zij(N-\ndj)\ndj(1-\nid^{-1})\\
&\phe+2\fpe (N-R)+\fpe(\ke+2)(N-R)(N-1)\\
&\phe+4\ssb\sse N(N-R). 
\end{align*}

\section{Covariance of $U_b$ and $U_e$}

By interchanging the roles of the rows and columns in $\cov(U_a,U_e)$, we find that
\begin{align*}
\cov(U_b,U_e) &= 
2\fpa\Bigl(\sum_j\ndj^{-1}\Bigl( \sum_i\zij\nid\Bigr)^2-\sum_{ij}\ndj^{-1}\zij\nid^2\Bigr)\\
&\phe+\fpa(\ka+2)\sum_{ij}\zij(N-\nid)\nid(1-\ndj^{-1})\\
&\phe+2\fpe (N-C)+\fpe(\ke+2)(N-C)(N-1)\\
&\phe+4\ssa\sse N(N-C). 
\end{align*}

\section{Asymptotic approximation: proof of Theorem~\ref{thm:approx}}\label{sec:proof:thm:approx}

We suppose that the following inequalities all hold
\begin{align*}
&\nid\le\epsilon N, && \ndj\le\epsilon N, && R\le \epsilon N, && C\le\epsilon N,\\
& N\le\epsilon \sum_i\nid^2, && N\le\epsilon\sum_j\ndj^2, && \sum_i\nid^2\le\epsilon N^2,\quad\text{and} && \sum_j\ndj^2\le\epsilon N^2 
\end{align*}
for the same small $\epsilon>0$.
The first six inequalities are assumed in the theorem statement.
The last two follow from the first two.
We also assume that
$$ 0<\um \le \ka+2, \kb+2,\ke+2,\fpa,\fpb,\fpe \le \om <\infty.$$
Note that we can bound
$\ssa\ssb$, $\ssa\sse$, and $\ssa\ssb$ away from $0$
and $\infty$ uniformly with those other quantities after
replacing $\um$ by $\min(\um,\um^2)$ and
$\om$ by $\max(\om,\om^2)$.

We also suppose that 
\begin{align}\label{eq:morebounds}
& \sum_{ij}\zij\nid^{-1}\ndj \le \epsilon\sum_i\nid^2,\quad\text{and}
&& \sum_{ij}\zij\nid\ndj^{-1} \le \epsilon\sum_j\ndj^2. 
\end{align}
The bounds in~\eqref{eq:morebounds} seem reasonable but it appears that 
they cannot be derived from the first eight bounds above. 

We begin with the coefficient of $\fpb(\kb+2)$ in $\var(U_a)$ from equation~\eqref{eq:varuathm}. 
It is
\begin{align*}
\sum_{ir}\zzt_{ir}(1-\nid^{-1}-\nrd^{-1}+\nid^{-1}\nrd^{-1})
&=\sum_j\ndj^2-2\sum_{ij}\zij\nid^{-1}\ndj +\sum_{ij}\zij\nid^{-1}\nrd^{-1}\\
&=\sum_j\ndj^2(1+O(\epsilon)).
\end{align*}
The third, fourth and fifth terms in $\var(U_a)$ are all $O(\epsilon)$.
The second term contains
\begin{align*}
\sum_{ir}\nid^{-1}\nrd^{-1}\zzt_{ir}(\zzt_{ir}-1)
&\le\sum_{ir}\nid^{-1}\zzt_{ir}\\
&=\sum_{irj}\nid^{-1}\zij\zrj\\
&=\sum_{ij}\zij\nid^{-1}\ndj\\
&=O(\epsilon).
\end{align*}
It follows that $\var(U_a) = \fpb(\kb+2)\sum_j\ndj^2(1+O(\epsilon))$.
Similarly $\var(U_b) = \fpa(\ka+2)\sum_i\nid^2(1+O(\epsilon))$.

The expression for 
$\var(U_e)$ contains terms
$\fpa(\ka+2)N^2\sum_j\ndj^2 + \fpb(\kb+2)N^2\sum_i\nid^2$.  All other terms
are $O(\epsilon)$ times these two, mostly through $N\ll\sum_i\nid^2,\sum_j\ndj^2\ll N^2$.
The coefficient of $\ssa\ssb$ contains
$$
N\sum_{ij}\zij\nid\ndj \le \epsilon N^2\sum_{ij}\zij\nid =\epsilon N^2\sum_i\nid^2
$$
so it is of smaller order than the lead term, as well as
$$
\sum_i\nid^2\sum_j\ndj^2 \le \epsilon N^2\sum_i\nid^2.
$$
As a result
$$
\var(U_e) = 
\Bigl(\fpa(\ka+2)N^2\sum_j\ndj^2 + \fpb(\kb+2)N^2\sum_i\nid^2\Bigr)(1+O(\epsilon)).
$$

Turning to the covariances
\begin{align*}
\cov(U_a,U_b) 
& = \fpe(\ke+2)\sum_{ij}\zij(1-\nid^{-1}-\ndj^{-1}+\nid^{-1}\ndj^{-1})\\
& = \fpe(\ke+2)(N-R-C+O(R))\\
& = \fpe(\ke+2)N(1+O(\epsilon)).
\end{align*}
Next $\cov(U_a,U_e)$ contains the term
$\fpb(\kb+2)N\sum_{ij}\zij\ndj=\fpb(\kb+2)N\sum_j\ndj^2$.
The terms appearing after that one are $O(N^2) = O(\epsilon N\sum_j\ndj^2)$.
The largest term preceding it is dominated by
$$
\sum_i\nid^{-1}\Bigl(\sum_j\zij\ndj\Bigr)^2
\le\epsilon N
\sum_i\nid^{-1}\Bigl(\sum_j\zij\ndj\Bigr) \Bigl(\sum_j\zij\Bigr)
=\epsilon N\sum_j\ndj^2.
$$
It follows that $\cov(U_a,U_e) = \fpb(\kb+2)N\sum_j\ndj^2(1+O(\epsilon))$
and similarly, $\cov(U_b,U_e) = \fpa(\ka+2)N\sum_i\nid^2(1+O(\epsilon))$.

Next, using~\eqref{eq:fromutoss}
\begin{align*}
\var(\ssah) 
&= \Bigl(\frac{\var(U_e)}{N^4} + \frac{\var(U_a)}{N^2} - 2\frac{\cov(U_a,U_e)}{N^3}\Bigr)(1+O(\epsilon))\\
&= \fpa(\ka+2) \frac1{N^2}\sum_i\nid^2(1+O(\epsilon)),\quad\text{and similarly}\\
\var(\ssbh)&= \fpb(\kb+2) \frac1{N^2}\sum_j\ndj^2(1+O(\epsilon)).
\end{align*}
The last variance is
\begin{align*}
\var(\sseh) & = \Bigl(
\frac{\var(U_a)}{N^2} +\frac{\var(U_b)}{N^2} + \frac{\var(U_e)}{N^4}
-\frac2{N^3}\cov(U_a,U_e) -\frac2{N^3}\cov(U_b,U_e)
+\frac2{N^2}\cov(U_a,U_b)\Bigr)(1+O(\epsilon))\\
&=\fpe(\ke+2)\frac1{N}(1+O(\epsilon)).
\end{align*}

Next we verify that these variance estimates are asymptotically uncorrelated.
Ignoring the $1+O(\epsilon)$ factors we have
\begin{align*}
\cov(\ssah,\ssbh)
& \doteq
\frac1{N^4}\var(U_e) -\frac1{N^3}\cov(U_b,U_e) -\frac1{N^3}\cov(U_a,U_e)+\frac1{N^2}\cov(U_a,U_b)\\
&\doteq
\frac1{N^2}
\Bigl(\fpa(\ka+2)\sum_i\nid^2+\fpb(\kb+2)\sum_j\ndj^2)\Bigr)\\
&\phe-\frac1{N^2}\fpa(\ka+2)\sum_i\nid^2
 -\frac1{N^2}\fpb(\kb+2)\sum_j\ndj^2+\frac1N\fpe(\ke+2)\\
& = \frac1N\fpe(\ke+2)
\end{align*}
which is $O(\epsilon)$ times $\var(\ssah)$ and $\var(\ssbh)$.
Likewise
\begin{align*}
\cov(\ssah,\sseh)
& \doteq\frac1{N^3}\cov(U_a,U_e)+\frac1{N^3}\cov(U_b,U_e) -\frac1{N^4}\var(U_e) -\frac1{N^2}\var(U_a) -\frac1{N^2}\cov(U_a,U_b)
+\frac1{N^3}\cov(U_a,U_e)\\
&\doteq\fpb(\kb+2)\frac2{N^2}\sum_j\ndj^2+\fpa(\ka+2)\frac1{N^2}\sum_i\nid^2
-\Bigl(\fpa(\ka+2)\sum_i\nid^2+\fpb(\kb+2)\sum_j\ndj^2\Bigl)\frac1{N^2}\\
&\phe-\fpb(\kb+2)\frac1{N^2}\sum_j\ndj^2 - \fpe(\ke+2)\frac1{N}\\
& = -\fpe(\ke+2)\frac1N
\end{align*}
which is much smaller than $\var(\ssah)$.
Similarly $\cov(\ssbh,\sseh)\doteq -\fpe(\ke+2)/N$, is much smaller than
$\var(\ssbh)$.

\section{Estimating Kurtoses}\label{sec:estimatingkurtoses}

To estimate the kurtoses $\ka$, $\kb$  and $\ke$ in the above variance expressions, it suffices to 
estimate fourth central moments such as 
$\muaf=\fpa(\ka+3)$ and similarly defined $\mubf$ and $\muef$. 
Given $\ssah$, $\ssbh$, and $\sseh$, we can do this via GMM. Consider the following estimating equations and their expectations,
\begin{align*}
W_a &= \dfrac{1}{2}\sum_{ijj'}\dfrac{1}{\nid}\zij\zijp(\yij-\yijp)^4 \\ 
W_b &= \dfrac{1}{2}\sum_{ii'j}\dfrac{1}{\ndj}\zij\zipj(\yij-\yipj)^4 \\
W_e &= \dfrac{1}{2}\sum_{ii'jj'}\zij\zipjp(\yij-\yipjp)^4 
\end{align*}

Using previous results,
\begin{align*}
\e(W_a) &= \dfrac{1}{2}\sum_{ijj'}\dfrac{1}{\nid}\zij\zijp \e[(\yij-\yijp)^4\bigr) \\
&= \dfrac{1}{2}\sum_{ijj'}\dfrac{1}{\nid}\zij\zijp \e\bigl((\bj-\bjp+\eij-\eijp)^4\bigr) \\
&= \dfrac{1}{2}\sum_{ijj'}\dfrac{1}{\nid}\zij\zijp \e\bigl((\bj-\bjp)^4+6(\bj-\bjp)^2(\eij-\eijp)^2+(\eij-\eijp)^4\bigr) \\
&= \dfrac{1}{2}\sum_{ijj'}\dfrac{1}{\nid}\zij\zijp \bigl(2\mubf+6\fpb+24\ssb\sse+2\muef+6\fpe\bigr)(1-\ojjp) \\
&= \sum_{ijj'}\dfrac{1}{\nid}\zij\zijp \bigl( \mubf+3\fpb+12\ssb\sse+\muef+3\fpe\bigr)(1-\ojjp) \\
&= \sum_{i}\bigl((\kb+2)\fpb+3\fpb+12\ssb\sse+(\ke+2)\fpe+3\fpe\bigr)(\nid-1) \\
&= (N-R)\bigl( \mubf+3\fpb+12\ssb\sse+\muef+3\fpe\bigr). 
\end{align*}
By symmetry,
\begin{align*}
\e(W_b) &= (N-C)\bigl(\muaf+3\fpa+12\ssa\sse+\muef+3\fpe\bigr). 
\end{align*}
Next 
\begin{align*}
\e(W_e) &= \dfrac{1}{2}\sum_{ii'jj'}\zij\zipjp\e\bigl((\yij-\yipjp)^4\bigr) \\
&= \dfrac{1}{2}\sum_{ii'jj'}\zij\zipjp\e\bigl((\ai-\aip+\bj-\bjp+\eij-\eipjp)^4\bigr) \\
&= \dfrac{1}{2}\sum_{ii'jj'}\zij\zipjp\e\bigl((\ai-\aip)^4+6(\ai-\aip)^2(\bj-\bjp)^2+6(\ai-\aip)^2(\eij-\eipjp)^2+(\bj-\bjp)^4\\
&\hspace{28mm}+6(\bj-\bjp)^2(\eij-\eipjp)^2+(\eij-\eipjp)^4\bigr) \\
&= \dfrac{1}{2}\sum_{ii'jj'}\zij\zipjp\bigl((2\muaf+6\fpa)(1-\oiip)+24\ssa\ssb(1-\oiip)(1-\ojjp)+24\ssa\sse(1-\oiip) \\
&\hspace{28mm}+(2\mubf+6\fpb)(1-\ojjp)+24\ssb\sse(1-\ojjp)+(2\muef+6\fpe)(1-\oiip\ojjp)\bigr) \\
&=(\muaf+3\fpa+12\ssa\sse)\sum_{ii'jj'}\zij\zipjp(1-\oiip)+(\mubf+3\fpb+12\ssb\sse)\sum_{ii'jj'}\zij\zipjp(1-\ojjp)\\
&\phe+(\muef+3\fpe)\sum_{ii'jj'}\zij\zipjp(1-\oiip\ojjp)+12\ssa\ssb\sum_{ii'jj'}\zij\zipjp(1-\oiip-\ojjp+\oiip\ojjp) \\
&=(\muaf+3\fpa+12\ssa\sse)(N^2-\sum_{ijj'}\zij\zijp)+(\mubf+3\fpb+12\ssb\sse)(N^2-\sum_{ii'j}\zij\zipj)\\
&\phe+(\muef+3\fpe)N(N-1)+12\ssa\ssb(N^2-\sum_{ijj'}\zij\zijp-\sum_{ii'j}\zij\zipj+N) \\
&=(\muaf+3\fpa+12\ssa\sse)(N^2-\sum_i \nid^2)+(\mubf+3\fpb+12\ssb\sse)(N^2-\sum_j \ndj^2)\\
&\phe+(\muef+3\fpe)N(N-1)+12\ssa\ssb(N^2-\sum_i \nid^2-\sum_j \ndj^2+N). 
\end{align*}

These expectations are all linear in the fourth moments. Therefore, given estimates of $\ssa$, $\ssb$, and $\sse$, we can solve another three-by-three system of equations to get estimates of the fourth moments. 

Letting $M$ be the matrix in equation~\eqref{eq:defm} we find that 
\begin{align*}
\begin{pmatrix}
\e(W_a)\\
\e(W_b)\\
\e(W_e) 
\end{pmatrix}
= M 
\begin{pmatrix}
\muaf\\
\mubf\\
\muef 
\end{pmatrix}
+
\begin{pmatrix}
3(N-R)\fpb+12(N-R)\ssb\sse + 3(N-R)\fpe \\
3(N-C)\fpa+12(N-C)\ssa\sse + 3(N-C)\fpe \\
H 
\end{pmatrix}
\end{align*}
where 
\begin{multline*}
H=(3\fpa+12\ssa\sse)(N^2-\sum_i \nid^2)+(3\fpb+12\ssb\sse)(N^2-\sum_j \ndj^2)\\
+3\fpe N(N-1)+12\ssa\ssb(N^2-\sum_i \nid^2-\sum_j \ndj^2+N). 
\end{multline*}

For plug-in method of moment estimators we replace expected $W$-statistics 
by their sample quantities, replace the variance components by their estimates 
and solve the matrix equation getting $\hat\mu_{A,4}$ et cetera.  
Then $\hat\kappa_A = \hat\mu_{A,4}/\hat\sigma_A^4-3$ and so on.

\section{Best linear predictor}

Here we predict consider linear predicton of $\yij$. We begin
with predictions of the form
$\yijh =\yijh(\lambda) = \sum_{rs}\lrs\zrs\yrs$.
Then we consider predictions of a reduced
form that consider only the totals in row $i$, in row $j$
and in the whole data set.

\subsection{Proof of Lemma~\ref{lem:mseofblp:main}}\label{sec:prooflemmseofblp}

Let $\yijh=\sum_{rs}\zij\lij\yij$ and $L=\e((\yij-\yijh)^2)$. Then 
$$
L = \mu^2\Bigl(1-\sum_{rs}\lambda_{rs}\zrs\Bigr)^2 +\var(\yij) +\var(\yijh)-2\cov(\yij,\yijh). 
$$
First $\var(\yij) =\ssa+\ssb+\sse$.  Next 
\begin{align*}
\cov(\yij,\yijh) 
&= \sum_{rs}\lrs\zrs\bigl(\ssa\oir+\ssb\ojs+\sse\oir\ojs\bigr)\\
&= \ssa\sum_{s}\lis\zis +\ssb\sum_{r}\lrj\zrj+\sse\lij^2\zij, 
\end{align*}
and finally 
\begin{align*}
\var(\yijh) 
&= \sum_{rs}\sum_{r's'} \lrs\lrpsp\zrs\zrpsp\bigl(
\ssa\orrp+\ssb\ossp+\sse\orrp\ossp\bigr)\\
&= 
\ssa\sum_{rss'}\lrs\lrsp\zrs\zrsp 
+\ssb\sum_{rsr'}\lrs\lrps\zrs\zrps 
+\sse\sum_{rs}\lrs^2\zrs. 
\end{align*}
Thus 
\begin{align*}
L &= \mu^2\Bigl(1-\sum_{rs}\lambda_{rs}\zrs\Bigr)^2 +\ssa+\ssb+\sse\\ 
&\phe +
\ssa\sum_{rss'}\lrs\lrsp\zrs\zrsp 
+\ssb\sum_{rsr'}\lrs\lrps\zrs\zrps 
+\sse\sum_{rs}\lrs^2\zrs\\
&\phe -2\Bigl(
\ssa\sum_{s}\lis\zis +\ssb\sum_{r}\lrj\zrj+\sse\lij^2\zij 
\Bigr). 
\end{align*}

Now suppose that we consider the loss
$\wt L = \e( ((\mu+\ai+\bj)-\yijh)^2)$. To do so we replace
$\var(\yijh)$ and $\cov(\yij,\yijh)$ above by
$\var(\ai+\bj)$ and $\cov(\ai+\bj,\yijh)$ respectively, yielding
\begin{align*}
\wt L &= \mu^2\Bigl(1-\sum_{rs}\lambda_{rs}\zrs\Bigr)^2 +\ssa+\ssb\\
&\phe 
+\ssa\sum_{rss'}\lrs\lrsp\zrs\zrsp 
+\ssb\sum_{rsr'}\lrs\lrps\zrs\zrps 
+\sse\sum_{rs}\lrs^2\zrs\\
&\phe -2\Bigl(\ssa\sum_{s}\lis\zis +\ssb\sum_{r}\lrj\zrj\Bigr). 
\end{align*}

\subsection{Stationary conditions}\label{sec:stationary}

The partial derivative of $\wt L$ with respect to $\lrppspp$ is
\begin{align*}
&2\mu^2\Bigl( 1-\sum_{rs}\lrs\zrs\Bigr)(-\zrppspp) + 2\sse\lrppspp\zrppspp\\
&+\ssa\sum_{rss'}\zrs\zrsp(\lrsp1_{rs=r''s''}+\lrs1_{rs'=r''s''})\\
&+\ssb\sum_{rsr'}\zrs\zrps(\lrps1_{rs=r''s''}+\lrs1_{r's=r''s''})\\
&-2\ssa\sum_s\zis1_{is=r''s''}-2\ssb\sum_r\zrj1_{rj=r''s''}.  
\end{align*} 
After taking account of the indicator functions we get
\begin{align*}
\begin{split}
&2\zrppspp\biggl(\mu^2\Bigl( 1-\sum_{rs}\lrs\zrs\Bigr)(-1) + \sse  \lrppspp  
+\ssa\sum_{s'}\zrppsp\lrppsp+\ssb\sum_{r'}\zrpspp\lrpspp  \\
&-\ssa\zispp\oirpp-\ssb\zrppj\ojspp\biggr).  
\end{split}
\end{align*}  
We can replace $\zispp\oirpp$ by $\oirpp$ because of the leading factor $\zrppspp$.
This and a corresponding change to the coefficient of $\ssb$ yield
\begin{align*}
\begin{split}
&2\zrppspp\biggl(\mu^2\Bigl( 1-\sum_{rs}\lrs\zrs\Bigr)(-1) + \sse  \lrppspp  
+\ssa\sum_{s'}\zrppsp\lrppsp+\ssb\sum_{r'}\zrpspp\lrpspp  -\ssa\oirpp-\ssb\ojspp\biggr).  
\end{split}
\end{align*}  
The simplified expression no longer requires the double primes and so we
find that the partial derivative of $\wt L$ with respect to $\lrs$ is
\begin{align*}
&2\zrs\biggl(\mu^2\Bigl(\sum_{r's'}\lrpsp\zrpsp-1\Bigr) + \sse \lrs
+\ssa\sum_{s'}\zrsp\lrsp+\ssb\sum_{r'}\zrps\lrps -\ssa\oir-\ssb\ojs\biggr). 
\end{align*}

\subsection{Proof of Lemma~\ref{lem:shrinkage:main}}\label{sec:prooflemshrinkage}

Here we consider 
$$\yijh = \lambda_0\ydd+\lambda_a\yid+\lambda_b\ydj$$
where
$$\ydd=\sum_{rs}\zrs\yrs,\quad\yid=\sum_s\zis\yis,\quad\text{and}\quad\ydj=\sum_r\zrj\yrj.$$
The mean squared error is $L=\e( (\yij-\yijh)^2)$. Expanding it we get
\begin{align*}
L &= \mu^2\bigl(1-( \lambda_0N+\lambda_a\nid+\lambda_b\ndj)\bigr)^2  
+\var(\yij)+\lambda_0^2\var(\ydd)+\lambda_a^2\var(\yid)+\lambda_b^2\var(\ydj)\\
&\phe-2\lambda_0\cov(\yij,\ydd)-2\lambda_a\cov(\yij,\yid)-2\lambda_b\cov(\yij,\ydj)\\
&\phe+2\lambda_0\lambda_a\cov(\ydd,\yid)+2\lambda_0\lambda_b  
\cov(\ydd,\ydj)+2 \lambda_a\lambda_b\cov(\yid,\ydj).  
\end{align*} 
As before $\var(\yij)=\ssa+\ssb+\sse$.  
We set about finding the other terms.

First 
\begin{align*}
\var(\ydd) &= \ssa\sum_r\nrd^2+\ssb \sum_s\nds^2+\sse N,\\
\var(\yid) &= \ssa\nid^2+\ssb \nid+\sse \nid,\quad\text{and}\\
\var(\ydj) &= \ssa\ndj+\ssb \ndj^2+\sse \ndj. 
\end{align*}

Second
\begin{align*}
\cov(\yij,\ydd) &= \ssa\nid+\ssb \ndj +\sse\zij,\\
\cov(\yij,\yid) &= \ssa\nid+\ssb \zij+\sse \zij,\quad\text{and}\\
\cov(\yij,\ydj) &= \ssa\zij+\ssb \ndj+\sse \zij. 
\end{align*}

The remaining terms use somewhat longer arguments.
\begin{align*}
\cov(\yid,\ydd) 
&= \sum_{rss'}\zrs\zisp\cov(\yrs,\yisp)\\
&= \sum_{rss'}\zrs\zisp\bigl( \oir\ssa + \ossp\ssb +\oir\ossp\sse\bigr)\\
&= \ssa\nid^2 + \ssb\sum_s\zis\nds +\sse\nid,\quad\text{and then}\\
\cov(\ydj,\ydd) &= \ssa\sum_r\zrj\nrd + \ssb\ndj^2+\sse\ndj
\end{align*}
by symmetry.  Finally
\begin{align*}
\cov(\yid,\ydj)  
 & = \sum_{rs}\zis\zrj\cov(\yis,\yrj)\\
 & = \sum_{rs}\zis\zrj\bigl(\ssa\oir + \ssb\ojs + \sse\oir\ojs\bigr)\\
 & = \ssa\sum_{s}\zis\zij + \ssb\sum_r\zij\zrj + \sse\zij\\
 & = \zij\bigl( \ssa\nid + \ssb\ndj+ \sse\bigr).
\end{align*}

Combining these pieces we find that
\begin{align*}
L & =  \mu^2\bigl(1-\lambda_0N-\lambda_a\nid-\lambda_b\ndj\bigr)^2 
+\ssa+\ssb+\sse 
+\lambda_0^2\Bigl(\ssa\sum_r\nrd^2+\ssb \sum_s\nds^2+\sse N\Bigr)\\
&\phe+\lambda_a^2\Bigl(\ssa\nid^2+\ssb \nid+\sse \nid\Bigr) 
+\lambda_b^2\Bigl(\ssa\ndj+\ssb \ndj^2+\sse \ndj\Bigr)\\
&\phe 
-2\lambda_0\Bigl(\ssa\nid+\ssb \ndj +\sse\zij\Bigr) 
-2\lambda_a\Bigl(\ssa\nid+\ssb \zij+\sse \zij\Bigr) 
-2\lambda_b\Bigl(\ssa\zij+\ssb \ndj+\sse \zij\Bigr)\\
&\phe 
+2\lambda_0\lambda_a\Bigl(\ssa\nid^2 + \ssb\sum_s\zis\nds +\sse\nid\Bigr) 
+2\lambda_0\lambda_b \Bigl(\ssa\sum_r\zrj\nrd + \ssb\ndj^2+\sse\ndj\Bigr)\\
&\phe+2 \lambda_a\lambda_b \zij \Bigl( \ssa\nid + \ssb\ndj+ \sse\Bigr). 
\end{align*}

Now suppose we consider instead $\wt L = \e( (\mu+\ai+\bj-\yijh)^2)$.
Then we must replace $\var(\yijh)$ by $\var(\mu+\ai+\bj)=\ssa+\ssb$
and remove the $\sse\zij$ terms from the covariances with $\yij$.
The result is
\begin{align*}
\wt L & =  \mu^2\bigl(1-\lambda_0N-\lambda_a\nid-\lambda_b\ndj\bigr)^2 
+\ssa+\ssb 
+\lambda_0^2\Bigl(\ssa\sum_r\nrd^2+\ssb \sum_s\nds^2+\sse N\Bigr)\\
&\phe+\lambda_a^2\Bigl(\ssa\nid^2+\ssb \nid+\sse \nid\Bigr) 
+\lambda_b^2\Bigl(\ssa\ndj+\ssb \ndj^2+\sse \ndj\Bigr)\\
&\phe 
-2\lambda_0\Bigl(\ssa\nid+\ssb \ndj\Bigr) 
-2\lambda_a\Bigl(\ssa\nid+\ssb \zij\Bigr) 
-2\lambda_b\Bigl(\ssa\zij+\ssb \ndj\Bigr)\\
&\phe 
+2\lambda_0\lambda_a\Bigl(\ssa\nid^2 + \ssb\sum_s\zis\nds +\sse\nid\Bigr) 
+2\lambda_0\lambda_b \Bigl(\ssa\sum_r\zrj\nrd + \ssb\ndj^2+\sse\ndj\Bigr)\\
&\phe+2 \lambda_a\lambda_b \zij \Bigl( \ssa\nid + \ssb\ndj+ \sse\Bigr). 
\end{align*}

\subsection{Proof of Theorem~\ref{thm:shrinkage:main}}\label{sec:proofshrinkage}

From the result of Lemma~\ref{lem:shrinkage:main},
we see that $\wt L$ is quadratic in $\lambda$. Since $\wt L$ is bounded below by $0$, 
it follows that $\wt L$ attains its minimum on $\real^3$, which 
would be any solution of the stationarity condition $\nabla_{\lambda} \wt L=0$. 
We find the components of this gradient. 

\begin{align*}
\frac12\frac\partial{\partial\lambda_0}\wt L 
&=N\mu^2(\lambda_0N+\lambda_a\nid+\lambda_b\ndj-1) 
+\lambda_0\Bigl( \ssa\sum_r\nrd^2+\ssb\sum_s\nds^2+\sse N\Bigr)
-\Bigl(\ssa\nid + \ssb\ndj\Bigr)\\
&\phe+\lambda_a\Bigl(\ssa\nid^2 + \ssb\sum_s\zis\nds+\sse\nid\Bigr) 
+\lambda_b\Bigl(\ssa\sum_r\zrj\nrd+ \ssb\ndj^2+\sse\ndj\Bigr)\\ 
\frac12\frac\partial{\partial\lambda_a}\wt L 
&=\nid\mu^2(\lambda_0N+\lambda_a\nid+\lambda_b\ndj-1) 
+\lambda_a\Bigl(\ssa\nid^2+\ssb\nid+\sse\nid\Bigl)-\Bigl( \ssa\nid+\ssb\zij\Bigr)\\
&\phe+\lambda_0\Bigl(\ssa\nid^2 + \ssb\sum_s\zis\nds+\sse\nid\Bigr) 
+\lambda_b\zij\Bigl(\ssa\nid+\ssb\ndj+\sse\Bigr),\quad\text{and}\\
\frac12\frac\partial{\partial\lambda_b}\wt L 
&=\ndj\mu^2(\lambda_0N+\lambda_a\nid+\lambda_b\ndj-1) 
+\lambda_b\Bigl(\ssa\ndj+\ssb\ndj^2+\sse\ndj\Bigl)-\Bigl( \ssa\zij+\ssb\ndj\Bigr)\\
&\phe
+\lambda_0\Bigl(\ssa\sum_r\zrj\nrd +\ssb\ndj^2+\sse\ndj\Bigr) 
+\lambda_a\zij\Bigl(\ssa\nid+\ssb\ndj+\sse\Bigr). 
\end{align*}

We write this as 
$$
H\lambda^* = c 
$$
where 
$$ 
c = 
\begin{pmatrix} 
N\mu^2+\ssa\nid+\ssb\ndj\\
\nid\mu^2+\ssa\nid+\ssb\zij\\
\ndj\mu^2+\ssa\zij+\ssb\ndj\\
\end{pmatrix}
=
\begin{pmatrix} 
N & \nid&\ndj\\
\nid &\nid&\zij\\
\ndj&\zij&\ndj\\
\end{pmatrix}
\begin{pmatrix}
\mu^2\\\ssa\\\ssb 
\end{pmatrix}
$$ 
and $H$ is a symmetric matrix with upper triangle 
$$
H = 
\begin{pmatrix}
H_{11} & H_{12} & H_{13}\\
* & H_{22} & H_{23}\\
* & * & H_{33}\\
\end{pmatrix}
$$
with elements 
\begin{align*}
H_{11} & = \mu^2N^2 + \ssa\sum_r\nrd^2 + \ssb\sum_s\nds^2+\sse N\\
H_{12} & = \mu^2N\nid + \ssa\nid^2+\ssb\sum_s\zis\nds+\sse\nid\\
H_{13} & = \mu^2N\ndj + \ssa\sum_r\zrj\nrd+\ssb\ndj^2+\sse\ndj\\
H_{22} & = \mu^2\nid^2 + \ssa\nid^2+\ssb\nid+\sse\nid\\
H_{23} & = \mu^2\nid\ndj + \ssa\zij\nid+\ssb\zij\ndj+\sse\zij,\quad\text{and}\\
H_{33} & = \mu^2\ndj^2 + \ssa\ndj+\ssb\ndj^2+\sse\ndj.
\end{align*}
Using $\tid \equiv \sum_s\zis\nds$ and $\tdj \equiv \sum_r\zrj\nrd$ some of these simplify:
\begin{align*}
H_{12}  & = \mu^2N\nid + \ssa\nid^2+\ssb\tid+\sse\nid,\quad\text{and}\\
H_{13}  & = \mu^2N\ndj + \ssa\tdj + \ssb\ndj^2+\sse\ndj.
\end{align*}

\subsection{Proof of Theorem~\ref{thm:newrow:main}}\label{sec:proof:thm:newrow} 

To begin with, we note that $\ndj =\sum_r\zrj\le \sum_r\nrd\zrj\le \epsilon N$.
We write
\begin{align*}
\begin{pmatrix}\lambda_0^* \\ \lambda_b^* \end{pmatrix} &= \dfrac{1}{\det{\tilde{H}}}
\begin{pmatrix}H_{33} & -H_{13} \\ -H_{31} & H_{11}\end{pmatrix}\begin{pmatrix}c_{1}\\ c_{3}\end{pmatrix}.
\end{align*}

Then 
\begin{align*}
\det{\tilde{H}} \lambda_0^* 
&= H_{33}c_{1}-H_{13}c_{3} \\
&=\ndj(\mu^2\ndj + \ssa+\ssb\ndj+\sse) (N\mu^2+\ndj\ssb)\\ 
&\phe-(\mu^2N\ndj + \ssa\sum_r\zrj\nrd+\ssb\ndj^2+\sse\ndj)\ndj(\mu^2+\ssb)\\
& = \mu^2\Bigl(\mu^2N\ndj^2+\ssa N\ndj+\ssb N\ndj^2+\sse N\ndj
-\mu^2N\ndj^2 -\ssa\ndj\sum_r\zrj\nrd-\ssb\ndj^3-\sse\ndj^2 \Bigr)\\
& \phe+ \ssb\Bigl(\mu^2\ndj^3+\ssa \ndj^2+\ssb \ndj^3+\sse \ndj^2
-\mu^2N\ndj^2 -\ssa\ndj\sum_r\zrj\nrd-\ssb\ndj^3-\sse \ndj^2\Bigr)\\
& = \mu^2\Bigl(\ssa N\ndj+\ssb N\ndj^2+\sse N\ndj -\ssa\ndj\sum_r\zrj\nrd-\ssb\ndj^3-\sse\ndj^2\Bigr)\\
& \phe+ \ssb\Bigl(\mu^2\ndj^3+\ssa \ndj^2 
-\mu^2N\ndj^2 -\ssa\ndj\sum_r\zrj\nrd\Bigr)\\
& = \mu^2\Bigl(\ssa N\ndj+\sse N\ndj -\ssa\ndj\sum_r\zrj\nrd-\sse\ndj^2\Bigr)\\
& \phe+ \ssb\Bigl(\ssa \ndj^2 -\ssa\ndj\sum_r\zrj\nrd\Bigr)\\
&=\mu^2(\ssa+\sse)N\ndj(1+O(\epsilon)).
\end{align*}
and
\begin{align*}
\det{\tilde{H}}\lambda_b^* &= H_{11}c_{3}-H_{31}c_{1} \\
&=(\mu^2N^2 + \ssa\sum_r\nrd^2 + \ssb\sum_s\nds^2+\sse N)\ndj(\mu^2+\ssb)\\
&\phe-
(\mu^2N\ndj + \ssa\sum_r\zrj\nrd+\ssb\ndj^2+\sse\ndj)(N\mu^2+\ndj\ssb)\\
& = \mu^2\Bigl( \mu^2 N^2\ndj+\ssa\ndj\sum_r\nrd^2+\ssb\ndj\sum_s\nds^2+\sse N\ndj\\
&\phantom{-\mu^2} -\mu^2 N^2\ndj -\ssa N\sum_r\zrj\nrd -\ssb N\ndj^2-\sse N\ndj\Bigr)\\
&\phe +\ssb\Bigl(\mu^2 N^2\ndj + \ssa \ndj\sum_r\nrd^2 +\ssb \ndj\sum_s\nds^2 + \sse N\ndj\\
&\phantom{-\mu^2} -\mu^2 N\ndj^2 -\ssa \ndj\sum_r\zrj\nrd -\ssb \ndj^3 -\sse\ndj^2\Bigr)\\
& = \mu^2\Bigl( \ssa\ndj\sum_r\nrd^2+\ssb\ndj\sum_s\nds^2 -\ssa N\sum_r\zrj\nrd -\ssb N\ndj^2\Bigr)\\
&\phe +\ssb\Bigl(\mu^2 N^2\ndj + \ssa \ndj\sum_r\nrd^2 +\ssb \ndj\sum_s\nds^2 + \sse N\ndj\\
&\phantom{-\mu^2} -\mu^2 N\ndj^2 -\ssa \ndj\sum_r\zrj\nrd -\ssb \ndj^3 -\sse\ndj^2\Bigr)\\
& = \mu^2 \ssb N^2\ndj(1+O(\epsilon)).
\end{align*}

Thus 
$$
\frac{\lambda_0^*}{\lambda_b^*} = \frac{\ssa+\sse}{\ssb N}(1+O(\epsilon)) .
$$

Next
\begin{align*}
\det \tilde H& = H_{11}H_{33}-H_{13}^2\\
&=
\Bigl(\mu^2N^2 +\ssa\sum_r\nrd^2+\ssb\sum_s\nds^2+\sse N\Bigr)
\Bigl(\mu^2\ndj^2+\ssb\ndj^2+\ssa\ndj+\sse\ndj\Bigr)\\
&\phe-
\Bigl(\mu^2N\ndj+\ssa\sum_r\nrd\zrj+\ssb\ndj^2+\sse\ndj\Bigr)^2\\
&\approx
\mu^2N^2\ndj^2(\mu^2+\ssb)-(\mu^2N\ndj)^2\\
&=\mu^2N^2\ndj^2\ssb.
\end{align*}

As a result the prediction for a new row in a large column is essentially
that column average plus $O(1/\ndj)$ times the global average.

\subsection{Special case $\nid=0$ and $\ndj=1$}

Now suppose that we have no data in the target row and exactly
one older observation in the target column. Let $i'$ be the single
row with $\zipj=1$. There are enough large rows and columns that
the usual conditions $N\ll\sum_i\nid^2\ll N^2$ hold but there are
also some lightly observed rows and columns.
Then
$$
\tilde c = 
\begin{pmatrix}
N & 0&1\\
1&0&1\\
\end{pmatrix}
\begin{pmatrix}
\mu^2\\\ssa\\\ssb 
\end{pmatrix}
=
\begin{pmatrix}
N\mu^2+\ssb\\ \mu^2+\ssb
\end{pmatrix},
$$
and
\begin{align*}
H_{11} & = \mu^2N^2 + \ssa\sum_r\nrd^2 + \ssb\sum_s\nds^2+\sse N\\
H_{13} & = \mu^2N\ndj + \ssa\sum_r\zrj\nrd+\ssb\ndj^2+\sse\ndj\\
         & = \mu^2N + \ssa\nipd+\ssb+\sse,\quad\text{and}\\
H_{33} & = \mu^2\ndj^2 + \ssa\ndj+\ssb\ndj^2+\sse\ndj\\
         & = \mu^2 + \ssa+\ssb+\sse.
\end{align*}

Then 
\begin{align*}
\begin{pmatrix}\lambda_0^* \\ \lambda_b^* \end{pmatrix} &= \dfrac{1}{H_{11}H_{33}-H_{13}^2}
\begin{pmatrix}
H_{33} & -H_{13}\\
-H_{13} & H_{11}\\
\end{pmatrix}
\begin{pmatrix}
N\mu^2+\ssb\\ \mu^2+\ssb
\end{pmatrix}.
\end{align*}

The determinant is
\begin{align*}
&\Bigl( \mu^2N^2 + \ssa\sum_r\nrd^2 + \ssb\sum_s\nds^2+\sse N\Bigr)
\Bigl(\mu^2 + \ssa+\ssb+\sse\Bigr)
-\Bigl(\mu^2N + \ssa\nipd+\ssb+\sse\Bigr)^2\\
&\approx
\mu^2N^2 \Bigl(\mu^2 + \ssa+\ssb+\sse\Bigr) - \mu^4N^2\\
& =
\mu^2N^2(\ssa+\ssb+\sse).
\end{align*}
The numerator for $\lambda_0^*$ is
\begin{align*}
&H_{33}(N\mu^2+\ssb) -H_{13}(\mu^2+\ssb)\\
&\approx(\mu^2+\ssa+\ssb+\sse)N\mu^2 -(\mu^2N)(\mu^2+\ssb)\\
& = (\ssa+\sse)N\mu^2,
\end{align*}
and so
$$
\lambda_0^* \approx \frac1N\frac{\ssa+\sse}{\ssa+\ssb+\sse}.
$$
Similarly, the numerator for $\lambda_b^*$ is
\begin{align*}
&-H_{13}(N\mu^2+\ssb) +H_{11}(\mu^2+\ssb)\\
&\approx-(\mu^2N)N\mu^2 +\mu^2N^2(\mu^2+\ssb)\\
& = \mu^2\ssb N^2
\end{align*}
and so
$$
\lambda_b^* \approx \frac{\ssb}{\ssa+\ssb+\sse}.
$$
In this case, the prediction for $\yij$ is
$$
\frac{\ssb\bydj + (\ssa+\sse)\bydd}{\ssa+\ssb+\sse}.
$$

\section{Asymptotic weights: proof of Theorem~\ref{thm:bignidndj:main}}\label{sec:proof:asywts}

Here we have
\begin{align*}
&1\le\nid\le\epsilon N, 
&&1\le\ndj\le\epsilon N, 
&&\nid\le\epsilon\nid^2, 
\\
&\ndj\le\epsilon\ndj^2, 
&&N\le\epsilon N^2, 
&&\sum_r\nrd^2\le\epsilon N^2, 
\\
&\sum_s\nds^2\le\epsilon N^2, 
&&\sum_r\nrd\zrj\le\epsilon N\ndj,\quad\text{and} 
&&\sum_s\nds\zis\le\epsilon N\nid. 
\end{align*}
The first five follow easily from $1<1/\epsilon \le \nid,\ndj\le\epsilon N$.
The last four follow from the others.  For instance
$\sum_r\nrd^2 \le \sum_r\nrd(\epsilon N) = \epsilon N^2$, and 
$\sum_r\nrd\zrj \le \sum_r\zrj(\epsilon N)=\epsilon\ndj N$.
We also have
$0<m\le \mu^2,\ssa,\ssb,\sse\le M<\infty$.

Then 
$$
H = \begin{pmatrix}
\mu^2N^2 & \mu^2 N\nid & \mu^2 N\ndj\\
\mu^2 N\nid & (\mu^2+\ssa) \nid^2 & \mu^2 \nid\ndj\\
\mu^2 N\ndj & \mu^2\nid\ndj & (\mu^2+\ssb)\ndj^2
\end{pmatrix}(1+O(\epsilon))
$$
and using symbolic computation (via Wolfram$\mid$Alpha, September 6, 2015)
$$
H^{-1} = 
\begin{pmatrix}
\dfrac{\mu^2(\ssa+\ssb) + \ssa\ssb}{\ssa\ssb\mu^2 N^2} & \dfrac{-1}{\ssa \nid N} & \dfrac{-1}{\ssb\ndj N}\\[2ex]
\dfrac{-1}{\ssa \nid N} & \dfrac1{\ssa\nid^2} & 0\\[2ex]
\dfrac{-1}{\ssb\ndj N} & 0 & \dfrac1{\ssb\ndj^2}
\end{pmatrix}(1+O(\epsilon)).
$$
The determinant  of $H^{-1}$ is $(\ssa\ssb\mu^2\nid^2\ndj^2N^2)^{-1}(1+O(\epsilon))$, so we need
$\nid\ge1$ and $\ndj\ge1$ to make matrix inversion a continuous operation.
Similarly
$$
c = \begin{pmatrix}
N\mu^2\\
\nid(\mu^2+\ssa)\\
\ndj(\mu^2+\ssb)
\end{pmatrix}(1+O(\epsilon)).
$$

Thus ignoring the $O(\epsilon)$ terms
\begin{align*}
\lambda_0^*
&\doteq
\Bigl(\dfrac{\mu^2(\ssa+\ssb) + \ssa\ssb}{\ssa\ssb\mu^2 N^2}\Bigr) N\mu^2
-\Bigl(\dfrac{1}{\ssa \nid N}\Bigr) \nid(\mu^2+\ssa)
- \Bigl(\dfrac{1}{\ssb\ndj N}\Bigr) \ndj(\mu^2+\ssb)\\
&=
\frac{\mu^2(\ssa+\ssb) + \ssa\ssb}{\ssa\ssb N}
-\dfrac{\mu^2+\ssa }{\ssa  N} -\dfrac{\mu^2+\ssb }{\ssb  N} \\
&=
\frac{\mu^2(\ssa+\ssb) + \ssa\ssb}{\ssa\ssb N}
-\dfrac{\mu^2\ssb+\ssa\ssb }{\ssa\ssb  N} -\dfrac{\mu^2\ssa+\ssb\ssb }{\ssa\ssb  N} \\
&=-\frac1N.
\end{align*}
The end result $-1/N$ is of the same order of magnitude as the original
terms.  Therefore $\lambda_0^* = (-1/N)(1+O(\epsilon))$.
Similarly
\begin{align*}
\lambda_a^* &\doteq-\dfrac{1}{\ssa\nid N} N\mu^2+\dfrac{1}{\ssa\nid^2}\nid(\mu^2+\ssa)
=-\dfrac{\mu^2}{\ssa\nid} +\dfrac{\mu^2+\ssa}{\ssa\nid}
 = \frac1\nid
\end{align*}
and 
$$
\lambda_b^* \doteq\frac1\ndj,
$$
and both of these approximations involve multiplication by $1+O(\epsilon)$.
In this limit then
$$\yijh = \byid(1+O(\epsilon))+\bydj(1+O(\epsilon))-\bydd(1+O(\epsilon))$$
which make intuitive sense as
$(\hat\mu+\hat{a}_i)+(\hat\mu+\hat{b}_j)-\hat\mu$.


\section{Smoothing predictors}\label{sec:prooflemsmoothing}

In some cases we may want a better estimate of $\e(\yij)$ than $\yij$ itself is.
Such a predictor could take the form
\begin{align}\label{eq:smoothing}
\yijh=\yijh(\lambda)=\lambda_0\sum_{rs}\zrs Y_{rs}+\lambda_a \sum_s \zis Y_{is}+\lambda_b \sum_r \zrj Y_{rj} +\lambda_{ab}\zij\yij.
\end{align} 
It puts either extra or reduced weight on $\yij$ itself, depending on the sign of $\lambda_{ab}$. 
This predictor is only useful when $\zij=1$, so it does not apply in the new row or new column 
cases either. It is only nontrivial when our goal is to 
estimate $\mu+\ai+\bj$, not $\yij$ itself.  
So we only consider $\wt L = \e( (\yijh-\mu-\ai-\bj)^2)$ here. 

\begin{lemma}\label{lem:smoothing:main}
The MSE for the linear predictor~\eqref{eq:smoothing} is 
\begin{align*}
\wt L&= \mu^2\bigl(1-\lambda_0N-\lambda_a\nid-\lambda_b\ndj\bigr)^2 
+\ssa+\ssb 
+\lambda_0^2\Bigl(\ssa\sum_r\nrd^2+\ssb \sum_s\nds^2+\sse N\Bigr)\\
&\phe+\lambda_a^2\Bigl(\ssa\nid^2+\ssb \nid+\sse \nid\Bigr) 
+\lambda_b^2\Bigl(\ssa\ndj+\ssb \ndj^2+\sse \ndj\Bigr)\\
&\phe 
-2\lambda_0\Bigl(\ssa\nid+\ssb \ndj\Bigr) 
-2\lambda_a\Bigl(\ssa\nid+\ssb \zij\Bigr) 
-2\lambda_b\Bigl(\ssa\zij+\ssb \ndj\Bigr)\\
&\phe 
+2\lambda_0\lambda_a\Bigl(\ssa\nid^2 + \ssb\sum_s\zis\nds +\sse\nid\Bigr) 
+2\lambda_0\lambda_b \Bigl(\ssa\sum_r\zrj\nrd + \ssb\ndj^2+\sse\ndj\Bigr)\\
&\phe+2 \lambda_a\lambda_b \zij \Bigl( \ssa\nid + \ssb\ndj+ \sse\Bigr). 
\end{align*}
\end{lemma}
\begin{proof}
This problem only arises when $\zij=1$, which we assume for the rest 
of this section. Then 
\begin{align*}
\e( (\yijt-(\mu+\ai+\bj))^2) & = 
\e( (\yijh+\lambda_{ab}\yij-(\mu+\ai+\bj))^2) \\
& = \wt L + \lambda_{ab}^2\e(\yij^2) 
+2\lambda_{ab}\e(\yij\yijh) 
-2\lambda_{ab}\e(\yij(\mu+\ai+\bj)). 
\end{align*}
Now $\e(\yij(\mu+\ai+\bj)) =\mu^2+\ssa+\ssb$ and 
$\e(\yij\yijh)=\mu^2(N\lambda_0+\nid\lambda_a+\ndj\lambda_b)+\cov(\yij,\yijh)$
for 
\begin{align*}
\cov(\yij,\yijh) & =
\cov(\yij,\lambda_{0}\ydd+\lambda_a\yid+\lambda_b\ydj)\\
&=
  \lambda_0\bigl( \ssa\nid+\ssb\ndj+\sse\zij\bigr) 
+\lambda_a\bigl( \ssa\nid+\ssb\zij+\sse\zij\bigr) 
+\lambda_b\bigl( \ssa\zij+\ssb\ndj+\sse\zij\bigr)\\
&=
  \lambda_0\bigl( \ssa\nid+\ssb\ndj+\sse\bigr) 
+\lambda_a\bigl( \ssa\nid+\ssb\zij+\sse\bigr) 
+\lambda_b\bigl( \ssa\zij+\ssb\ndj+\sse\bigr),
\end{align*}
since we assume that $\zij=1$. 
Therefore $\e( (\yijt-\mu-\ai-\bj)^2 )$ equals 
\begin{align*}
& \mu^2\bigl(1-\lambda_0N-\lambda_a\nid-\lambda_b\ndj\bigr)^2 
+\ssa+\ssb 
+\lambda_0^2\Bigl(\ssa\sum_r\nrd^2+\ssb \sum_s\nds^2+\sse N\Bigr)\\
&\phe+\lambda_a^2\Bigl(\ssa\nid^2+\ssb \nid+\sse \nid\Bigr) 
+\lambda_b^2\Bigl(\ssa\ndj+\ssb \ndj^2+\sse \ndj\Bigr)\\
&\phe 
-2\lambda_0\Bigl(\ssa\nid+\ssb \ndj\Bigr) 
-2\lambda_a\Bigl(\ssa\nid+\ssb \Bigr) 
-2\lambda_b\Bigl(\ssa+\ssb \ndj\Bigr)\\
&\phe 
+2\lambda_0\lambda_a\Bigl(\ssa\nid^2 + \ssb\sum_s\zis\nds +\sse\nid\Bigr) 
+2\lambda_0\lambda_b \Bigl(\ssa\sum_r\zrj\nrd + \ssb\ndj^2+\sse\ndj\Bigr)\\
&\phe+2 \lambda_a\lambda_b \Bigl( \ssa\nid + \ssb\ndj+ \sse\Bigr)  
+ \lambda_{ab}^2(\mu^2+\ssa+\ssb+\sse)\\
&\phe+2\lambda_{ab}\Bigl(
  \lambda_0\bigl( \ssa\nid+\ssb\ndj+\sse\bigr)  
+\lambda_a\bigl( \ssa\nid+\ssb+\sse\bigr)  
+\lambda_b\bigl( \ssa+\ssb\ndj+\sse\bigr)\Bigr)\\
&\phe -2\lambda_{ab}(\mu^2+\ssa+\ssb) + 
2\mu^2\lambda_{ab}(N\lambda_0+\nid\lambda_a+\ndj\lambda_b).  
\end{align*} 
Gathering up the coefficient of $\mu^2$ we get 
\begin{align*}
& \mu^2\bigl(1-\lambda_0N-\lambda_a\nid-\lambda_b\ndj-\lambda_{ab}\bigr)^2 
+\ssa+\ssb 
+\lambda_0^2\Bigl(\ssa\sum_r\nrd^2+\ssb \sum_s\nds^2+\sse N\Bigr)\\
&\phe+\lambda_a^2\Bigl(\ssa\nid^2+\ssb \nid+\sse \nid\Bigr) 
+\lambda_b^2\Bigl(\ssa\ndj+\ssb \ndj^2+\sse \ndj\Bigr)\\
&\phe 
-2\lambda_0\Bigl(\ssa\nid+\ssb \ndj\Bigr) 
-2\lambda_a\Bigl(\ssa\nid+\ssb \Bigr) 
-2\lambda_b\Bigl(\ssa+\ssb \ndj\Bigr)\\
&\phe 
+2\lambda_0\lambda_a\Bigl(\ssa\nid^2 + \ssb\sum_s\zis\nds +\sse\nid\Bigr) 
+2\lambda_0\lambda_b \Bigl(\ssa\sum_r\zrj\nrd + \ssb\ndj^2+\sse\ndj\Bigr)\\
&\phe+2 \lambda_a\lambda_b \Bigl( \ssa\nid + \ssb\ndj+ \sse\Bigr) 
+ \lambda_{ab}^2(\ssa+\ssb+\sse)\\
&\phe+2\lambda_{ab}\Bigl(
  \lambda_0\bigl( \ssa\nid+\ssb\ndj+\sse\bigr) 
+\lambda_a\bigl( \ssa\nid+\ssb+\sse\bigr) 
+\lambda_b\bigl( \ssa+\ssb\ndj+\sse\bigr)\Bigr)\\
&\phe -2\lambda_{ab}(\ssa+\ssb). 
\end{align*}

Half of the derivative of this squared error with respect to $\lambda_{ab}$ is 
\begin{align*}
&\lambda_{ab}(\mu^2+\ssa+\ssb+\sse)\\
&\phe 
+\lambda_0\bigl( \ssa\nid+\ssb\ndj+\sse\bigr) 
+\lambda_a\bigl( \ssa\nid+\ssb+\sse\bigr) 
+\lambda_b\bigl( \ssa+\ssb\ndj+\sse\bigr)\\
&\phe-\mu^2-\ssa-\ssb+\mu^2(N\lambda_0+\nid\lambda_a+\ndj\lambda_b). 
\end{align*}
We see that given the other $\lambda$ choices, this derivative is 
decreasing at $0$ (hence we favor positive self-weight) if 
$$
\mu^2+\ssa+\ssb >
\lambda_0\bigl(N\mu^2+ \ssa\nid+\ssb\ndj+\sse\bigr) 
+\lambda_a\bigl( \nid\mu^2+\ssa\nid+\ssb+\sse\bigr) 
+\lambda_b\bigl( \ndj\mu^2+\ssa+\ssb\ndj+\sse\bigr). 
$$
Furthermore, the optimal self-weight, given the other $\lambda$'s is 
\begin{align*}
&\frac1{\mu^2+\ssa+\ssb+\sse}\times\Bigl(
\mu^2+\ssa+\ssb 
-\lambda_0\bigl( N\mu^2+\ssa\nid+\ssb\ndj+\sse\bigr) \\
&\phe-\lambda_a\bigl( \nid\mu^2+\ssa\nid+\ssb+\sse\bigr) 
-\lambda_b\bigl( \ndj\mu^2+\ssa+\ssb\ndj+\sse\bigr) 
\Bigr). 
\end{align*}
\end{proof}

The point of this predictor is that we might expect another observation 
to be made later in row $i$ and column $j$.  Then estimating 
$\mu+a_i+b_j$ is a better way to predict than repeating the earlier $\yij$. 
To use Lemma~\ref{lem:smoothing:main} after a second pass, one can compute 
$\wt L$ as the given quadratic function in the four variables $\lambda_0$,
$\lambda_a$, $\lambda_b$ and $\lambda_{ab}$. The minimizer of that quadratic 
gives weights to apply in prediction. When $\sse$ is very small then $\yij$ is 
already close to $\mu+\ai+\bj$ and placing special weight on $\yij$ will 
be advantageous.

\end{document}